\documentclass[12pt,a4paper,notitlepage]{article}
\usepackage[utf8]{inputenc}
\usepackage[margin=1in]{geometry}
\usepackage{adjustbox}
\usepackage{placeins}
\usepackage{booktabs}
\usepackage[onehalfspacing]{setspace}
\usepackage{amsmath,bm}
\usepackage{amsfonts}
\usepackage{amssymb}
\usepackage{natbib}
\usepackage{array}
\usepackage{url}
\usepackage{makecell}
\usepackage{hyperref}   
\usepackage{overpic}
\usepackage{xurl}

\usepackage{pdflscape}
\usepackage{epigraph}
\usepackage{caption}
\usepackage{subcaption}
\usepackage{comment}
\usepackage{longtable}

\usepackage{tikz}
\usetikzlibrary{positioning}
\usetikzlibrary{snakes}
\usetikzlibrary{calc}
\usetikzlibrary{arrows}
\usetikzlibrary{decorations.markings}
\usetikzlibrary{shapes.misc}
\usetikzlibrary{matrix,shapes,arrows,fit,tikzmark}
\definecolor{red}{RGB}{213,94,0}
\def\sym#1{\ifmmode^{#1}\else\(^{#1}\)\fi}

\newcommand{\binarybtwnvar}{0.034}
\newcommand{\binarybtwnr}{19}

\newcommand{\industrybtwnr}{47}

\newcommand{\sqbtwnvar}{0.047}
\newcommand{\sqbtwnr}{37}

\newcommand{\indsqbtwnr}{56}    

\newcommand{\binarydr}{3.9}
\newcommand{\binarytau}{0.21}
\newcommand{\binarylambdaone}{0.46}
\newcommand{\binarydrlambdaone}{0.27}

\newcommand{\industrydr}{5.2}
\newcommand{\industrytau}{0.32}
\newcommand{\industrylambdaone}{0.51}
\newcommand{\industrydrlambdaone}{0.24}

\newcommand{\sqdr}{4.7}
\newcommand{\industrysqdr}{3.1}

\newcommand{\namedrsq}{5.4}

\usepackage{amsthm}
\theoremstyle{definition}
\newtheorem{example}{Example}
\newtheorem{proposition}{Proposition}

\begin{document}

\title{A Discrimination Report Card}
\author{Patrick Kline, Evan K. Rose, and Christopher R. Walters\thanks{Kline: UC Berkeley and NBER, pkline@econ.berkeley.edu; Rose: University of Chicago and NBER, ekrose@uchicago.edu; Walters: UC Berkeley and NBER, crwalters@econ.berkeley.edu. We thank Ben Scuderi for helpful feedback on an early draft of this paper and Hadar Avivi and Luca Adorni for outstanding research assistance. Seminar participants at Brown University, the 2022 California Econometrics Conference, Columbia University, CIREQ 2022 Montreal, Microsoft Research, Monash University, Peking University, Royal Holloway, UC Santa Barbara, UC Berkeley, The University of Virginia, and The University of Chicago Interactions Conference provided useful comments. Routines for implementing the ranking procedures developed in this paper are available online at \url{https://github.com/ekrose/drrank}.}}

\maketitle
\begin{abstract}
\begin{singlespace}
We develop an Empirical Bayes grading scheme that balances the informativeness of the assigned grades against the expected frequency of ranking errors. Applying the method to a massive correspondence experiment, we grade the racial biases of 97 U.S. employers. A four-grade ranking limits the chances that a randomly selected pair of firms is mis-ranked to 5\% while explaining nearly half of the variation in firms’ racial contact gaps. The grades are presented alongside measures of uncertainty about each firm's contact gap in an accessible rubric that is easily adapted to other settings where ranks and levels are of simultaneous interest.
\end{singlespace}
\end{abstract}

\begin{center}
\end{center}

\newpage

\section{Introduction} \label{sec:introduction}

\epigraph{Sunlight is said to be the best of disinfectants; electric light the most efficient policeman.}
			{\textit{Louis Brandeis } }

Though half a century has passed since the enactment of the Civil Right Act, there is ample evidence that American employers still discriminate based upon legally protected characteristics including sex and race \citep{charles_guryan_jel,bertrand_duflo_review, quillian2017meta, quillian2020evidence}. One reason for the stubborn persistence of employer discrimination is the paucity of reliable information regarding the tendencies of specific organizations to engage in such illicit behavior. Job seekers cannot direct their search effort away from biased firms if the identities of these firms are unknown. Corporate executives may have little incentive to search for more equitable recruiting practices if they cannot benchmark their organizations' biases to those of their peers.

This paper constructs a \emph{discrimination report card} that summarizes objective information regarding the relative biases of a broad collection of Fortune 500 companies. Our analysis leverages a massive correspondence experiment sending up to 1,000 job applications to each of 108 firms. In a previous analysis of this experiment \citep{kline2021systemic}, we established that applications listing distinctively Black names received seven percent fewer employer contacts than those listing distinctively white names. This contact penalty displays striking heterogeneity across companies: applying Empirical Bayes (EB) deconvolution methods, we found that the top 20\% of discriminating firms were responsible for nearly half of the employer contacts lost to racial discrimination.

While this earlier work established that discrimination is highly concentrated among a small subset of employers, the experimental estimate for any particular firm is subject to significant sampling uncertainty. This statistical noise presents an obstacle to reliable estimation of discrimination by specific firms. Though scientific communication is generally aided by transparency \citep{andrews2021model}, lay audiences can find it challenging to translate point estimates and standard errors into conclusions of interest. Experimental evidence from psychology suggests audiences frequently depart from Bayesian reasoning, which may generate over- or under-reactions to statistical evidence \citep{mullainathan2002thinking,mullainathan2008coarse,bordalo2016stereotypes}. These considerations suggest that a low-dimensional, easily-digestible summary of the experimental findings is important for effectively communicating the results to a broad audience. Likely for similar reasons, scholars, policymakers, and private businesses increasingly report simple ``report cards" summarizing estimates of quality for various types of institutions, including colleges \citep{chetty2017college}, K-12 schools \citep{bergman2020housing,angrist2021gs}, teachers \citep{bergmanhill2018teachers,pope2019teachers}, healthcare providers \citep{randhealthcarereportcards2002,pope2009hospitals,kolstad2013}, and neighborhoods \citep{chettyhendren2018part2,chetty2018atlas}.

This paper develops new methods for grading firms (or other units) based upon noisy measures of their conduct while maintaining statistical guarantees on the reliability of the resulting grades. The information content of our classification scheme is quantified by Kendall's $\tau$ measure of correlation \citep{kendall1938new} between our proposed ranking and the true ranking of firms' latent discriminatory behavior. The reliability of the report card grades is quantified by an analogue of the False Discovery Rate \citep{benjamini1995controlling,storey2002direct} that we term the Discordance Rate (DR). The DR between two grades gives the expected probability that a firm assigned the worse grade discriminates less than a firm assigned the better grade. We also develop an extension of DR that weights mistakes by their cardinal magnitudes, which captures the idea that large mistakes are more costly than small ones.

We show that the tradeoff between these notions of information and reliability arises naturally from a series of lotteries over firm pairs where an analyst guesses which member of each pair exhibits worse conduct. When facing multiple gambles of this form, the analyst faces an optimization problem subject to logical transitivity constraints requiring all pairwise comparisons to be consistent with a coherent underlying ranking. A parameter $\lambda$ trades off the gains of correctly ranking pairs of firms against the costs of misclassifying them. When $\lambda=1$, it is optimal to assign every firm a unique grade that maximizes the expected rank correlation with the true latent discrimination levels. These maximally-informative grades turn out to be closely connected to classic proposals for preference aggregation via pairwise elections found in the social choice literature \citep{borda1784memoire,condorcet1785essay,young1978consistent,young1986optimal}, with the posterior probability that one firm is more biased than another serving as a ``vote share.'' When $\lambda<1$ it is only optimal to strictly rank firm pairs that can be distinguished with sufficiently high posterior probability, potentially yielding ties and therefore a low number of distinct ``grades." These coarse grades protect against misinterpretation at the cost of losing information, thereby reducing correlation with the true ranks.

The grades generated by our procedure yield a simple classification of firms into groups that facilitates pairwise comparisons of firm conduct. The grades allow an assessment not only of which firms are discriminating most, but also which firms are discriminating \emph{least}. The latter sort of information may prove especially useful in the hunt for best practices that lead to inclusive hiring. For example, our past work documented smaller contact gaps at firms that contacted applicants from fewer phone numbers, a proxy for centralized human resources (HR) practices. \cite{berson2020outsourcing} provide corroborating evidence that contact gaps are lower among French firms with centralized recruiting practices, while \citet{challeeffect} report encouraging experimental evidence on the potential for HR reforms to foster inclusive hiring at large firms. Our report card results suggest directions for work comparing these and other policies across firms graded as high and low performers.

To supplement the ordinal information conveyed by our grades, we report Empirical Bayes posterior means and credible intervals summarizing the \emph{level} of discriminatory conduct at each employer. This ordinal and cardinal information is visualized in a single report card that provides a powerful and easily-accessible rubric for assessing absolute and relative discriminatory conduct. Although this rubric was designed to communicate information about discrimination, we expect the methods proposed here to prove useful in other settings where noisy measures of institutional performance and conduct have been studied, including assessments of school value added \citep{angristqje2017}, hospital quality \citep{chandra2016health}, health insurance plans \citep{abaluckhealthinsurance2021}, and regional intergenerational mobility \citep{chetty2018impacts2}. Routines for implementing our ranking procedure are available online at \url{https://github.com/ekrose/drrank}.

As an introductory illustration of the method, we rank the contact rates of the first names used in our correspondence experiment. A non-parametric deconvolution suggests that name-specific contact rates cluster around two distinct values capturing mean contact rates for distinctively white and Black names. Weighing the loss from incorrectly ordering a pair of names four times as heavily as the gain from correctly ordering them, our ranking procedure stratifies the names into two groups with distinct grades. These grades are shown to be strongly predictive of a name's nominal race but not its sex. Allowing additional grades has little impact on these correlations, suggesting that our ranking procedure is suitable for recovering ``missing labels'' with a low-dimensional structure.

Proceeding to our primary application of ranking firm biases against Black applicants, we compute optimal grades subject to the same preferences over correct and incorrect rankings used for first name pairs. In a single pairwise gamble, these preferences (which correspond to a particular choice of $\lambda$) require at least 80\% posterior confidence to justify a strict ordering of firms. In our baseline specification, applying this choice of $\lambda$ to generate a transitive ordering over all firms yields three unique grade levels. These grades capture roughly \binarybtwnr\% of the between-firm variation in proportional contact penalties and yield an expected rank correlation with the true penalties of \binarytau. Although our grading system reflects only ordinal considerations, we estimate that the average white/Black gap in contact rates among firms awarded the worst grade is 24\%, while the gap among firms awarded the best grade is only 3\%.

Our earlier work found that industry explains roughly half of the variation in racial discrimination levels across firms \citep{kline2021systemic}. Motivated by this finding, we extend our procedure to build industry information into the report card grades. This extension is achieved by augmenting \cite{efron2016empirical}'s log-spline deconvolution approach to flexibly estimate separate distributions of discrimination within and between industries. Consistent with our past work, we find that industry affiliation accounts for roughly half of the cross-firm variation in proportional contact penalties. Incorporating industry affiliation into the ranking procedure with the same choice of $\lambda$ yields four grades. These improved grades explain \industrybtwnr\% of the variation in contact penalties across firms and yield a correlation with the latent ranks of \industrytau, while limiting the expected share of firm pairs that are mis-ranked to \industrydr\%. 

These more informative industry-dependent grades constitute our preferred ranking of the companies in our experiment.  Firms assigned the worst grade in this ranking contact white applicants 22\% more often than Black applicants, similar to the lowest category in the ranking without industry effects. However, 5 firms receive this label in the model with industry effects compared to only 2 in the baseline model, an indication of the extra information conveyed by industry. Similar to the specification without industry, firms receiving the best grade in the industry effects model (11 firms) exhibit very small racial biases. To the extent that these differences are driven by HR practices or other firm policies, there may be opportunities for the substantial set of firms that scored poorly to improve their behavior by imitating the practices of those who scored more highly. 

Our work extends a burgeoning literature on EB ranking methods. A large empirical literature ranks teachers, schools, hospitals, and neighborhoods using James-Stein style shrinkage rules \citep[e.g.,][]{chetty2014measuring,chetty2018opportunity}. \cite{portnoy1982maximizing} established conditions under which ranking based on such rules maximizes the probability of a correct ordering, while \cite{laird1989empirical} proposed directly computing posterior mean ranks under a normality assumption on the latent heterogeneity. 
Both sorts of ranks may be noisy, however, leading to a proliferation of ranking mistakes when the number of units grows large. A recent econometrics literature confirms that this problem can become severe in practice and proposes approaches to testing hypotheses regarding either ranks themselves or the levels of highly-ranked units \citep{andrews2019winners,mogstad2020rankings}. 

Building on the analogy with multiple testing, \cite{gu2020invidious} consider the use of non-parametric EB methods to select tail performers subject to constraints on the False Discovery Rate, which limits the number of ordering mistakes expected when selecting top performers. Our proposal generalizes the approach in \cite{gu2020invidious} by accommodating more than two grades and avoids the requirement to treat one of the grades as a null hypothesis. More recent work by \cite{gu2022ranking} considers a ranking of journals based on pairwise citation counts using a penalized Bradley-Terry model \citep{bradley1952rank}. While our proposed approach shares \cite{gu2022ranking}'s focus on pairwise differences, the method does not require pairwise data on tournaments and allows users to trade off transparent notions of the information content and reliability of the resulting grades. 

After undergoing peer review, we plan to release the names of the firms in this study to the public. Regulatory agencies such as the Equal Employment Opportunity Commission (EEOC) and the Office of Federal Contract Compliance (OFCCP) have broad discretion to launch investigations into possible violations of equal employment opportunity laws, especially violations by federal contractors. Many of the firms receiving poor grades turn out to be federal contractors, suggesting this information may be of help in targeting future compliance efforts. However, compliance efforts are inevitably long and costly and many firms remain out of compliance even after having been fined \citep{maxwell2013using}. 

As the introductory quote by Brandeis suggests, shining some statistical light on the problem of discrimination may have a more immediately salutary effect than regulatory enforcement efforts. Little scientific information about the discriminatory conduct of particular firms is available to the public. The most powerful ``disinfectant'' may well be the decentralized reactions of employees, customers, and leaders of these organizations to the provision of such information.

\section{The experiment}

We construct discrimination report cards based on the resume correspondence experiment analyzed in \cite{kline2021systemic}. The experiment's sampling frame began with the 2018 list of companies in the Fortune 500. Attention was then restricted to firms with sufficient geographic variation and entry-level job posting for our experiment to be feasible. Over the course of the study, 125 entry-level job vacancies were sampled from each of these employers, with each vacancy corresponding to an establishment in a different U.S. county. This restriction was intended to ensure nation-wide coverage of each firm's recruitment conduct and to minimize the chances that multiple sampled job vacancies were being managed by the same individual.

Sampling was organized in a series of 5 waves, with a target of 25 jobs sampled for each firm in each wave. The majority of firms (72) were sampled in all waves; the rest were excluded in some waves due to COVID-19 and technological interruptions. We attempted to send each sampled job four pairs of applications, with each pair consisting of one Black applicant and one white applicant. Some vacancies received fewer than 8 total applications because the job opening closed while applications were still in progress. Our final sample included roughly 84,000 applications.

To signal race and gender, we followed previous correspondence experiments and used distinctive names. Our set of names started with that of \cite{bertrand_mullainathan_2004}, who used 9 unique names for each race and gender group. This list was supplemented with 10 additional names per group from a database of speeding tickets issued in North Carolina between 2006 and 2018. We classified a name as racially distinctive if more than 90\% of individuals with that name are of a particular race, and selected the most common distinctive Black and white names for those born between 1974 and 1979. Distinctive last names were taken from the 2010 U.S. Census. We selected names with high race-specific shares among those that occur at least 10,000 times nationally. 

One application within each pair was randomly assigned a distinctively white name while the other was randomly assigned a distinctively Black name. Fifty-percent of names were distinctively female and the rest distinctively male, but assignment of sex was not stratified. Each fictitious applicant was independently randomly assigned a large set of additional characteristics, including educational and previous employment histories. 

Our primary outcome is whether an employer attempted to contact the fictitious applicant. Phone numbers and e-mail addresses assigned to the fictitious applicants were monitored to determine when employers reached out for an interview. Contact information was assigned to ensure that no two applicants to the same firm shared an e-mail address or phone number. Our analysis focuses on whether the employer attempted to contact an applicant by any method within 30 days of applying. Further details on the experimental design are available in \cite{kline2021systemic}.

\section{Decision problem}\label{sec:decision}

Consider the problem of ranking a collection of $n$ firms, indexed by $i\in\{1,\dots,n\}\equiv[n]$, according to their values of a scalar measure of discrimination $\theta_{i}\in \mathbb{R}$. The decision variable $d_i\in[n]$ gives the \emph{grade} assigned to firm $i$. Larger values of $d_i$ indicate a firm is more biased. Hence, when $d_i>d_j$ for two firms $i$ and $j$, we say that firm $i$ received a ``worse'' grade than firm $j$.

For each firm $i\in[n]$, we have the measurements $Y_{i}=(\hat \theta_i, s_i)$, where $\hat \theta_i$ is a consistent estimate of $\theta_i$ and $s_i$ is that estimate's asymptotic standard error. We assume the $\{\hat \theta_i\}_{i=1}^n$ are mutually independent and that $\hat \theta_i \sim \mathcal{N}(\theta_i, s_i^2)$. The normality assumption can be justified by an asymptotic approximation as each $\hat{\theta}_{i}$ is based on the contact rate of a large number of underlying applications.

It is convenient to recast the problem of ranking $n$ firms as that of ranking all $\binom{n}{2}$ pairs of firms subject to a set of transitivity constraints. Correctly ranking the bias of a pair of firms yields a \emph{concordance} while ranking the pair incorrectly yields a \emph{discordance}. A pair can also be deemed a tie, which yields neither a discordance nor a concordance. 

\subsection{Gambling over ranks}\label{sec:gambling}

To build intuition it is helpful to first consider the problem of deciding on the rank of a single pair of firms $i$ and $j$ based upon realizations $(y_i,y_j)$ of independent signals $(Y_i,Y_j)$. Suppose that correctly ranking the pair yields payoff $\lambda\in[0,1]$ while reversing their true rank yields payoff -1. We can also declare the comparison a draw by assigning the firms equal ranks, which amounts to abstaining from the gamble and yields certain payoff 0. 

In considering this lottery, we view the pair $(\theta_i,\theta_j)$ of firm types as $i.i.d.$ draws from a continuous prior distribution $G:\mathbb{R}\rightarrow[0,1]$. The posterior probability that firm $i$ is more biased than firm $j$ will be denoted $\pi_{ij}=\Pr\left(\theta_{i}>\theta_{j}|Y_i=y_{i},Y_j=y_{j}\right)$. Because $G$ is continuous, ties are measure zero and $\pi_{ij}=1-\pi_{ji}$. 

The expected utility of assigning grades $d=(d_1,d_2)\in\{1,2\}^2$ to these firms can therefore be written
\begin{align*}
    EU(\pi_{ij},d)=&[\lambda\pi_{ij}-\pi_{ji}]\cdot1\{d_{i}>d_{j}\}+[\lambda\pi_{ji}-\pi_{ij}]\cdot1\{d_{i}<d_{j}\}\\
                    =&[(1+\lambda)\pi_{ij}-1]\cdot1\{d_{i}>d_{j}\}+[(1+\lambda)(1-\pi_{ji})-1]\cdot1\{d_{i}<d_{j}\}.
\end{align*} 
The optimal policy is a simple posterior threshold rule:
    \smallskip{}
    \begin{itemize}
        \item  Set $(d_i=2,d_j=1)$ iff $\pi_{ij}>\frac{1}{1+\lambda}$. 
        \smallskip{}
        \item Set $(d_i=1,d_j=2)$ iff $\pi_{ji}>\frac{1}{1+\lambda}$.
        \smallskip{}
        \item Otherwise, set $d_i=d_j$.
    \end{itemize}
\bigskip{}

When $\lambda=1$, it is optimal to follow a maximum a posteriori (MAP) rule, assigning the higher rank to whichever firm has a greater probability of having the largest value of $\theta$. But when $\lambda<1$, it is better to assign pairs of firms with $\pi_{ij}$ near 1/2 equal grades rather than risk ranking them incorrectly. The quantity $(1-\lambda)$ can therefore be thought of as measuring discordance aversion.

\subsection{Compound Loss}

Now consider the case where we can gamble on the relative rank of all $\binom{n}{2}$ pairs of firms. \cite{kendall1938new}'s classic $\tau$ measure of rank correlation can be defined as the share of pairs yielding a concordance minus the share yielding a discordance. The loss function we propose is a generalization of $\tau$ indexed by a scalar $\lambda\in[0,1]$ that controls the benefit of a concordance relative to the cost of a discordance. 

Letting $\theta=(\theta_1,\dots,\theta_n)'$ denote the vector of latent biases and $d=(d_1,\dots,d_n)'$ a vector of assigned grades, our loss function can be written:
\begin{flalign}\label{eq:kendall}
L\left(d,\theta;\lambda\right)=\binom{n}{2}^{-1}\sum_{i=2}^{n}\sum_{j=1}^{i}
\Bigg[\underbrace{1\left\{ \theta_{i}>\theta_{j},d_{i}<d_{j}\right\} +1\left\{ \theta_{i}<\theta_{j},d_{i}>d_{j}\right\} }_{\text{discordant pairs}}-\\ \nonumber
\lambda\bigg(\underbrace{1\left\{ \theta_{i}<\theta_{j},d_{i}<d_{j}\right\} +1\left\{ \theta_{i}>\theta_{j},d_{i}>d_{j}\right\} }_{\text{concordant pairs}}\bigg) \Bigg].
\end{flalign}
While every discordant pair yields a loss of 1, every concordant pair reduces loss by $\lambda$. When $\lambda=1$ the loss function equals minus one times Kendall's $\tau$ measure of rank correlation between $d$ and $\theta$, which we will denote $\tau(d,\theta)$. When $\lambda<1$, ranking mistakes are more costly than forgone concordances, which creates an incentive to deem the pair a tie. 

Building on the insight that $\tau(d,\theta)=-L\left(d,\theta;1\right)$, we can also write the loss function:
\begin{flalign}
L\left(d,\theta;\lambda\right)=&\left(1-\lambda\right)\binom{n}{2}^{-1}\sum_{i=2}^{n}\sum_{j=1}^{i}\Bigg[1\left\{ \theta_{i}>\theta_{j},d_{i}<d_{j}\right\} +1\left\{ \theta_{i}<\theta_{j},d_{i}>d_{j}\right\}\Bigg]-\lambda\tau(d,\theta)\nonumber\\=&\left(1-\lambda\right)DP(d,\theta)-\lambda \tau(d,\theta),\nonumber
\end{flalign}
where the quantity $DP(d,\theta)=\binom{n}{2}^{-1}\sum_{i=2}^{n}\sum_{j=1}^{i}[1\left\{ \theta_{i}>\theta_{j},d_{i}<d_{j}\right\} +1\left\{ \theta_{i}<\theta_{j},d_{i}>d_{j}\right\}]$ is the \emph{Discordance Proportion}. The Discordance Proportion gives the share of firm pairs that are misranked according to their grades. Interpreting our decision problem as a series of tests of the null hypothesis that each member of a pair discriminates equally, the Discordance Proportion may be seen as a directional (sometimes called type III) error rate. That is, we reject equality in favor of an erroneous alternative, yielding a discordance. This representation clarifies that the parameter $\lambda$ trades off the desire to accurately communicate information to the audience by maximizing $\tau(d,\theta)$ against concern about misclassifying the ranks of firms, as reflected by $DP(d,\theta)$.

\subsection{Risk function}\label{sec:risk}
While one would ideally like to choose grades $d$ that balance the rank correlation $\tau(d,\theta)$ against the Discordance Proportion $DP(d,\theta)$, these quantities are not directly observed. Suppose, however, that we have a continuous $i.i.d.$ prior $G:\mathbb{R}\rightarrow[0,1]$ over the elements of $\theta$ and that we observe a realization $y$ of the random vector $Y=(Y_1,\dots,Y_n)'$. The posterior probability that firm $i\in[n]$ is more biased than firm $j \neq i$ will again be denoted by $\pi_{ij}=\Pr\left(\theta_{i}>\theta_{j}|Y=y\right)$.

The posterior expectation under $G$ of both the unknown $\tau(d,\theta)$ and $DP(d,\theta)$ can be expressed in terms of the pairwise probabilities $\pi_{ij}$. The posterior expected rank correlation $\bar \tau(d)=\mathbb{E}[\tau(d,\theta)|Y=y]$ is given by 
\[
\bar \tau(d) = \binom{n}{2}^{-1}\sum_{i=2}^{n}\sum_{j=1}^{i}1\left\{ d_{i}<d_{j}\right\}\cdot \left(\pi_{ij}-\pi_{ji}\right) +1\left\{ d_{i}>d_{j}\right\} \cdot \left(\pi_{ji}-\pi_{ij}\right).
\]
Likewise, the posterior expected value of $DP(d,\theta)$, a quantity we term the \emph{Discordance Rate} (DR), is
\begin{align}
DR(d)=\binom{n}{2}^{-1}\sum_{i=2}^{n}\sum_{j=1}^{i-1}1\left\{ d_{i}<d_{j}\right\} \pi_{ij}+1\left\{ d_{i}>d_{j}\right\} \pi_{ji}. \label{def: DR}
\end{align}
Consequently, the posterior expected loss (i.e., the Bayes risk) of assigning grades $d\in[n]^n$ can be written:
\begin{flalign}
    \mathcal{R}(d;\lambda)&=\mathbb{E}[L(d,\theta;\lambda)\mid Y=y] \nonumber \\\nonumber
    &=(1-\lambda)DR(d) - \lambda \bar \tau(d) \\ \nonumber
    &=\binom{n}{2}^{-1} \sum_{i=2}^{n}\sum_{j=1}^{i}\pi_{ji}1\left\{ d_{i}>d_{j}\right\} +\pi_{ij}\left(1-1\left\{ d_{i}=d_{j}\right\} -1\left\{ d_{i}>d_{j}\right\} \right)-\\
&\lambda\pi_{ji}\left(1-1\left\{ d_{i}=d_{j}\right\} -1\left\{ d_{i}>d_{j}\right\} \right)-\lambda\pi_{ij}1\left\{ d_{i}>d_{j}\right\} .\label{eq:obj}
\end{flalign}

The optimal grades $d^*(\lambda)$ minimize $\mathcal{R}(d;\lambda)$. To simplify this minimization problem, it is convenient to recast the relevant decision
variables as pairwise indicators $d_{ij}=1\left\{ d_{i}>d_{j}\right\} $
and $e_{ij}=1\left\{ d_{i}=d_{j}\right\}$. Transitivity requires that for
any triple $\left(i,j,k\right)\in[n]^{3}$ the following constraints
hold:
\begin{equation}
d_{ij} +d_{jk} \leq 1+d_{ik}, \label{eq:constraints}
\end{equation}
\[
d_{ik} +\left(1-d_{jk}\right)\leq1+ d_{ij},
\]
\[
e_{ij} + e_{jk} \leq 1 + e_{ik}.
\]

Hence, we can rewrite the problem of choosing $d\in[n]^n$ to minimize \eqref{eq:obj} as that of choosing the binary indicators $\{d_{ij},e_{ij}\}_{i=2,j=1}^{i=n,j=i}$ to minimize
\begin{flalign}
    \sum_{i=2}^{n}\sum_{j=1}^{i}\pi_{ji}d_{ij} +\pi_{ij}\left(1-e_{ij} - d_{ij} \right)-\lambda\pi_{ji}\left(1-e_{ij} -d_{ij} \right)-\lambda\pi_{ij}d_{ij} ,\label{eq:obj2}
\end{flalign}
subject to the transitivity constraints in \eqref{eq:constraints} and the logical constraint that $e_{ij}+d_{ij}+d_{ji} = 1$ for all $(i,j)\in[n]^2$. Note that both the objective \eqref{eq:obj2} and the constraints are linear in the control variables. This reformulation therefore yields an integer linear programming problem, the solution to which can be computed with standard optimization packages.

\subsection{The role of $\lambda$}
To develop intuition for the role that $\lambda$ plays in the nature of the solution to our linear programming problem, it is useful to consider the task of ranking a single pair in \eqref{eq:obj2} while ignoring cross-pair constraints. Setting $d_{ij}=1$ minimizes risk whenever
\[
\underbrace{\pi_{ji}-\lambda\pi_{ij}}_{d_{ij}=1 \text{ } (i \succ j)}<\min\left\{ \underbrace{0}_{e_{ij}=1 \text{ (tie)}},\underbrace{\pi_{ij}-\lambda\pi_{ji}}_{d_{ij}=e_{ij}=0 \text{ } (i \prec j)}\right\}. 
\]

For $\lambda\in\left[0,1\right]$, we can rewrite this condition as
\begin{equation}
\pi_{ij}>(1+\lambda)^{-1}, \label{eq:thresholds}
\end{equation}
which is equivalent to the rule derived in section \ref{sec:gambling}. Hence, with $\lambda=1$, it is optimal to choose $d_{ij}=1\{\pi_{ij}>1/2\},$ which can be thought of as a MAP estimate of the pairwise rank. When $\lambda\in\left(0,1\right)$, greater posterior certainty is required to conclude that $d_{ij}=1$ and values of $\pi_{ij}$ near 1/2 will yield ties even though $G$ is continuous. As $\lambda$ approaches zero, fewer distinct grades will be assigned. When $\lambda=0$, all $n$ firms are assigned the same grade because $\pi_{ij}\leq1$. 

The coarse grades that result from applying the pairwise thresholding rule in \eqref{eq:thresholds} when $\lambda<1$ can generate a form of Condorcet cycle in \emph{indifferences} that violates the transitivity constraints in \eqref{eq:constraints} even if they would be satisfied under $\lambda=1$. The following three firm example illustrates the problem.

\begin{example}[Three firms, normal posteriors]
Suppose $n=3$ and $\theta_i | Y_i = y_i \sim N(\omega_i, 1)$. If the $\{\theta_i\}_{i=1}^3$ are mutually independent, then:
\begin{flalign*}
\pi_{ij} = \Pr(\theta_i > \theta_j | Y_i = y_i,  Y_j = y_j) = \Phi \left( \frac{\omega_i - \omega_j}{\sqrt{2}}\right).
\end{flalign*}
Let $\lambda=1/4$, which implies $(1+\lambda)^{-1}=0.8$. If $(\omega_1,\omega_3)=(2,0)$, so that $\pi_{13}=\Phi(\sqrt{2}) = .92$ and $\pi_{31}=1-\pi_{13}=.08$, then it is optimal to rank $\theta_1>\theta_3$. But if $\omega_{2}\in(0.81,1.19)$, it is optimal to rank both $(\theta_1,\theta_2)$ and $(\theta_2,\theta_3)$ as ties because $\max\{\pi_{12},\pi_{23}\}<0.8$. By transitivity, this implies $\theta_1=\theta_3$ which contradicts our earlier assertion that $\theta_1>\theta_3$.\qed
\end{example}

Note that if we had set $\lambda=1$ in the above example transitivity would have been satisfied because the posteriors themselves are transitive in the sense that for any triple $(i,j,k)$ of firms, $\pi_{ij}>\pi_{ji}$ and $\pi_{jk}>\pi_{kj}$ imply $\pi_{ik}>\pi_{ki}$. This transitivity derives from the scalar index structure of the posteriors, revealed by the fact that
\[
\pi_{ij}>\pi_{ji} \Longleftrightarrow \omega_i>\omega_j.
\]
\cite{sobel1993bayes} establishes the transitivity of posteriors in a broader exponential family subject to a corresponding index restriction. In general, however, when the observations are heteroscedastic, such index representations are not available and transitivity is not assured. When transitivity fails, the constraints in \eqref{eq:constraints} will bind and multiple units may receive the same grade even when $\lambda=1$.

Finally, note that coarse grades need not be a consequence of transitivity violations. If $\omega_2 \in (-1.19,0.81)$ in the preceding example, it is optimal to rank $\theta_1 > \theta_3$, $\theta_1 > \theta_2$, and $\theta_2 = \theta_3$. Thus pairwise thresholding yields two grades and no transitivity violations. Whether the transitivity constraints bind therefore depends on the structure of the pairwise posterior contrasts. 

\subsection{Connections to social choice} \label{sec: social_choice}

The literature on ranking methods bears a close connection to problems of social choice. If we re-interpret $\pi_{ij}$ as the share of votes for firm $i$ over firm $j$ in a pairwise election then a number of standard preference aggregation schemes suggest themselves.\footnote{In developing this analogy, we temporarily depart from the convention that $d_i>d_j$ implies firm $i$ has been assigned a ``worse'' grade than firm $j$, referring instead to firms with high $d_i$ as highly ranked.} For example, \cite{borda1784memoire}'s voting method simply ranks each firm $i$ based on the number of elections it has won; i.e., based upon $\sum_{j\neq i} 1\{\pi_{ij}>1/2\}$. If (as we have assumed) $G$ is continuous, then the Borda measure is equivalent to the posterior mean rank, a quantity studied by \cite{laird1989empirical}. 

The ranking procedure devised in section \ref{sec:risk} turns out to be closely tied to \cite{condorcet1785essay}'s voting scheme. To develop this connection, it is useful to define the \cite{kemeny1959mathematics} distance between the vectors $\theta$ and $d$, which can be written
\begin{align*}
    K\left(\theta,d\right)=\sum_{i=2}^{n}\sum_{j=1}^{i}\left|1\left\{ \theta_{i}>\theta_{j}\right\} -1\left\{ \theta_{i}<\theta_{j}\right\} -\left(1\left\{ d_{i}>d_{j}\right\} -1\left\{ d_{i}<d_{j}\right\} \right)\right|.
\end{align*}
Integrating out $\theta$ yields
\begin{align}
    \mathbb{E}\left[K\left(\theta,d\right)|Y=y\right]	&=	\sum_{i=2}^{n}\sum_{j=1}^{i}\left(\pi_{ij}+\pi_{ji}\right)1\left\{ d_{i}=d_{j}\right\} +\left(1+\pi_{ji}-\pi_{ij}\right)1\left\{ d_{i}>d_{j}\right\} \nonumber \\ &+\left(1+\pi_{ij}-\pi_{ji}\right)1\left\{ d_{i}<d_{j}\right\}. \nonumber
\end{align}
When ties are not possible (i.e., when $\pi_{ij}=1-\pi_{ji}$ for all $i\neq j$) we obtain the simplification
\begin{align}
\mathbb{E}\left[K\left(\theta,d\right)|Y=y\right]\propto	\sum_{i=2}^{n}\sum_{j=1}^{i}\left(2\pi_{ij}-1\right)\left(d_{ji}-d_{ij}\right) \label{eq:MAP}.
\end{align}
\cite{young1978consistent} show that \cite{condorcet1785essay}'s voting scheme is equivalent to choosing a ranking $d$ that minimizes \eqref{eq:MAP}. \cite{young1986optimal} establishes that this vote aggregation scheme is the unique rule that is unanimous, neutral, and satisfies reinforcement and independence of remote alternatives. It can also be viewed as a type of maximum likelihood estimator giving ``the ranking of all candidates that is most likely to be correct'' \citep[][p.114]{young1986optimal}.

Note that $\left(2\pi_{ij}-1\right)\left(d_{ji}-d_{ij}\right)$ is minimized by the MAP thresholding rule $d_{ij}^{MAP}=1\{\pi_{ij}>1/2\}$. When $\lambda=1$, our objective in \eqref{eq:obj2} reduces to \eqref{eq:MAP}. Consequently, the most granular version of our grading scheme minimizes the expected Kemeny distance between the assigned grades and the true rankings. When $\lambda<1$ we depart from the Kemeny criterion by calling elections a draw when they are close. Here, a close election is one where $\pi_{ij}<(1+\lambda)^{-1}.$

Condorcet rankings obey the famous Condorcet winner criterion: a unit that wins all pairwise elections between candidates (that is, satisfies $\pi_{ij}>1/2$ $\forall j \neq i$) will be ranked first. The following proposition reveals that when $\lambda<1$ our grades fulfill a modified version of the Condorcet winner criterion.

\begin{proposition}[$\lambda$-Condorcet Criterion]\label{prop1}
Suppose that firm $i$ satisfies $\pi_{ij} > (1+\lambda)^{-1} \ \forall \ j \neq i$. Then $d_i > d_j \ \forall \ j \neq i$. Moreover, suppose that firm $k$ satisfies $\pi_{ik} > (1+\lambda)^{-1}$ and  $\pi_{kj} > (1+\lambda)^{-1} \ \forall \ j \neq i, j \neq k$. Then $d_i > d_k > d_j \ \forall \ j \neq i, j \neq k$.
\end{proposition}

We leave the short proof for \ref{sec: proofs}. By symmetry of the objective in \eqref{eq:obj2}, the firm assigned the lowest grade by our method must achieve the highest grade when the sign of the estimand being ranked is reversed. Hence, Proposition \ref{prop1} also implies that any Condorcet \emph{loser} -- i.e., any candidate firm $i$ with $\pi_{ji}>(1+\lambda)^{-1}$ for all $j\neq i$ -- must be assigned the lowest grade. 

Another well-known property of Condorcet rankings is that when no Condorcet winner exists, the top ranked candidate must be a member of the \citet{smith1973aggregation} set: the smallest non-empty subset of candidates such that every candidate in the subset is majority-preferred over every candidate not in the subset. The following proposition establishes a corresponding property of our grades in the case where $\lambda<1$.

\begin{proposition}[$\lambda$-Smith criterion]\label{prop2A}
Let $\mathcal{S}$ denote a collection of firms exhibiting the following dominance property: $\pi_{ij}>(1+\lambda)^{-1}$ $\forall i\in \mathcal{S},j \notin \mathcal{S}$. Then the top graded firms must be a member of $\mathcal{S}$.
\end{proposition}

The proof is again left for the appendix. Symmetrically, Proposition \ref{prop2A} implies the firm assigned the lowest grade must be a member of the Smith loser set of candidates that are majority non-preferred to all others. Finally, we note that when $\lambda<1$ and no ordering is possible within the Smith set, all firms in the set will receive equal grades.

\begin{proposition}[Unordered $\lambda$-Smith candidates are tied]
Let $\mathcal{S}$ denote a collection of firms exhibiting the following dominance property: $\pi_{ij}>(1+\lambda)^{-1}$ $\forall i\in \mathcal{S},j \notin \mathcal{S}$. Moreover, suppose $\pi_{ij}<(1+\lambda)^{-1}$ $\forall (i,j) \in \mathcal{S}$. Then all firms in $\mathcal{S}$ receive the highest grade.
\end{proposition}

As with the preceding propositions, the proof appears in \ref{sec: proofs}.
 
\subsection{Extension to weighted loss}\label{sec:sqwt}

Ranking mistakes are likely to be more costly when the magnitude of the mistake is larger. The following family of loss functions weight pairwise concordances and discordances by the $p$'th power of the difference between the cardinal biases of the two firms:
\begin{flalign}\nonumber
L^p\left(d,\theta;\lambda\right)=\binom{n}{2}^{-1}\sum_{i=2}^{n}\sum_{j=1}^{i}
\Bigg[\underbrace{1\left\{ \theta_{i}>\theta_{j},d_{i}<d_{j}\right\}(\theta_{i}-\theta_{j})^p +1\left\{ \theta_{i}<\theta_{j},d_{i}>d_{j}\right\}(\theta_{j}-\theta_{i})^p }_{\text{discordant pairs}}-\\ \nonumber
\lambda\bigg(\underbrace{1\left\{ \theta_{i}<\theta_{j},d_{i}<d_{j}\right\}(\theta_{i}-\theta_{j})^p +1\left\{ \theta_{i}>\theta_{j},d_{i}>d_{j}\right\}(\theta_{j}-\theta_{i})^p }_{\text{concordant pairs}}\bigg) \Bigg].
\end{flalign}
A loss function corresponding to the $(p=2,\lambda=1)$ case was previously considered by  \cite{sobel1990complete}. The corresponding family of risk functions take the form
\begin{flalign}
    \mathcal{R}^p(d;\lambda)=\binom{n}{2}^{-1}\sum_{i=2}^{n}\sum_{j=1}^{i}\mu^p_{ji}d_{ij} +\mu^p_{ij}\left(1-e_{ij} - d_{ij} \right)-\lambda\mu^p_{ji}\left(1-e_{ij} -d_{ij} \right)-\lambda\mu^p_{ij}d_{ij},\nonumber
\end{flalign}
where $\mu^p_{ij}=\mathbb{E}\left[\max\{(\theta_i-\theta_j),0\}^p\mid Y=y\right]$. Note that $\lim_{p\rightarrow0} \mu_{ij}^p = \pi_{ij}$. Hence, one can think of our baseline risk function in \eqref{eq:obj2} as a limiting case of $\mathcal{R}^p$ as $p$ approaches zero. In what follows, we focus on this limiting case where $p\rightarrow0$ (``binary loss'') but also report results analyzing the case where $p=2$ (``square-weighted loss'') in \ref{appdx:sqwt}.

\subsection{Discordance rates}\label{sec:DR}

To summarize the reliability of the grades it is useful to report the Discordance Rate $DR(d^*)$, which gives the posterior expected frequency of discordances between pairs of grades. From \eqref{def: DR}, this quantity is trivial to compute, as it depends only on the optimized decisions $\{d_i^*\}_{i=1}^n$ and the posterior probabilities $\{\pi_{ij}\}_{i\neq j}$.

We can also define the conditional Discordance Rate between a specific pair of grades $g$ and $g'<g$ as
\begin{eqnarray*}
DR_{g,g'} & = & \frac{\sum_{i=2}^{n}\sum_{j=1}^{i-1}1\left\{ d_{i}^{*}=g\right\} 1\left\{ d_{j}^{*}=g'\right\} \mathbb{E}\left[1\left\{ \theta_{i}<\theta_{j}\right\} |Y=y\right]}{\sum_{i=2}^{n}\sum_{j=1}^{i-1}1\left\{ d_{i}^{*}=g\right\} 1\left\{ d_{j}^{*}=g'\right\} }\\
 & = & \frac{\sum_{i=2}^{n}\sum_{j=1}^{i-1}1\left\{ d_{i}^{*}=g\right\} 1\left\{ d_{j}^{*}=g'\right\} \pi_{ji}}{\sum_{i=2}^{n}\sum_{j=1}^{i-1}1\left\{ d_{i}^{*}=g\right\} 1\left\{ d_{j}^{*}=g'\right\} }.
\end{eqnarray*}
The denominator of each pairwise rate can be thought of as giving the number of rejections of the null hypothesis that a pair of firms discriminate equally in favor of the alternative that the firm assigned to group $g'$ is more biased than the firm assigned to group $g$. Hence, $DR_{g,g'}$ is an  analogue of the directional false discovery rate \citep{benjamini1995controlling,benjamini2005false}, giving the expected share of pairs with differing grades that are misranked. 

Because the conditional DRs are symmetric ($DR_{g,g'}=DR_{g',g}$), we report them as a lower triangular matrix. Note that the overall $DR$ is a weighted average of the conditional DRs with positive weight put on the ``on-diagonal'' terms $DR_{g,g}$, which are necessarily zero. When working with $p$-weighted loss a corresponding weighted version of the conditional discordance rate can be employed:
\begin{eqnarray*}
DR_{g,g'}^{p}&=&\frac{\sum_{i=2}^{n}\sum_{j=1}^{i-1}1\left\{ d_{i}^{*}=g\right\} 1\left\{ d_{j}^{*}=g'\right\} \mathbb{E}\left[\max\{(\theta_{i}-\theta_{j}),0\}^{p}\mid Y=y\right]}{\sum_{i=2}^{n}\sum_{j=1}^{i-1}1\left\{ d_{i}^{*}=g\right\} 1\left\{ d_{j}^{*}=g'\right\} \mathbb{E}\left[(\theta_{i}-\theta_{j})^{p}\mid Y=y\right]} \\
&=&\frac{\sum_{i=2}^{n}\sum_{j=1}^{i-1}1\left\{ d_{i}^{*}=g\right\} 1\left\{ d_{j}^{*}=g'\right\} (1-\mu_{ij}^{p})}{\sum_{i=2}^{n}\sum_{j=1}^{i-1}1\left\{ d_{i}^{*}=g\right\} 1\left\{ d_{j}^{*}=g'\right\} m_{ij}^{p}},
\end{eqnarray*}
where $m^p_{ij}=\mathbb{E}_{G}\left[(\theta_{i}-\theta_{j})^{p}\mid Y=y\right].$ The $p$-weighted discordance rate nests the corresponding unweighted rate as $DR^0_{g,g'}=DR_{g,g'}$. For any $p>0$, $DR^0_{g,g'}$ is guaranteed to lie in the unit interval. 

If we view the grades $\{d_i^{*}\}_{i=1}^n$ as exchangeable random variables and take $DR_{g,g'}$ as an assessment of the conditional misranking probability $\Pr(\theta_i>\theta_j \mid d_{i}^{*}=g,d_{j}^{*}=g')$ we can convert the conditional DR's into Bayes Factors expressing the informativeness of assigned grade pairs about the relative ranks of firms. By Bayes' rule,
\begin{eqnarray*}
BF_{g,g'}	&=&	\frac{\Pr\left(d_{i}^{*}=g,d_{j}^{*}=g'\mid\theta_{i}\leq\theta_{j}\right)}{\Pr\left(d_{i}^{*}=g,d_{j}^{*}=g'\mid\theta_{i}>\theta_{j}\right)}\\
	&=&	\frac{\Pr\left(\theta_{i}\leq\theta_{j} \mid d_{i}^{*}=g,d_{j}^{*}=g'\right)\Pr\left(\theta_{i}\leq\theta_{j}\right)}{\Pr\left(\theta_{i}>\theta_{j} \mid d_{i}^{*}=g,d_{j}^{*}=g'\right)\Pr\left(\theta_{i}>\theta_{j}\right)}\\
	&=&	\frac{1-DR_{g,g'}}{DR_{g,g'}},
\end{eqnarray*}
where we have used the fact that under a smooth $i.i.d.$ prior $\Pr\left(\theta_{i}>\theta_{j}\right)=1-\Pr\left(\theta_{i}\leq\theta_{j}\right).$ The Bayes Factor $BF_{g,g'}$ conveys the  odds that a randomly selected firm assigned grade $g$ is more biased than a randomly selected firm assigned grade $g'<g$. Corresponding $p$-weighted Bayes Factors can be formed in the same manner using $DR^p_{g,g'}$.

\section{Ranking Names}

As an introductory illustration of the methods developed in the previous section, we now rank the employer contact propensities of the names used in our correspondence experiment. The experiment utilized 76 first names, which were split equally between the nominal categories of: Black male, Black female, white male, and white female.

Table \ref{tab:sumstats_names} lists the mean contact rates of names in each of these categories, along with the number of applications. Distinctively white and female names were called back most often in the experiment, followed by white male names, then Black male names, with Black female names being called back least often. Though the same names were intended to be sent to each firm, the COVID-19 epidemic and other disruptions led to minor imbalances reflected in the sample counts. In our past work \citep{kline2021systemic} we were unable to reject the null hypothesis that the first names randomly assigned to our applications exhibited equal employer contact probabilities conditional on their nominal race and sex, which suggests these imbalances had little effect on mean contact rates by race or sex. For completeness, Table \ref{tab:sumstats_names} reports for each demographic group the $p$-value of a Wald test of the null hypothesis that the name specific contact rates are equal. The smallest such $p$-value is 0.24.

To get a sense of what share of the variation in name contact probabilities is likely explained by the names' putative race and sex labels, we compute a bias-corrected estimate of the between demographic group variance using the formula $\frac{G-1}{G}\left(S^2-\overline{s^{2}}\right)$ where $G=4$ is the number of demographic groups, $S^2$ is the sample variance across demographic groups of the point estimates  reported in Table \ref{tab:sumstats_names}, and $\overline{s^{2}}$ is the average squared standard error across those groups. This calculation yields an unbiased estimate of the between demographic group variance of $(0.011)^2$. Applying a corresponding calculation to the name specific means and standard errors yields a bias-corrected estimate of the variability in contact probabilities across all 76 first names of $(0.010)^2$. The finding that our between demographic group variance estimate exceeds our between name variance estimate suggests that the nominal race and sex of a first name explains nearly all of the variation across first names in contact probabilities; that is, we expect that a regression of the latent name specific contact probabilities on race and sex interactions would yield an $R^2$ very close to one. 

In principle, even if race and sex perfectly predict employer treatment of names, the causal factors generating this association could be other features of names that correlate strongly with race and sex. A candidate factor that has attracted substantial attention from social scientists is employer stereotypes about the likely productivity of individuals with different names \citep{fryer2004causes,gaddis2017names}. This hypothesis was evaluated by \cite{bertrand_mullainathan_2004}, who found that the average socioeconomic status of the first names considered in their experiment (as proxied by average maternal education) varied widely within race but were insignificantly related to contact rates.\footnote{Recent work by \cite{crabtree2022racially} directly elicits perceptions of educational attainment and income by first name on a variety of online platforms. Remarkably, they find that perceptions of social class exhibit variability across racially distinctive first names in the same race category comparable in magnitude to the variability found between race categories (see their Figure 4).} Because roughly half (40 out of 76) of our first names coincide with those studied by \cite{bertrand_mullainathan_2004}, our finding of insignificant contact probability differences within race and gender casts further doubt on the view that employer responses are driven primarily by features of names other than their likely race or sex. 

The finding that race and sex provide an accurate low dimensional summary of the 76 name specific contact probabilities suggests it is possible to build a highly informative ranking of the names involving just a few grades. Below, we investigate this conjecture in two ways. First, we examine how the expected Kendall's $\tau$ produced by our procedure scales with the number of grades assigned. Second, we treat each name's nominal race and sex as ``missing labels'' and study the extent to which the coarse grades assigned to first names by our ranking algorithm can recover these labels from data on firms' sample contact rates.

\subsection{Estimating $G$}

Abusing notation somewhat, let $i$ in this section refer to a first name and denote the number of applications with name $i$ sent in the experiment by $N_i$. The number of employer contacts received within 30 days by those applications is denoted by $C_i$. If the contacts are viewed as independent Bernoulli trials with name specific contact probabilities $p_i$ then the contact rate $C_i/N_i$ of name $i$ has mean $p_i$ and variance $p_i(1-p_i)/N_i$. This dependence of the variance on the contact probability complicates ranking exercises, as contact rates for names that deserve the best and worst grades -- that is, those with $p_{i}$ near one or zero -- will also be estimated with the least noise.

To stabilize the variance, we rank names according to a \cite{bartlett1936square} transformation of their contact rates 
\[
\hat \theta_i = \sin^{-1} \sqrt{C_i/N_i}.
\]
The logic of this transform follows from the observation that $\frac{d}{dx} \sin^{-1} \sqrt{x}=\left[2\sqrt{x(1-x)}\right]^{-1}$. Consequently, the Delta method implies $\hat \theta_i$ has asymptotic distribution $\mathcal{N}(\theta_i, (4N_i)^{-1}),$ where $\theta_i=\sin^{-1} \sqrt{p_i}$.

To estimate the distribution $G$ from which $\theta_i$ was drawn we first apply the non-parametric maximum likelihood (NPMLE) estimator developed by \cite{koenker2014convex} as implemented by the `GLmix' command in the R package REBayes \citep{koenker2017rebayes}. The NPMLE estimates a discrete approximation to the distribution of $\theta_i$ assuming that $\hat \theta_i\mid \theta_i,N_i \sim \mathcal{N}(\theta_i,(4N_i)^{-1})$. Supporting the maintained independence of $\theta_i$ from $N_i$, a regression of $\hat \theta_i$ on $\ln N_i$ yields a statistically insignificant relationship ($p$=0.17).\footnote{The variation in $N_i$ is primarily attributable to the fact that a subset of our first and last name pairs were taken from the study of \cite{bertrand_mullainathan_2004}, while the remaining name pairs were drawn from North Carolina data on speeding tickets and Census data. The number of last names considered differed across the two data sources, leading to imbalances in the average number of last names (and hence applications) per first name.}

A plot of the estimated marginal distribution $\hat G$ of $\theta_i$ produced by the NPMLE is provided in Figure \ref{fig: name_g}. The bars correspond to histograms of $\hat \theta_i$ while the yellow spikes represent the estimated probability mass function $d \hat G$ of the $\theta_i$. This discrete distribution does an excellent job matching the mean value of the $\hat \theta_i$ and its bias corrected variance, which we compute as its sample variance minus its average squared standard error $n^{-1}\sum_{i=1}^ns_i^2=n^{-1}\sum_{i=1}^n(4N_i)^{-1}$.

Figure \ref{fig: name_g} also plots the estimated density of $\theta_i$ produced by  \cite{efron2016empirical}'s log-spline estimator, which models $G$ as a 5th order spline in the exponential family. Estimation of the spline parameters is conducted via penalized maximum likelihood, where $N_i$ is treated as independent of $\theta_i$. The penalization parameter has been chosen by GMM to match unbiased estimates of $G$'s first two moments as closely as possible. Despite being continuous, the bimodal shape of the log-spline estimate is remarkably consistent with that of the NPMLE. For reference, the sample mean values of $\hat \theta_i$ for each nominal race and sex category are portrayed on the Figure as vertical lines. Despite not using race in the estimation procedure, the two modes of the log-spline estimate fall near the race specific means as do the modes of the NPMLE estimate.

The lower panel of Figure \ref{fig: name_g} converts these estimates back into probability points via the inverse transform $\theta \mapsto \sin(\theta)^2$. The NPMLE finds two large mass points: one at $p=0.226$ and another at $p=0.244$. The 1.8 percentage point gap between these mass points is very near the Black-white contact gap in the experiment of 2.1 percentage points. Likewise, the distance between the modes of the log-spline estimate is roughly 2.1 percentage points. The NPMLE also finds a third mass point at $p=0.260$, which lies just below the estimated average contact rate for distinctively white female names.

As noted earlier, the discrete $\hat G$ produced by the NPMLE is a data dependent approximation to the true $G$. Even if differences in the treatment of names are driven primarily by employer perceptions of their race and sex, it seems unlikely that the true $G$ is literally characterized by a few mass points, as small differences across names in their perceived race should generate corresponding contact rate differences. In fact, although three discrete modes are visible from the Figure, the NPMLE solution involves dozens of atoms with positive mass, most of which are imperceptible. In what follows, we rely on the log-spline estimate $\hat G$ of $G$ which, as in the theoretical analysis of Section \ref{sec:decision}, implies that ties are measure zero.

\subsection{Reporting possibilities}

The top left panel of Figure \ref{fig: names_res1} depicts the posterior contrast probabilities $\pi_{ij}=\Pr(\theta_i>\theta_j \mid \hat \theta_i,\hat \theta_j,N_i,N_j)$; see the Appendix for details on how these posteriors were computed. Names have been ordered according to their Condorcet rank (i.e., their grade when $\lambda=1$). To ease interpretation, we have labeled the name with the highest ranked contact probability 1 and that with the lowest ranked contact probability 76. Name pairs with adjacent ranks exhibit $\pi_{ij}$'s near 1/2, indicating that we have little confidence in their relative order. Reassuringly, names with vastly different ranks exhibit $\pi_{ij}$'s near 0 or 1, indicating that the experimental data are informative about the relative ordering of these pairs. 

The top right panel of Figure \ref{fig: names_res1} depicts the Discordance Rate that arises from minimizing $\mathcal{R}(d;\lambda)$ -- that is, from solving \eqref{eq:obj2} subject to \eqref{eq:constraints} -- for different choices of $\lambda$. The number used to demarcate each solution corresponds to the number of distinct grades the integer linear programming procedure produces for that choice of $\lambda$. A sharp ``elbow'' emerges around $\lambda=0.18$, above which the $DR$ grows rapidly. 

The bottom panel depicts the trade off between grade reliability $1-DR$ and informativeness $\bar \tau$ associated with our choice of $\lambda$. The data are potentially quite informative about name rankings: as $\lambda$ approaches 1, the expected rank correlation $\bar \tau$ approaches 0.44. However, the reliability of such a report would be fairly low, yielding a Discordance Rate of 0.28. For comparison, we also show the results of naively ranking based on $\hat \theta$ or the EB posterior mean $\bar \theta_i = \mathbb{E}[\theta_i|Y=y]$. Remarkably, both naive approaches yield ranks with $\bar \tau$ and $DR$ very similar to those produced by our procedure when $\lambda=1$.

To improve the reliability of the grades, we set $\lambda=.25$, implying via equation \eqref{eq:thresholds} that, in the absence of transitivity considerations, we would require 80\% posterior certainty of each pairwise ranking decision that is not an abstention. This choice yields two grades that exhibit a remarkably high level of informativeness ($\bar \tau=0.29$) and an acceptable level of reliability ($DR=.07$). For comparison, lowering the implicit posterior threshold to 70\% by setting $\lambda=0.41$ would yield three grades and increase their informativeness by 11\% (to $\bar \tau=0.32$) at the expense of a 21\% increase in $DR$. Conversely, requiring $\lambda < 0.18$ would generate only one grade, yielding both $\bar \tau$ and $DR$ of zero by construction.

\subsection{Grades and demographics}

Figure \ref{fig: names_ranks} lists the first names according to their Condorcet ranking, along with the posterior mean of each name's contact probability $p_i=\sin(\theta_i)^2$. In addition to the posterior means, which are depicted as dots, we report posterior credible intervals connecting the 2.5th percentile of each name's posterior distribution of contact probabilities to the 97.5th percentile of its posterior distribution. One should expect that approximately 72 (i.e., 95\%) of these 76 intervals contain their name's true latent contact rate.\footnote{The asymmetry of the credible intervals reflects both that the estimated mixing distribution $\hat G$ of $\theta_i$ is asymmetric and that we have fed the interval limits through the nonlinear transformation $\theta \mapsto \sin(\theta)^2$.} While the credible intervals tend to be fairly short---spanning between two and three percentage points in most cases---there is clearly enough uncertainty about each name's contact probability to significantly complicate the task of ranking them. 

The Condorcet ranks can be thought of as a one dimensional encoding of the $\binom{76}{2}=2,850$ pairwise contrasts depicted in the upper left panel of Figure \ref{fig: names_res1}. When the posteriors are transitive, as they are in this case, this reduction amounts to the most likely ordering of the name contact probabilities. Variation in $N_i$ across names, and hence the precision with which contact rates are measured, could in principle generate substantial non-monotonicity of the posterior mean in the Condorcet rank. In practice, however, names' Condorcet rankings are very nearly monotone in their posterior means.

The Condorcet ranks are extremely correlated with race. Of the top 38 ranked first names, only 8 are distinctively Black. Though the three top ranked names ``Misty,'' ``Heather,'' and ``Laurie'' are all distinctively female, the presumptive sex of a name turns out to be only weakly related to its Condorcet rank: 19 of the top 38 names are distinctively male. Hence, the Condorcet ranks manage to recover the race labels from contact rates with very little error but serve as unreliable proxies of a name's sex. 

By construction, the Condorcet ranks exhibit the strongest expected rank correlation with the latent $\theta$ ranks. The coarse ranks that emerge when $\lambda<1$ sacrifice rank correlation in exchange for the prospects of incurring fewer mistakes. Each name's color reflects its assigned grade (i.e., its coarse rank). A depiction of how the grades vary with name-specific contact rates and their standard errors is provided in Appendix Figure \ref{fig:grade_frontiers_names}. As expected, names with higher contact rates tend to earn the better grade $\star\star$. However, heteroscedasticity in the estimates prevents the grades from being characterized by a single cutoff contact rate.

Though we saw earlier that the expected rank correlation of our grades with the true latent ranks is 0.29, it is also of interest to know how much $p_i$ varies across grades. As described in \ref{sec: posteriors}, we can use our EB posteriors to compute an estimate of the variance of $p_i$ across grades. Though our procedure assigns only two grades to the names, we estimate that the (name-weighted) between grade standard deviation in contact probabilities is 0.006. Since the marginal standard deviation of $p_i$ is roughly 0.010, a regression of the latent $p_i$ on our grades should yield an $R^2$ of 35\%.

The coarse grades that emerge from our procedure continue to align closely with our race labels: 35 of the 53 names (66\%) in the top grade are distinctively white, while just 3 of the 23 names (13\%) in the second grade are white. Notably, the top two names are also female; however, they do not appear in their own grade. Hence, a two-group ranking recovers the missing race label with limited error and, consistent with our findings in Table \ref{tab:sumstats_names}, suggests that white female names are particularly favored.

It is natural to wonder if a solution with more grades would be more predictive of sex. Appendix Figure \ref{fig:pseudoR2} reports the pseudo-$R^2$ and Area under the Curve (AUC) from a series of logistic regressions of the name's sex on grade indicators for different choices of $\lambda$. Note that if we were to set $\lambda=1$, this regression would necessarily predict sex perfectly, as every name would receive its own dummy indicator. However, the four-grade solution with the smallest value of $\lambda$ yields a pseudo-$R^2$ for sex of 0.012. With five grades we find a pseudo-$R^2$ for sex of 0.034. By contrast, a corresponding logistic regression of race on assigned grades yields pseudo-$R^2$s for four- and five-grade solutions of 0.28 and 0.23, respectively. 

Appendix Figure \ref{fig: names_ranks_sq_weighted} shows the grades that result from minimizing expected square-weighted loss $\mathcal{R}^2(d,\lambda)$, which weights large mistakes in pairwise rankings more heavily than small mistakes. As above, we set $\lambda=0.25$, which now yields five distinct grades. Once again, the grade categories are better predictors of race than sex. As shown in Appendix Figure \ref{fig:pseudoR2_sq}, a logistic regression of sex on these five grade categories yields a pseudo-$R^2$ of roughly 0.05, while a corresponding race regression yields a pseudo-$R^2$ of 0.24. \ref{appdx:sqwt} provides further details on the square-weighted rankings. 

We conclude that our grades are fairly strong predictors of a name's race but not its sex. Given that the overall gender gap in contact rates is indistinguishable from statistical noise in our experiment, the failure to predict gender is not particularly surprising. The ability to predict race for a wide range of choices of $\lambda$, however, suggests that our grading scheme can be effective at detecting latent group structure even when the number of units being ranked is relatively modest.

\section{Ranking firms}

We turn now to ranking firms in their relative treatment of Black versus white names. The conduct of each firm $i$ in our experiment is characterized by the race-specific contact probabilities $(p_{iw},p_{ib})$. These probabilities represent the hypothetical 30 day contact rates that would arise for applications with distinctively white and Black names, respectively, if we were to sample an infinite number of job vacancies from firm $i$ and send each job four pairs of applications. The sample contact rates $(\hat p_{iw},\hat p_{ib})$ provide unbiased estimates of these contact probabilities.

To mitigate the potential influence of firm heterogeneity in baseline contact rates on our measure of discrimination, we focus on the following proportional measure of bias against Black names at firm $i$:
\begin{align*}
    \theta_i=\ln(p_{iw})-\ln(p_{ib}).
\end{align*}
Our estimator of $\theta_i$ will be the plug-in analog $\hat \theta_i = \ln(\hat p_{iw})-\ln(\hat p_{ib})$. Because the number of applications sent to each firm is large, we employ the Delta method to construct a standard error $s_i$ for each $\hat \theta_i$ based on the job-clustered sampling covariance matrix of the sample contact rates. Although $\hat \theta_i$ is not fully variance-stabilized, the log transform removes any direct dependence of the variance on $\theta_i$ itself.\footnote{Specifically, a second-order Taylor expansion of $\hat{p}_{iw}/\hat{p}_{ib}=\exp\left(\hat{\theta}_{i}\right)$ around the point $(p_{iw},p_{ib})$ yields the approximation $\mathbb{V}\left[\hat{p}_{iw}/\hat{p}_{ib}\right]\approx\theta_{i}^{2}\left\{ \frac{\mathbb{V}\left[\hat{p}_{iw}\right]}{p_{iw}^{2}}+\frac{\mathbb{V}\left[\hat{p}_{ib}\right]}{p_{ib}^{2}}-2\frac{\mathbb{C}\left[\hat{p}_{iw},\hat{p}_{ib}\right]}{p_{iw}p_{ib}}\right\}.$ Consequently, the Delta method implies that $\mathbb{V}\left[\hat{\theta}_{i}\right]\approx\frac{\mathbb{V}\left[\hat{p}_{iw}\right]}{p_{iw}^{2}}+\frac{\mathbb{V}\left[\hat{p}_{ib}\right]}{p_{ib}^{2}}-2\frac{\mathbb{C}\left[\hat{p}_{iw},\hat{p}_{ib}\right]}{p_{iw}p_{ib}}.$}

In what follows, we exclude the eleven firms in the original experiment with callback rates below 3\% or fewer than 40 total sampled jobs, since the estimated contact ratios for these firms may be unreliable. Summary statistics for the remaining estimation sample of 97 firms are provided in Table \ref{tab:sumstats}. The unweighted average value of $\hat \theta_i$ across these 97 firms is 0.095, implying the typical firm in our sample exhibits a bias against Black names of roughly 10\%. Detailed point estimates and uncertainty measures for all 97 firms used in our analysis are provided in Appendix \ref{tab: results_detail}.

Twenty-one of the 97 estimated contact gaps are negative, indicating a preference for distinctively Black names. The firm-specific estimates are noisy, however, possessing an average standard error of 0.104. To test whether all firms in fact weakly prefer white to Black names (i.e., the joint null that $\theta_i\geq0$ $\forall i\in[n]$) we apply the high dimensional inequality testing procedure of \cite{bai2021two}. This procedure yields a $p$-value of 0.94, suggesting the observed negative point estimates are likely attributable to chance. 

Although the asymptotic variance of $\hat \theta_i$ does not mechanically depend on $\theta_i$, it is possible for $\theta_i$ and $s_i$ to be correlated. The top panel of Appendix Figure \ref{fig: theta_vs_s} plots $\hat \theta_i$ against $s_i$, revealing that firms with more precise estimates tend to show less bias against Black names. The Spearman correlation between between $\hat \theta_i$ and $s_i$ is 0.36 ($p<0.001$).

\subsection{A model of precision dependence} \label{sec: precision_model}
In light of the above findings, we assume that each $\theta_i$ is non-negative and may depend (statistically) on its standard error $s_i$. A simple model satisfying these criteria is:
\begin{equation}
\label{eq:dependence_model}
\theta_{i}=\exp\left(\beta\ln s_{i}+\ln v_{i}\right)=s_{i}^{\beta}v_{i}, \quad v_{i} \mid s_i \sim G_v \text{ for all } i\in [n]. 
\end{equation}
The parameter $\beta$ governs how the conditional distribution of bias varies with the standard error $s_i$. The latent variable $v_i$ captures heterogeneity in discriminatory conduct among firms with similar standard errors, and follows a distribution $G_v$ with strictly positive support. When $\beta$ is positive, both the mean and variance of $\theta_i$ increase monotonically with $s_i$. To complete the model, we link our estimates $\hat \theta_i$ to $\theta_i$ as follows:
\[
\hat \theta_i = \theta_i + s_ie_i, \quad e_i \mid s_i,v_i \sim \mathcal{N}(0,1) \text{ for all } i\in [n],
\]
where $s_ie_i$ is the noise in $\hat{\theta}_{i}$ attributable to the fact that a finite number of jobs were sampled from firm $i$.

To evaluate the plausibility of this model, we scrutinize some of the moment conditions it implies. Letting $\mathbb{E}[v_{i}|s_{i}]=\mu_{v}>0$ and $\mathbb{V}(v_{i}|s_{i})=\sigma_{v}^{2}>0$, consider the following ``studentized'' version of $\hat \theta_i$:
\[
T_{i}=\dfrac{\hat{\theta}_{i}-s_{i}^{\beta}\mu_{v}}{\sqrt{s_{i}^{2\beta}\sigma_{v}^{2}+s_{i}^{2}}}.
\]
Our model implies that $T_i$ has mean zero and marginal variance one. Moreover, $T_i$ should be independent of $s_i$. These restrictions imply the following four moment conditions:
\begin{align}
\mathbb{E}[T_{i}]=0,\ \mathbb{E}[T_{i}s_{i}]=0,\ \mathbb{E}[T_{i}^{2}-1]=0,\ \mathbb{E}[(T_{i}^{2}-1)s_{i}]=0. \label{eq:Tmoments}
\end{align}
Imposing these moment conditions via two-step efficient GMM yields the parameter estimates reported in Table \ref{tab:gmm_table}. The minimized value of the GMM criterion function suggests the model's over-identifying restrictions -- which test the joint requirement that the $T_i$ have mean zero and constant variance across all values of $s_i$ -- are satisfied ($p=0.97$). The GMM estimate of $\beta$ is $\hat \beta \approx 1/2$, indicating that $\theta_i$ is roughly proportional to $\sqrt{s_i}$. The large estimated value of $\sigma_v$ reveals that discriminatory conduct varies substantially among firms with similar standard errors. 

The top panel of Appendix Figure \ref{fig: theta_vs_s} superimposes the estimated conditional expectation function $\mathbb{\hat E} [\theta_i|s_i]=s_i^{\hat \beta} \hat \mu_v$ on the scatterplot of $\hat \theta_i$ against $s_i$. Consistent with the $J$-test from GMM estimation, the estimated conditional mean fits the cloud of points closely. The bottom panel of Appendix Figure \ref{fig: theta_vs_s} plots values of the estimated residual $\hat T_i=\dfrac{\hat{\theta}_{i}-s_{i}^{\hat \beta}\hat \mu_{v}}{\sqrt{s_{i}^{2\hat\beta}\hat \sigma_{v}^{2}+s_{i}^{2}}}$ against $s_i$. Consistent with our model, $\hat T_i$ exhibits roughly constant variance and a mean near zero throughout the observed range of $s_i$.

\subsection{Estimating $G$}

To estimate the population distribution of contact penalties, we deconvolve the residual $\hat v_i= \hat \theta_i / s_i^{\hat \beta}$ which, by the Delta method, obeys
 \begin{flalign*}
 \hat{v}_i \ | \ v_i,s_{i} \sim \mathcal{N}\left(v_i,s_i^{2(1-\beta)}\right), \ \ v_i \ | \ s_i \sim G_v, \text{ for all } i\in[n].
 \end{flalign*}
Relying again on a variant of \cite{efron2016empirical}'s log-spline estimator, we parametrize $G_v$ as a fifth order natural spline with strictly positive support. The spline parameters are estimated by penalized maximum likelihood with the penalty term chosen to minimize the distance to our earlier GMM estimates $(\hat \mu_v, \hat \sigma_v^2)$ of the first two moments of $v_i$. We then integrate over the empirical distribution of $s_{i}$ to convert the estimated $\hat{G}_{v}$ into an estimate $\hat{G}_{\theta}$ of the distribution of contact gaps $\theta_{i}$ .

Figure \ref{fig: poisson_g} plots the log-spline estimate $\hat{G}_\theta$ overlaid against the histogram of contact gap estimates.  $\hat{G}_\theta$ is less dispersed than the histogram, reflecting the noise in the estimates $\{\hat{\theta}_i\}_{i=1}^n$. Unlike with our earlier analysis of names, the density $\hat{G}_\theta$ is unimodal but highly skewed. While most firms exhibit little bias against Black names, some exhibit large biases of 20-40\%. By construction, no firms are estimated to discriminate against white names.

As a robustness check, we also compute NPMLE estimates using the GLVmix procedure developed by \cite{koenker2017rebayes}, which estimates a bivariate discrete distribution for $(\theta_i,N_is_i^2)$ under the assumption that $\theta_i$ is independent of $N_i$. The resulting marginal distribution of $\theta_i$ exhibits many mass points and is also unimodal. The NPMLE estimate of the variance of the $\theta_i$'s departs somewhat from both the log-spline estimate and a simple bias-corrected variance estimator $n^{-1}\sum_i [(\hat \theta_i - \bar \theta)^2 -s_i^2]$. However, the NPMLE and log-spline estimates appear comparable in their overall shape. Since a discrete distribution with exact ties seems implausible, we rely again on the log-spline estimates in what follows. The EB posterior distribution of $\theta_{i}|Y_{i}$ inherits the continuity of the log-spline deconvolution estimate of the prior distribution, which has the added benefit of simplifying computation of posterior credible intervals for each $\theta_i$.

\subsection{Industry effects} \label{sec: industry_rfe}

In \cite{kline2021systemic} we found large differences in the magnitude of contact gaps across 2-digit NAICS industries. Many of these industries have only 2 or 3 firms, precluding a fixed effects approach to incorporating industry affiliation into the model. We therefore employ a hierarchical random effects specification of $v_i$ taking the form:
\[
    v_{i} = \eta_{k(i)}\xi_{i}, \label{eq:doubleG} 
\]
\[
    \xi_{i} \ | \ s_{i},\eta_{k(i)} \sim G_{\xi}, \ \ i\in\{1,...,n\}, 
\]
\[
\eta_k \ | \ \mathbf{s}_{k} \sim G_{\eta}, \ \ k\in\{1,....,K\},
\]
where the function $k:\{1,\dots,n\}\rightarrow\{1,\dots,K\}$ returns a firm's industry and $\mathbf{s}_{k}$ is the vector of standard errors for all firms with $k(i)=k$. The industry effect $\eta_{k(i)}$ captures correlation in discriminatory conduct among firms in the same industry, while the firm effect $\xi_i$ captures departures from the industry average. As a normalization we assume $\mathbb{E}[\eta_k]=1$, which implies $\mathbb{E}[\xi_i]=\mu_v$.

The marginal variance of $v_{i}$ in this model is $\sigma_{v}^{2}=\sigma_{\eta}^{2}\sigma_{\xi}^{2}+\sigma_{\eta}^{2}\mu_{v}^{2}+\sigma_{\xi}^{2}$, where $\sigma_{\xi}^{2}$ gives the variance of $\xi_i$ and $\sigma_{\eta}^{2}$ the variance of $\eta_k$. To separately identify the between and within industry variance components, we add two new moment conditions to the set listed in \eqref{eq:Tmoments}. Denote the average value of $\hat v_i$ in industry $k$ by
\[
\ensuremath{\bar{v}_{k}= n_{k}^{-1}\sum_{i:k(i)=k}\hat{\theta}_{i}/s_{i}^{\beta}}=n_{k}^{-1}\sum_{i:k(i)=k}v_{i}+n_{k}^{-1}\sum_{i:k(i)=k}e_{i}/s_{i}^{\beta},
\]
where $n_k$ gives the number of firms in industry $k$. The variance of $\bar{v}_k$ in this model can be shown to be $V_k \equiv \left(\sigma_{\eta}^{2}\sigma_{\xi}^{2}/n_{k}+\sigma_{\eta}^{2}\mu_{v}^{2}+\sigma_{\xi}^{2}/n_{k}\right)+n_{k}^{-1}\sum_{i:k(i)=k}s_{i}^{2\left(1-\beta\right)}$. Letting $\ensuremath{\bar{s}_{k}=n_{k}^{-1}\sum_{i:k(i)=k}s_{i}}$ denote the average standard error in industry $k$, our two new moment conditions can be written
\[
\mathbb{E}\left[\left(\bar{v}_{k}-\mu_{v}\right)^{2}-V_{k}\right]=0,\quad\mathbb{E}\left[\left\{ \left(\bar{v}_{k}-\mu_{v}\right)^{2}-V_{k}\right\} \bar{s}_{k}\right]=0.
\]
The first condition simply equates the empirical squared deviations of the $\bar{v}_{k}$ around the model implied mean to the model implied variance. The second condition prohibits heteroscedasticity with respect to $\bar{s}_k$.

GMM estimates of the parameters of this hierarchical model are reported in the second column of Table \ref{tab:gmm_table}. The model's over-identifying restrictions again appear to be satisfied ($p=0.95$). While the variance $\sigma^2_{\eta}$ of the industry component is estimated to be nearly 10 times as large as the variance $\sigma^2_{\xi}$ of the firm specific component, the multiplicative influence of these components on $v_i$ implies that roughly half of the marginal variance in $v_{i}$ stems from within industry variation.\footnote{By the law of total variance $\mathbb{V}\left[v_{i}\right]=\mathbb{V}[\mathbb{E}[v_{i}|k(i)]]+\mathbb{E}[\mathbb{V}[v_{i}|k(i)]].$ Plugging in $v_i=\eta_{k(i)}\xi_i$ reveals that the within share $\mathbb{V}[\mathbb{E}[v_{i}|k(i)]]/\mathbb{V}\left[v_{i}\right]$ evaluates to $(\sigma_{\eta}^2+1)\sigma_{\xi}^2/\sigma_v^2.$} 

To identify the marginal distribution of $\theta_i$, we assume that both $G_{\eta}$ and $G_{\xi}$ lie in the exponential family parameterized by a 5th order spline. Generalizing \cite{efron2016empirical}'s log-spline estimator to the hierarchical case, these distributions are estimated by penalized maximum likelihood (see \ref{sec: industry_fx} for details). The two penalty parameters in this likelihood function are chosen so that the resulting distributions match GMM estimates of the between-industry and total variances of $\theta_{i}$.

Estimates of $G_{\xi}$ and $G_{\eta}$ are displayed in the top panel of Figure \ref{fig: industry_g}. Table \ref{tab:double_g_table} reports moments of the within- and between-industry distributions implied by the log-spline estimates as well as moments of the overall contact ratio $\theta_{i}=s_i^{\beta}\eta_{k(i)}\xi_i$. Standard errors for the moments are computed via the Delta method, treating the fifth-order splines as correctly specified models for the log density functions. The  mean contact gap, between-industry standard deviation, and total standard deviation reported in Table \ref{tab:double_g_table} closely match the corresponding GMM estimates of these parameters in Table \ref{tab:gmm_table}.

As can be seen in Figure \ref{fig: industry_g}, the industry component $\eta_{k}$ is more variable than the firm component $\xi_i$ and exhibits positive skew and excess kurtosis, reflecting that some industries feature particularly heavy discrimination against Black names. Recall however that the location of the industry effect distribution is not informative as we have normalized $\mathbb{E}[\eta_k]=1$. The bottom panel of Figure \ref{fig: industry_g} shows that the implied distribution of $\theta_i$ is similar to the estimate from the model without industry effects in Figure \ref{fig: poisson_g}, with a peak at small
contact penalties and a long right tail. As expected, the deconvolved distribution is more compressed than the empirical distribution of estimated contact gaps.

\subsection{Reporting possibilities}

Figure \ref{fig: Pis} plots the pairwise posterior ranking probabilities $\pi_{ij}$ with firms ordered by their rank under $\lambda=1$. Following our earlier convention with the names, these ranks range from 1 (the largest contact penalty) to 97 (the smallest contact penalty). Panel (a) shows results from our baseline specification with the log-spline estimate of the marginal mixing distribution as prior, while panel (b) reports results based on the hierarchical log-spline model with industry effects. Because the firm assigned rank 1 is deemed most discriminatory, many other firms are more likely than not to have lower values of $\theta$. Firms of middling rank, on the other hand, are more difficult to distinguish from others. Including industry effects tightens the posteriors, which leads the $\pi_{ij}$'s to become more dispersed around 1/2. This phenomenon is apparent in the more distinct ridge along the diagonal of the matrix in panel (b) compared to the baseline specification in panel (a).

Figure \ref{fig: Pis_graded} displays only the pairwise probabilities that satisfy the naive thresholding rule $\pi_{ij}>(1+\lambda)^{-1}$ when $\lambda$ has been set to 0.25. The resulting frontier implies numerous transitivity violations. For example, in panel (a), firm \#9 cannot be distinguished from firm \#4 or firm \#49, suggesting each of these pairs in isolation would be labeled a tie. However, firm \#49 is clearly distinguishable from firm \#4, yielding a contradiction. Super-imposed on the figure we show a frontier corresponding to the three grades that solve \eqref{eq:obj2} subject to \eqref{eq:constraints} when $\lambda=0.25$. These frontiers can be viewed as a transitivity-constrained version of the thresholding rule. 

Figure \ref{fig:poisson_binary_lambda} plots the number of distinct grades that result from minimizing $\mathcal{R}(d;\lambda)$ along with the  Discordance Rate of those grades as a function of the parameter $\lambda$. As expected, the number of grades tends to increase with $\lambda$ as does the $DR$. In the absence of industry effects, setting $\lambda=0.25$ yields three groups and an unconditional DR of roughly \binarydr\%. Introducing industry effects yields four groups and decreases the DR to \industrydr\%. 

Figure \ref{fig:PPF} illustrates the empirical tradeoff between the information content of our grades, quantified by the expected rank correlation $\bar \tau$, and their reliability, as quantified by the Discordance Rate. Without industry effects, setting $\lambda=1$ yields $\bar \tau=\binarylambdaone$ and a Discordance Rate of \binarydrlambdaone. Including industry effects increases the $\bar \tau$ of the Condorcet ranks to \industrylambdaone\ and lowers their DR to \industrydrlambdaone. In contrast, ranking naively on $\hat \theta_{i}$ yields both a higher Discordance Rate and lower $\bar \tau$ than the Condorcet ranks, indicating such an approach is both less informative and less reliable. Interestingly, ranking based upon the EB posterior means yields a $\bar \tau$ and DR essentially equivalent to the Condorcet ranks.\footnote{While ranking based upon posterior means is known to possess certain optimality properties when $G$ is normal and the normal noise is homoscedastic \citep{portnoy1982maximizing}, our environment features both heteroscedasticity and a decidedly non-normal mixing distribution $\hat G$.} 

To improve the reliability of the Condorcet ranks, we set $\lambda=0.25$. In the absence of transitivity violations, this choice of $\lambda$ requires a posterior threshold of at least 80\% to make pairwise ranking decisions. As depicted in Figure \ref{fig: Pis_graded}, however, numerous transitivity violations emerge in this example. Resolving these violations raises the required posterior certainty above 80\% in most instances, yielding a Discordance Rate of only \binarydr \% in the baseline specification without industry effects and \industrydr \% in the hierarchical specification with industry effects. Fortunately, the resulting grades remain highly informative: $\bar \tau$ is \binarytau \ in our baseline specification and \industrytau \ when industry effects are included.

\section{Discrimination report cards}

The report cards generated by our grading procedure provide a concise, low-dimensional summary of differences in discrimination across firms. This can be seen in Figure \ref{fig:poisson_binary}, which displays the report card results for the baseline specification without industry effects. In addition to the report card grades, the Figure plots a posterior mean estimate of each firm's bias $\theta_i$ along with 95\% credible intervals, which are constructed by connecting the posterior 2.5th percentile of $\theta_i$ to the posterior 97.5th percentile. The lower limit of each credible interval is positive as a result of our support restriction ruling out bias against white applicants. The firms are ordered by their Condorcet ranks (i.e., their grades under $\lambda=1$). In this draft we have replaced each firm's name with the rank of its bias estimate $\hat \theta_i$, which allows us to compare firm orderings across loss functions. Firms that are federal contractors, and hence subject to higher regulatory standards regarding equal opportunity laws, have been listed in black, while those that are not contractors are listed in gray. After the paper has undergone peer review, we intend to label the report card with firms' actual names. 

Setting $\lambda=0.25$ generates a report card with three grades, represented in Figure \ref{fig:poisson_binary} by a number of $\star$'s between one (the worst grade) and three (the best). The shading of credible intervals reflects the grade assigned to each firm. Most firms receive the middle grade of $\star\star$, which reflects both the noise in our estimates and the shape of the estimated distribution $\hat G$. By contrast, only two firms out of 97 are assigned the grade of $\star$, suggesting they are the heaviest discriminators against Black names. Fourteen firms are assigned the score of $\star\star\star$, which indicates that this group is the least-biased against Black applicants.

While the Condorcet ranks of the posterior means are highly correlated with the ranks of the bias estimates, heteroscedasticity makes the correlation less than perfect. For example, the firm with the sixth largest point estimate of bias (AKA firm \#6) has a Condorcet rank of 1 and the largest posterior mean, while the firm with the largest bias point estimate (AKA firm \#1) has a Condorcet rank of 2 and the second largest posterior mean. This rank reversal reflects that firm \#6 has a larger standard error than firm \#1, which leads the EB posterior mean to apply more shrinkage of the point estimate towards the center of the distribution. 

Appendix Figure \ref{fig:grade_frontiers} depicts the relationship between report card grades and firm-specific bias estimates and standard errors. Firms assigned the best grade of $\star\star\star$ tend to have both small contact gap estimates and standard errors, while firms assigned the grade $\star\star$ range widely in their standard errors but have modest contact gap estimates falling uniformly below 0.2. Firms assigned the worst grade of $\star$ exhibit very large contact gap estimates and widely varying standard errors. Appendix Figure \ref{fig: binary_all_grades} depicts the grade assignments that result from different choices of $\lambda$.

Though we have used stars to represent the firm ranks, it is important to remember that these grades were designed to convey ordinal rather than cardinal information. One of us \citep{gu2023comment} has recently cautioned against focusing excessively on rankings without also considering absolute standards of conduct. There is nothing in our integer linear programming problem that guarantees a grade of $\star$ implies a particularly egregious level of discrimination. Conversely, there is nothing  that guarantees firms assigned a grade of $\star\star\star$ exhibit no bias against Black names. 

As it turns out, however, the grades assigned by our procedure yield groups of firms with large cardinal differences in contact gaps. The firms assigned the grade of $\star\star\star$ have an average posterior mean value of $\theta_i$ of 0.03, while the two firms assigned the worst grade exhibit posterior means indicating a 24\% penalty against Black names on average. Notably, both of these heavily discriminating firms are federal contractors.

Our past work \citep{kline2021systemic} found that federal contractors, who are subject to monitoring by OFCCP for compliance with equal employment laws, tend to be substantially less biased against Black names on average, which is consistent with a variety of other evidence on the causal effects of affirmative action provisions on hiring behavior \citep[e.g.,][]{mccrary2007effect,kurtulus2016impact,miller2017persistent}. Indeed, an early audit study of federal contractors by \cite{newman1978discrimination} found evidence of a systematic preference for Black over white applicants among such firms. It is somewhat surprising then that the Condorcet ranks indicate that the four most heavily discriminating firms are all federal contractors. This finding is, to some extent, a reflection of the fact that the vast majority of the firms in our sample of large employers are contractors (63 of 97). The mean Condorcet rank of federal contractors is 54 (with rank 1 showing the most bias against Black applicants) while the mean Condorcet rank of non-contractors is 42. 

Although a legal precedent for audit studies has yet to be established, a commonly applied standard in discrimination cases is the so called ``four fifths rule,'' described in the Uniform Guidelines on Employee Selection \citep{equal1978department} which state that
\begin{quote}
A selection rate for any race, sex, or ethnic group which is less than four-fifths (4/5) (or eighty percent) of the rate for the group with the highest rate will generally be regarded by the Federal enforcement agencies as evidence of adverse impact.   
\end{quote}
Our estimates suggest the contact rates for fictitious applicants in our experiment may have violated this standard.

\subsection{Incorporating industry effects}

As a result of the substantial variation in contact penalties across industries, a report card that incorporates industry information is substantially more informative. This can be seen in Figure \ref{fig:poisson_binary_industry}, which displays discrimination report card results based on a model with industry effects. Adding industry information while maintaining the preference parameter $\lambda$ at 0.25 yields a report card with four grades rather than three. The number of firms assigned the worst grade of $\star$ increases from two to five, while nine firms are now assigned the second-worst grade $\star\star$. Eleven firms are assigned the best grade $\star\star\star\star$. Appendix Figure \ref{fig: binary_irfe_all_grades} depicts the grade assignments that result from different choices of $\lambda$.

The average value of the posterior mean $\bar \theta_i$ among the firms assigned the grade $\star$ is 0.22, indicating an expected bias against Black names in this group of 22\%. In contrast, the average value of $\bar \theta_i$ among the eleven firms assigned grade $\star\star\star\star$ is 0.03, suggesting a negligible effect of race on callback outcomes in this group. This finding indicates that many large firms are nearly unbiased, an important possibility result for companies seeking to improve the fairness of their recruiting process.

The small number of grades generated by our report card procedure explain a substantial portion of the total variation in discrimination across employers, especially when we incorporate industry. To summarize the explanatory power of the grades, we again utilize the grade-average posterior means as detailed in the Appendix. The variance estimate is weighted by the number of firms per grade, so that the ratio of between-grade to total variance has an $R^2$ interpretation. The estimated between-grade standard deviation in contact penalties is \binarybtwnvar \ for the three grades reported in Figure \ref{fig:poisson_binary}, implying an $R^2$ of roughly \binarybtwnr\%. Adding industry boosts the $R^2$ to \industrybtwnr\%. In other words, the four categories displayed in Figure \ref{fig:poisson_binary_industry} explain nearly half of the variance in discrimination across the 97 companies in our experiment.

Our ranking procedure allows us to grade the conduct of entire industries in addition to individual firms. Figure \ref{fig:btwn_binary} plots posterior estimates $\mathbb{E}[\eta_k|Y=y]n_k^{-1}\sum_{i:k(i)=k}{s_i^{\hat \beta}}$ of industry mean contact penalties. The most biased industry is estimated to be SIC 55, ``Auto dealers / services,'' with a posterior mean contact penalty of 25\%. In an industry grading scheme with $\lambda=0.25$ (which yields three total grades), SIC 55 is assigned its own unique grade of $\star$, indicating that this industry can be distinguished as the most biased in the experiment with high confidence. A group of six industries receive the best grade of $\star\star\star$, all of which have posterior mean contact gaps of roughly 5\%. The role of common industry-level practices in generating the stark differences between these low- and high-performers is an interesting topic for further inquiry.

\subsection{Grades under square-weighted loss}

\ref{appdx:sqwt} details the firm rankings derived from minimizing our square-weighted notion of risk $\mathcal{R}^2(d;\lambda)$. Under this grading scheme, misrankings involving small differences in discriminatory conduct between firms yield negligible losses, which makes a more aggressive partitioning of the firms optimal. Choosing $\lambda=0.25$ yields a six grade ranking without industry effects and an eight grade ranking when industry affiliation is taken into account. Because more grades are employed under square-weighting than under our default scheme, the variance in contact penalties explained by the grades is somewhat larger than for the grades reported in Figures \ref{fig:poisson_binary} and \ref{fig:poisson_binary_industry}. As explained in the Appendix, the square-weighted grades nevertheless maintain tight control over the expected probability of mistakes involving misrankings of firms with large absolute differences in discriminatory conduct.

\subsection{Misclassification and Bayes Factors}

In addition to conveying substantial information about discrimination, our report card system limits the ranking errors generated by the resulting classification. Figure \ref{fig:poisson_DR} assesses the reliability of report card grades by reporting the lower-triangular matrix of estimated between-grade DRs in our baseline model that omits industry effects. Panel (a), for example, shows that 11\% of the firm comparisons across grades $\star$ and $\star\star$ are expected to be misordered. The Discordance Rate naturally declines when comparing non-adjacent grades. The expected share of misordered comparisons across grades $\star$ and $\star\star\star$ is below 1\%. Adjacent grades have $DR$'s between 11 and 14\%, implying Bayes Factors between 7 and 6, respectively. The discordance rate for non-adjacent grades of 0.8\% implies a Bayes Factor exceeding 11.

Our preferred report card  with industry effects limits misrankings between firms selected as high- and low-performers. Panel (b) of Figure \ref{fig:poisson_DR} summarizes the reliability of the grades obtained when conditioning on industry effects. Discordance Rates between adjacent grades range from 15\% to 17\%. DRs for grades separated by two categories are less than 4\%, and the DR between the worst grade ($\star$) and the best grade ($\star\star\star\star$) is 0.4\%. Combined with the results from Section 6.1, this implies that a comparison of the best- and worst-performers in Figure \ref{fig:poisson_binary_industry} isolates firms with large differences in discriminatory conduct while yielding a negligible chance of misclassification.

\section{Conclusion}

We have proposed a new Empirical Bayes method for ranking units based upon noisy measurements and used it to grade the discriminatory conduct of firms towards distinctively Black names in a large-scale correspondence experiment. The experiment is shown to contain a wealth of information about the relative conduct of firms: our most granular (Condorcet) grades taking into account industry affiliation yield an expected correlation with the true firm ranks of \industrylambdaone. These grades are noisy, however, resulting in (expected) mistakes in nearly one quarter of the $\binom{97}{2}=4,656$ possible pairwise firm comparisons.

A generalization of the Condorcet scheme based on a desired 80\% posterior certainty threshold for pairwise contrasts yields a report card with only four grades. These coarse grades turn out to be substantially more reliable than the Condorcet ranks, lowering the chances that a randomly sampled pair of firms is misordered to \industrydr\%. The four grades are also highly informative, offering an expected correlation with the true firm ranks of \industrytau. In addition to conveying information about the ranking of firm conduct, the grades capture important differences in conduct levels. Firms assigned the worst grade favor white applicants over Black applicants by 23\% while those assigned the best grade favor white applicants by only 3\%. 

The finding of negligible contact gaps in a large group of firms provides a possibility result for employers seeking to improve the fairness of their hiring processes. Recent research points towards centralization of hiring processes as a possible means of dampening bias in large organizations \citep{berson2020outsourcing,challeeffect}, a conjecture that aligns with findings in behavioral economics that snap judgments by individuals are especially susceptible to bias  \citep[e.g.,][]{agan2023automating}. Further corroboration of this view comes from \cite{miller2017persistent}'s finding that temporary exposure to the heightened scrutiny over HR practices accompanying federal contractor status has persistent effects on the composition of firm hires. 

We plan to eventually release the information in this report card to the public along with the identities of the companies in the experiment. Our hope is that this information will prompt corporate leaders to rethink whether their inclusive hiring goals are being met. Much work remains to establish which sorts of reforms to organizational practices can improve the fairness and efficiency of corporate recruiting efforts. Releasing these data for use by other researchers will hopefully accelerate the pace of research into strategies for subduing hiring discrimination.

\clearpage
\bibliographystyle{aea}
\bibliography{bibfile.bib}

\clearpage

\clearpage
\FloatBarrier
\section*{Figures}
\begin{figure}[ht!]
\vspace{-.5cm}
\centering
    \caption{Deconvolutions of name-specific contact rate estimates}
    \label{fig: name_g}
    \begin{tabular}{c}
    a) Variance-stabilized contact rates ($\sin^{-1} \sqrt{p_i}$) \\
    \includegraphics[width=.65\textwidth]{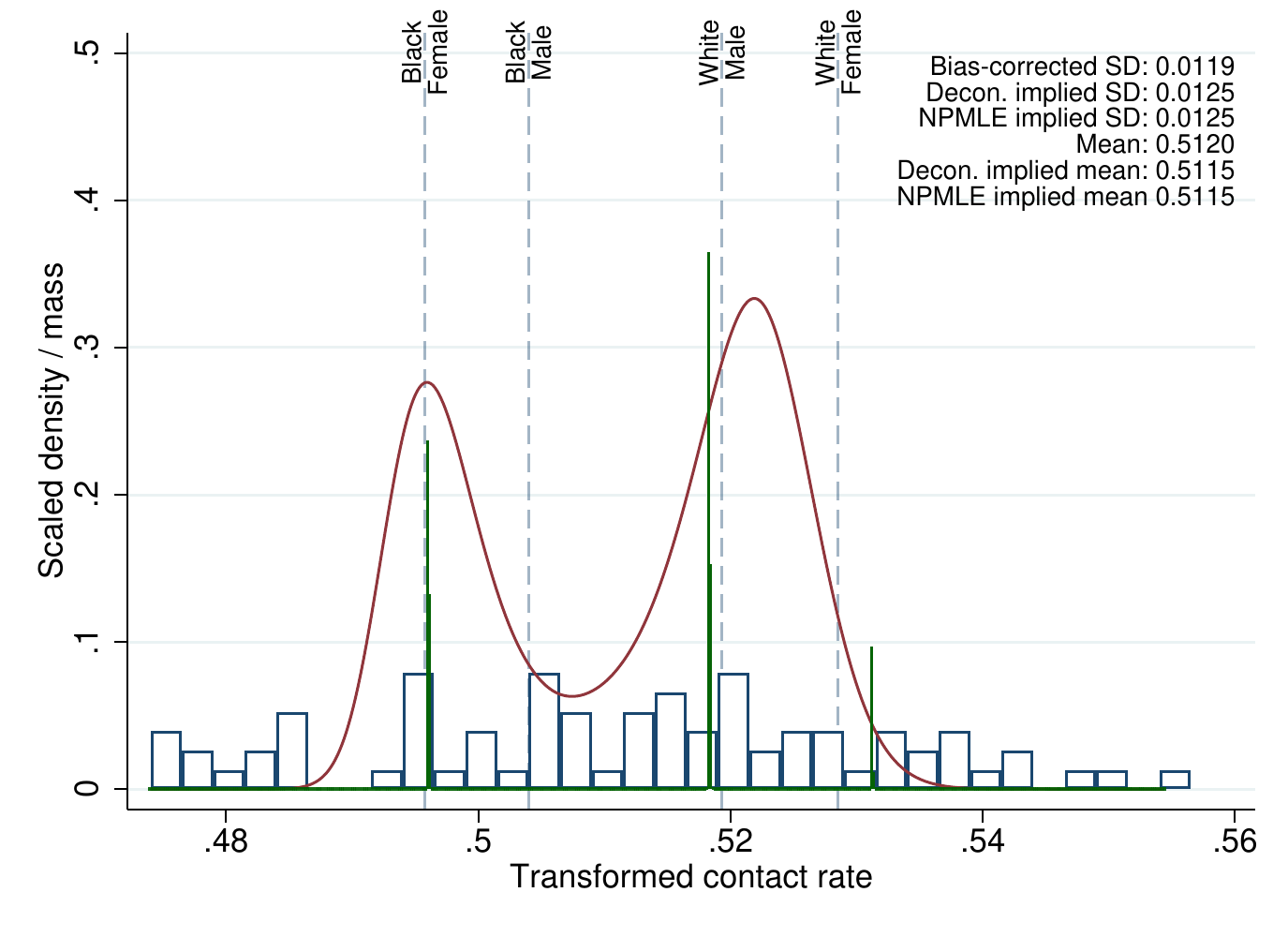} \\
    b) Contact rates ($p_i$) \\
     \includegraphics[width=.65\textwidth]{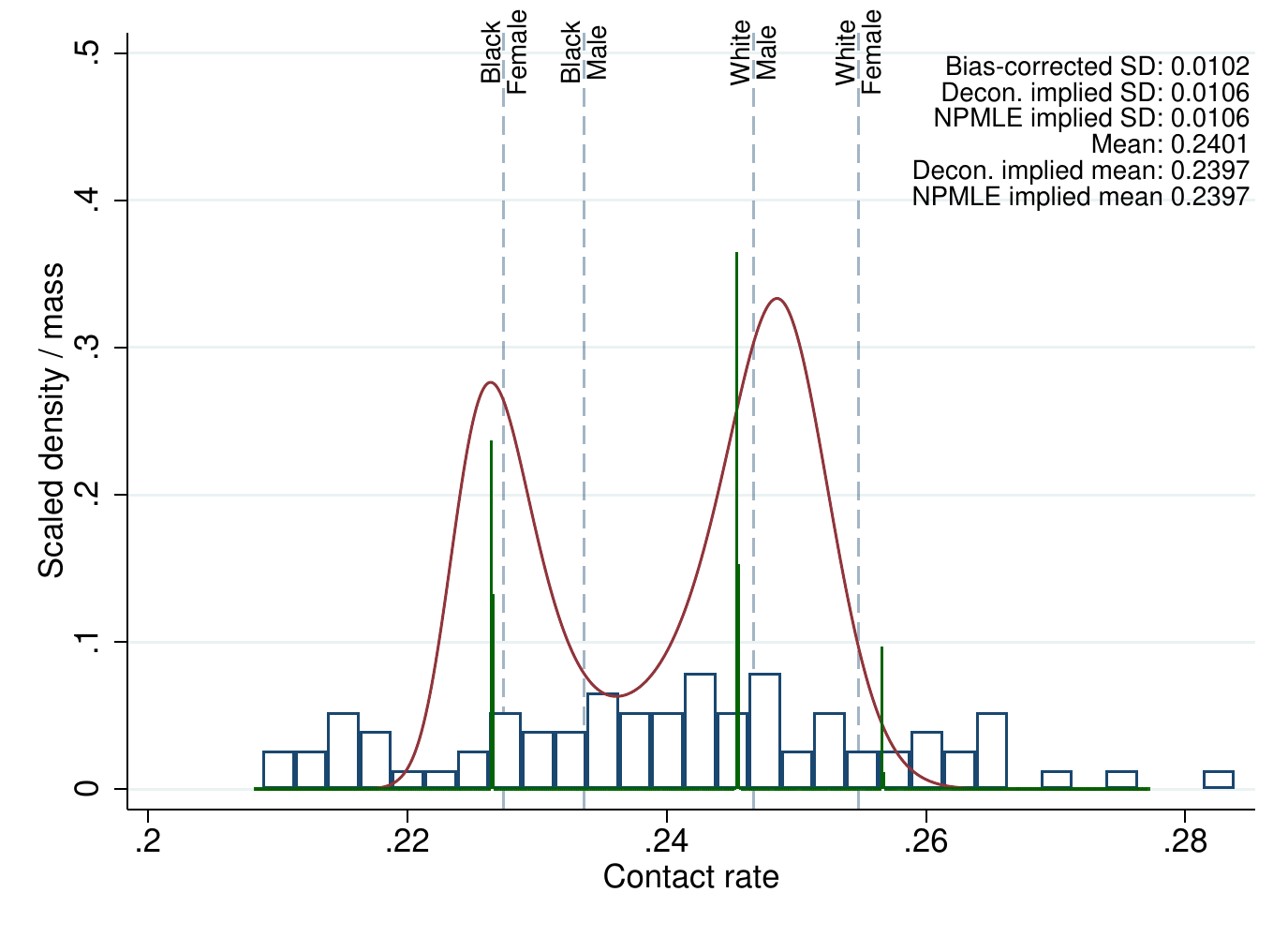}
     \end{tabular}
\parbox{\textwidth}{\small
\vspace{1eX}\emph{Notes}: This figure presents non-parametric estimates of the distribution name-specific contact rates. Panel (a) deconvolves transformed contact rates $\hat{\theta}_i = \sin^{-1}\left(\sqrt{\hat{p}_i}\right)$, where $\hat{p}_i$ is the contact rate for applications sent with first name $i$. The hollow blue histogram shows the distribution of estimated variance-stabilized contact rates. The red line shows a deconvolution estimate of the population contact rate distribution. The deconvolution procedure parameterizes the log-density with a fifth-order spline, and the parameters are estimated by penalized maximum likelihood, with penalization parameter chosen to match the mean and bias-corrected variance estimate as closely as possible. The dark green mass points plot the distribution of population contact rates estimated by non-parametric maximum likelihood (NPMLE). The vertical dashed lines plot mean contact rates for each race and gender group of names. Panel (b) converts the estimated distributions of variance-stabilized contact rates into distributions of contact rates $p_i$.}
\end{figure}

\begin{figure}
\centering
\caption{Name ranking exercises}
    \label{fig: names_res1}
\begin{adjustbox}{center}
\begin{tabular}{c c}
a) Pairwise posterior contrasts & b) Grades and discordance  \\
\includegraphics[width=0.6\textwidth]{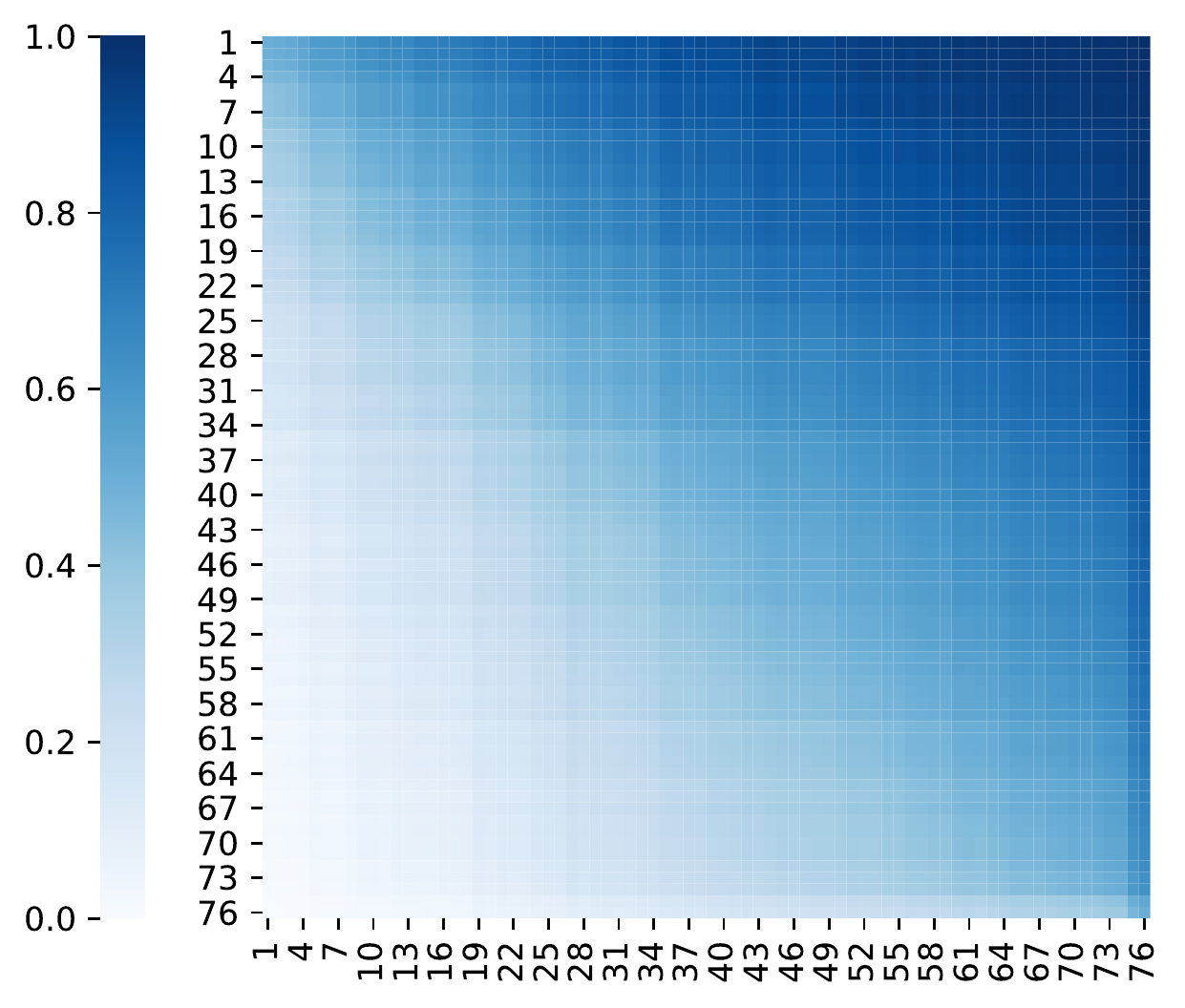} & \includegraphics[width=0.6\textwidth]{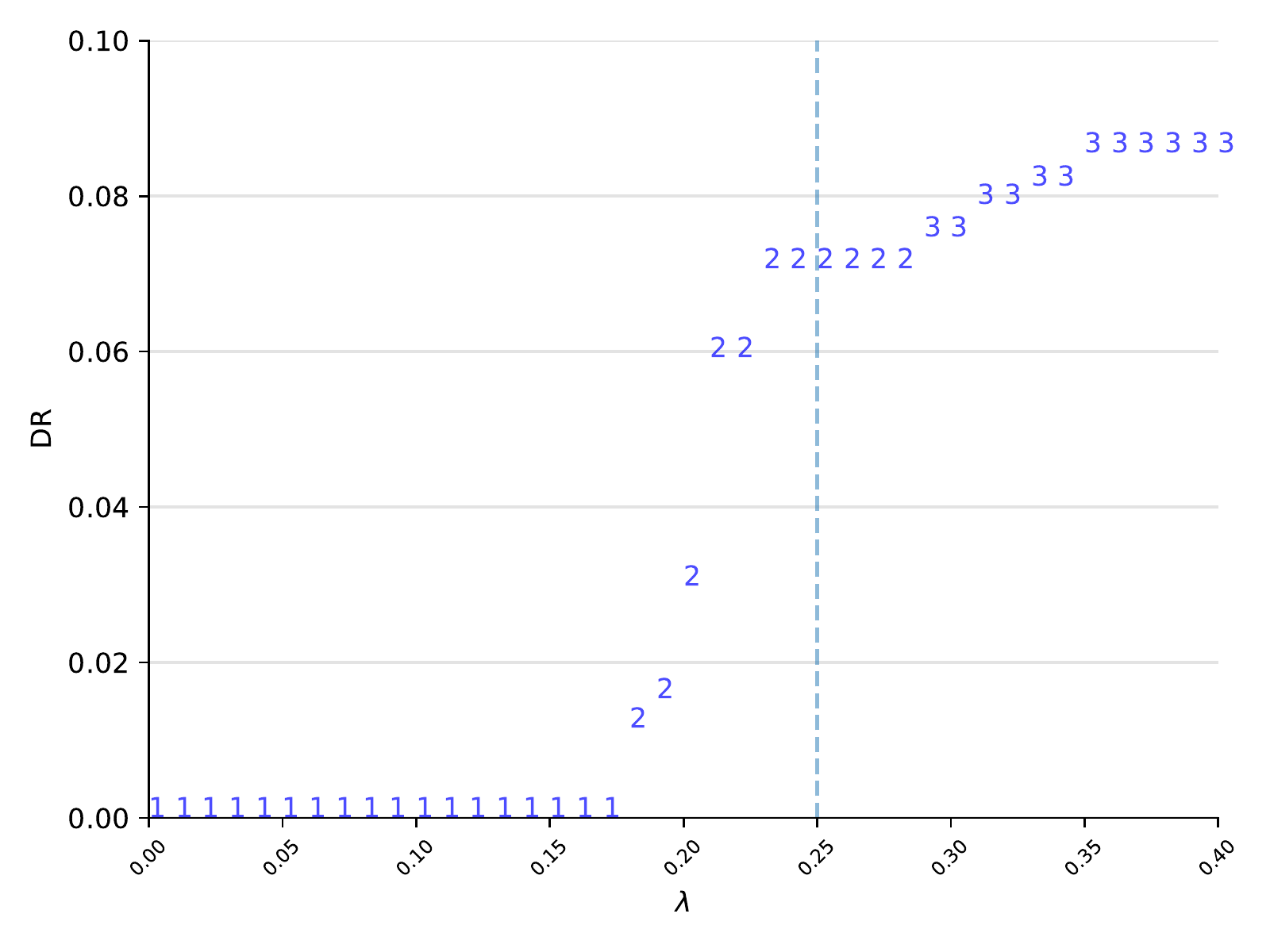} \\
\multicolumn{2}{c}{c) Reporting possibilities } \\
\multicolumn{2}{c}{\includegraphics[width=0.6\textwidth]{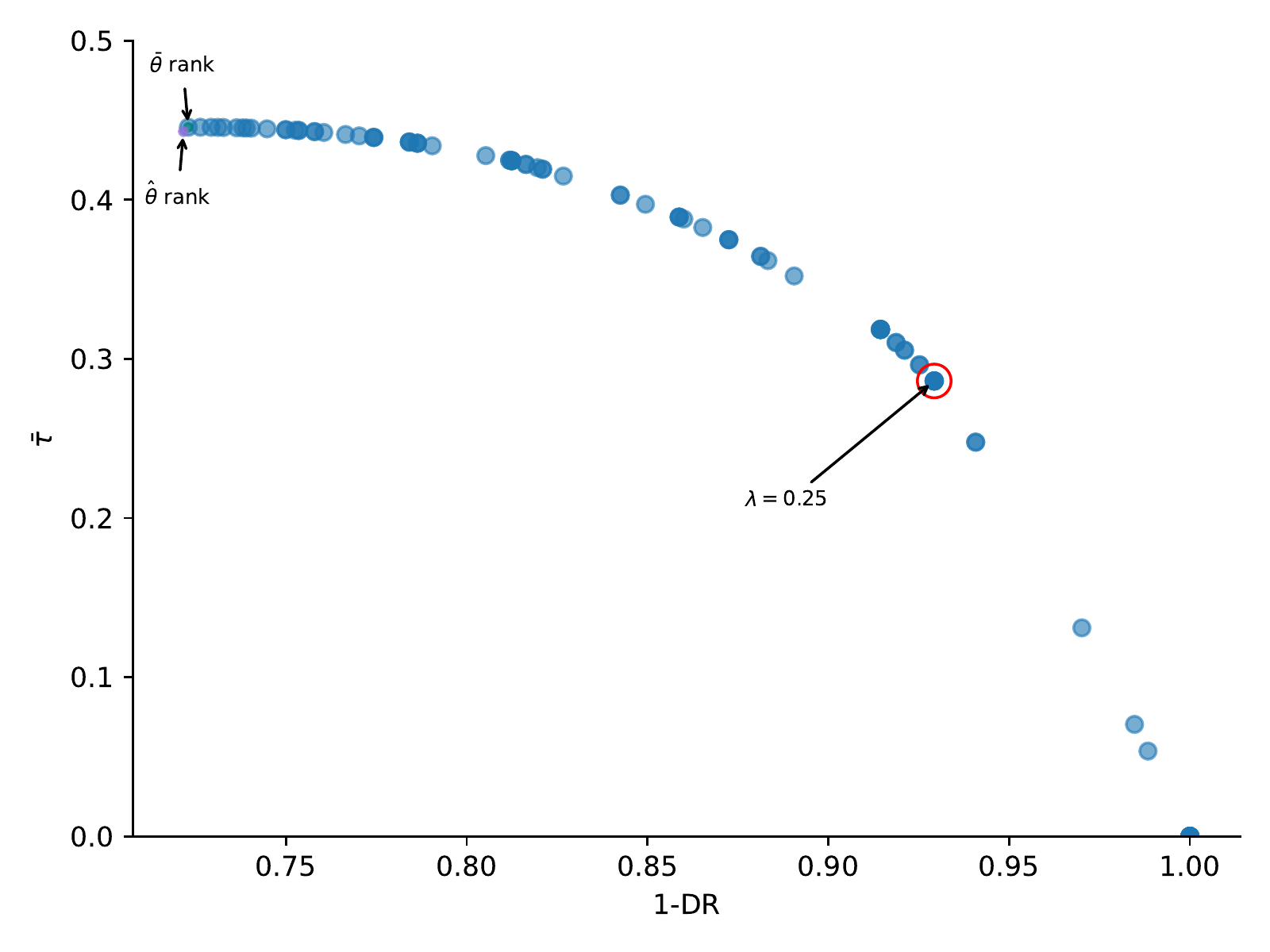}}
\end{tabular}
\end{adjustbox}
\parbox{\textwidth}{\small
\vspace{1eX}\emph{Notes}: This figure summarizes the results from grading contact rates for names. Panel (a) shows pairwise posterior ordering probabilities for all names. Posteriors are computed using the log-spline estimate plotted in Figure \ref{fig: name_g} as the prior. Names are ordered by their rank under $\lambda=1$. Shading indicates  the posterior probability that the contact rate for the name on the vertical axis exceeds the contact rate for the name on the horizontal axis. Panel (b) shows estimated Discordance Rates (DR) for an intermediate range of $\lambda$. Panel (c) plots the expectation of Kendall's $\tau$ rank correlation between true contact rates and grades against Discordance Rates (DR) for a range of grades indexed by $\lambda$. The red circle highlights the DR and expected $\tau$ corresponding to $\lambda = 0.25$. ``$\hat \theta$ rank'' refers to ranks based upon point estimates. ``$\bar \theta$ rank'' refers to ranks based upon Empirical Bayes posterior means.}
\end{figure}

\begin{figure}
\centering
\caption{Posterior means and grades of first names}
    \label{fig: names_ranks}
\begin{adjustbox}{center}
\includegraphics[width=\textwidth]{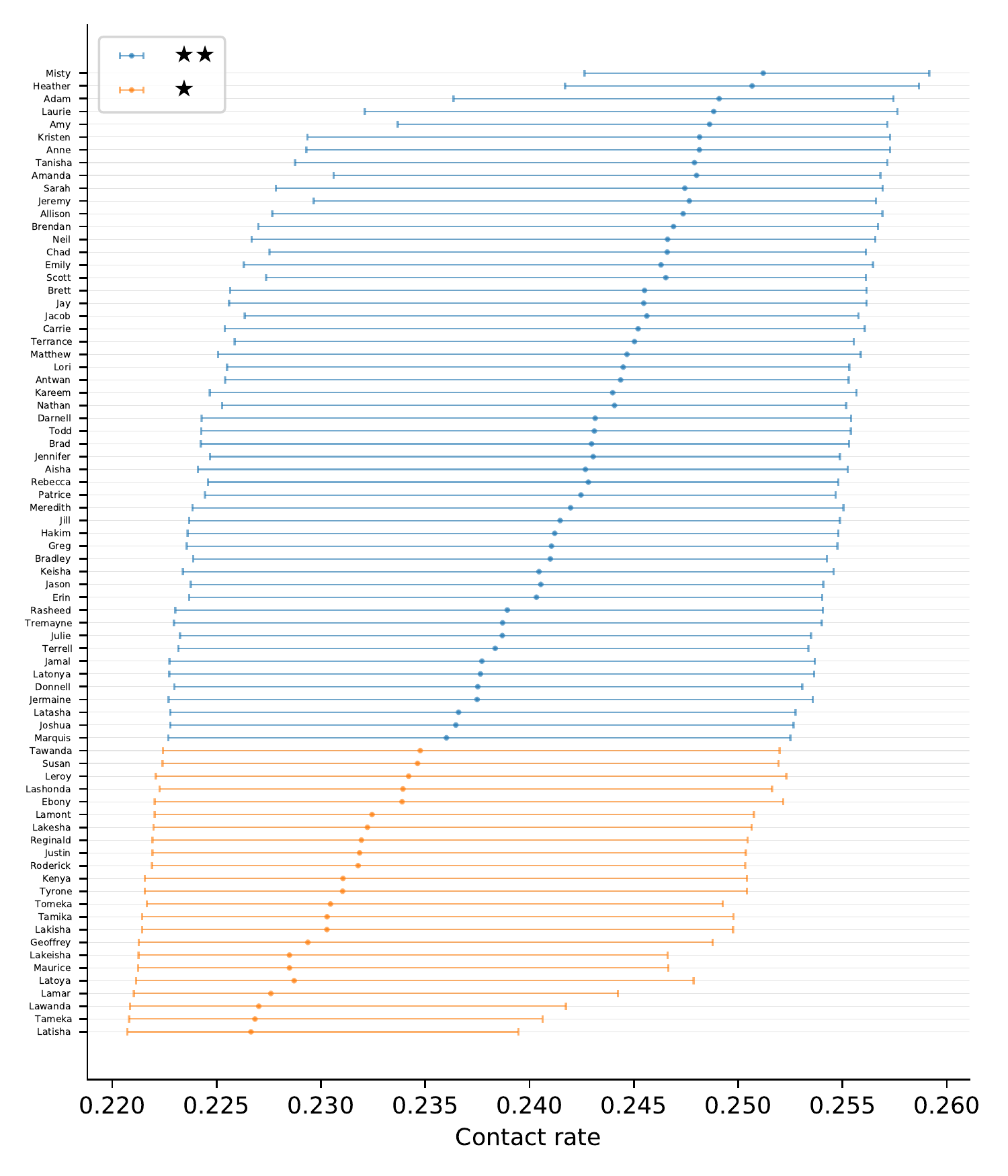}
\end{adjustbox}
\parbox{\textwidth}{\small
\vspace{1eX}\emph{Notes}: This figure shows posterior mean contact rates, 95\% credible intervals, and assigned grades for names. Results are shown for $\lambda = 0.25$, implying an 80\% threshold for posterior ranking probabilities. Names are ordered by their rank under $\lambda = 1$, when each name is assigned its own grade.}
\end{figure}

\begin{figure}[ht!]
    \centering
    \caption{Deconvolution estimates of firm-level contact penalty distribution}
    \includegraphics[width=.8\textwidth]{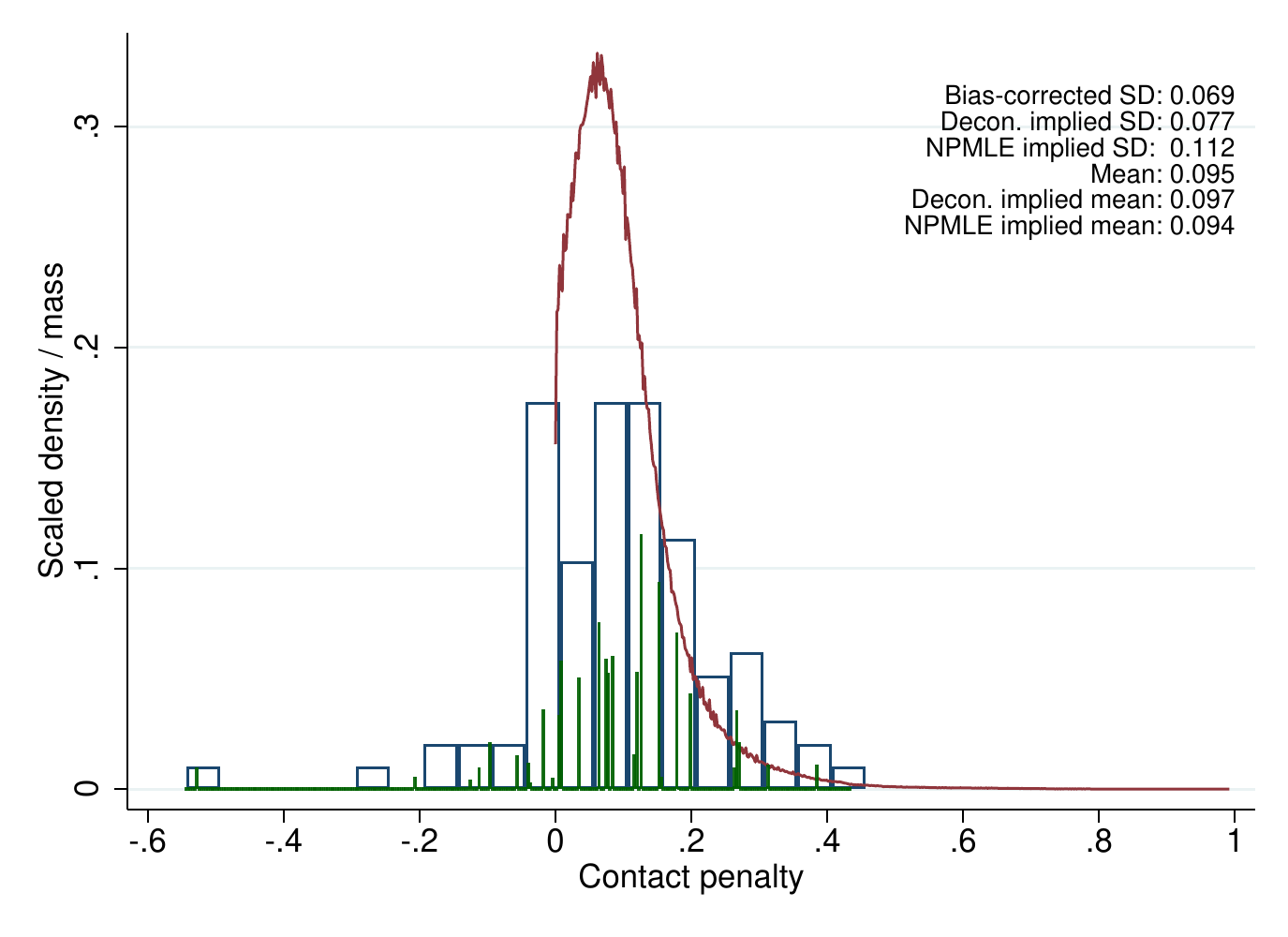}
    \label{fig: poisson_g}
\parbox{\textwidth}{\small
\vspace{1eX}\emph{Notes}: This figure presents non-parametric estimates of the distribution of firm-specific contact penalties. The blue histogram shows the distribution of estimated proportional contact penalties. The red line shows a  log-spline deconvolution estimate of the population contact penalty distribution. The dark green mass points plot a non-parametric maximum likelihood (NPMLE) estimate of the population contact penalty distribution. The bias-corrected standard deviation estimate is computed by subtracting the average squared standard error from the sample variance of estimated contact penalties, then taking the square root.}
\end{figure}

\begin{figure}[ht!]
    \centering
    \caption{Deconvolution estimates of between- and within-industry contact penalty distributions}
    \begin{tabular}{c}
    a) Standardized contact penalty distribution \\
    \includegraphics[width=.75\textwidth]{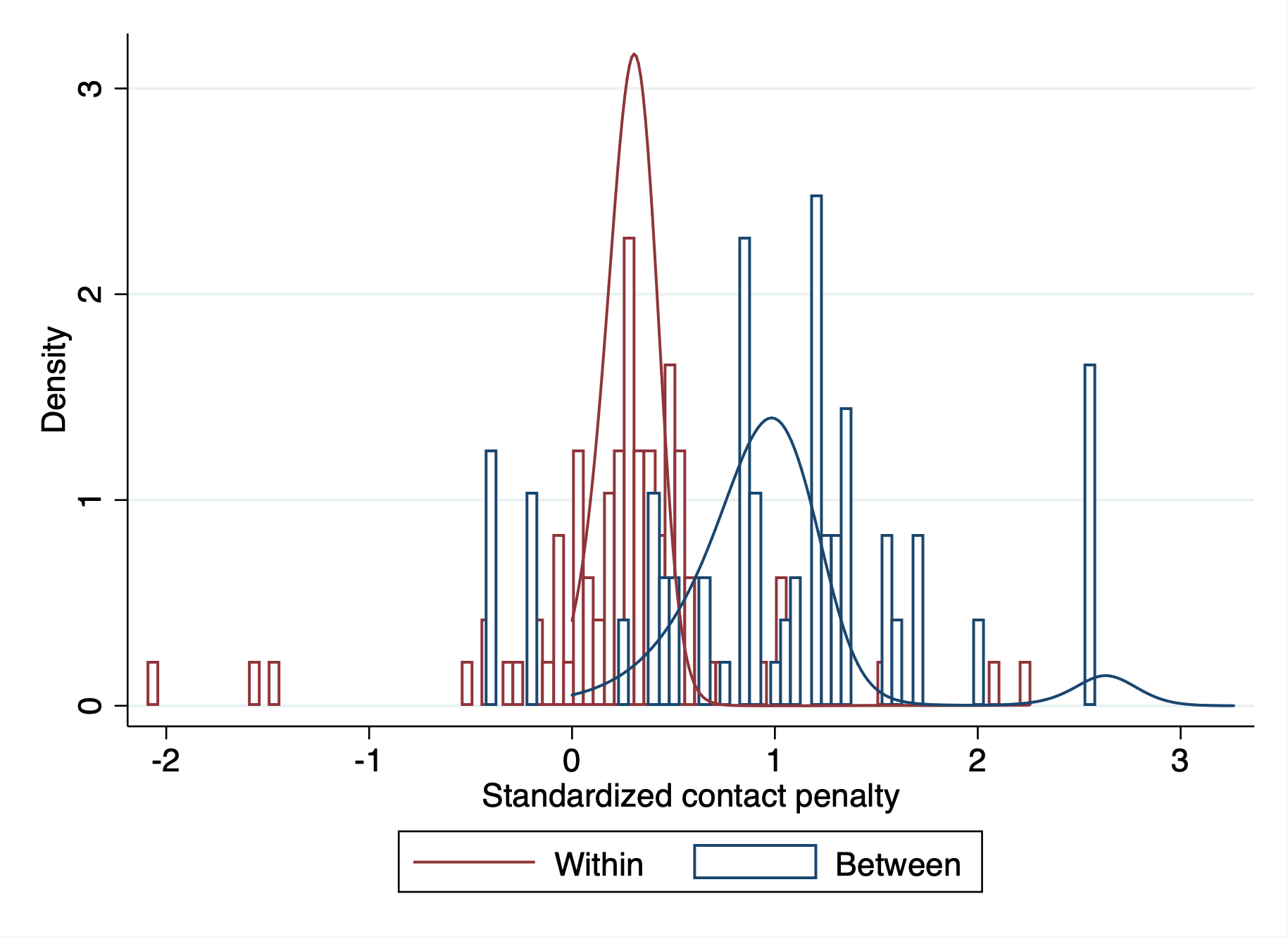} \\
  b) Marginal contact penalty distribution \\    \includegraphics[width=.75\textwidth]{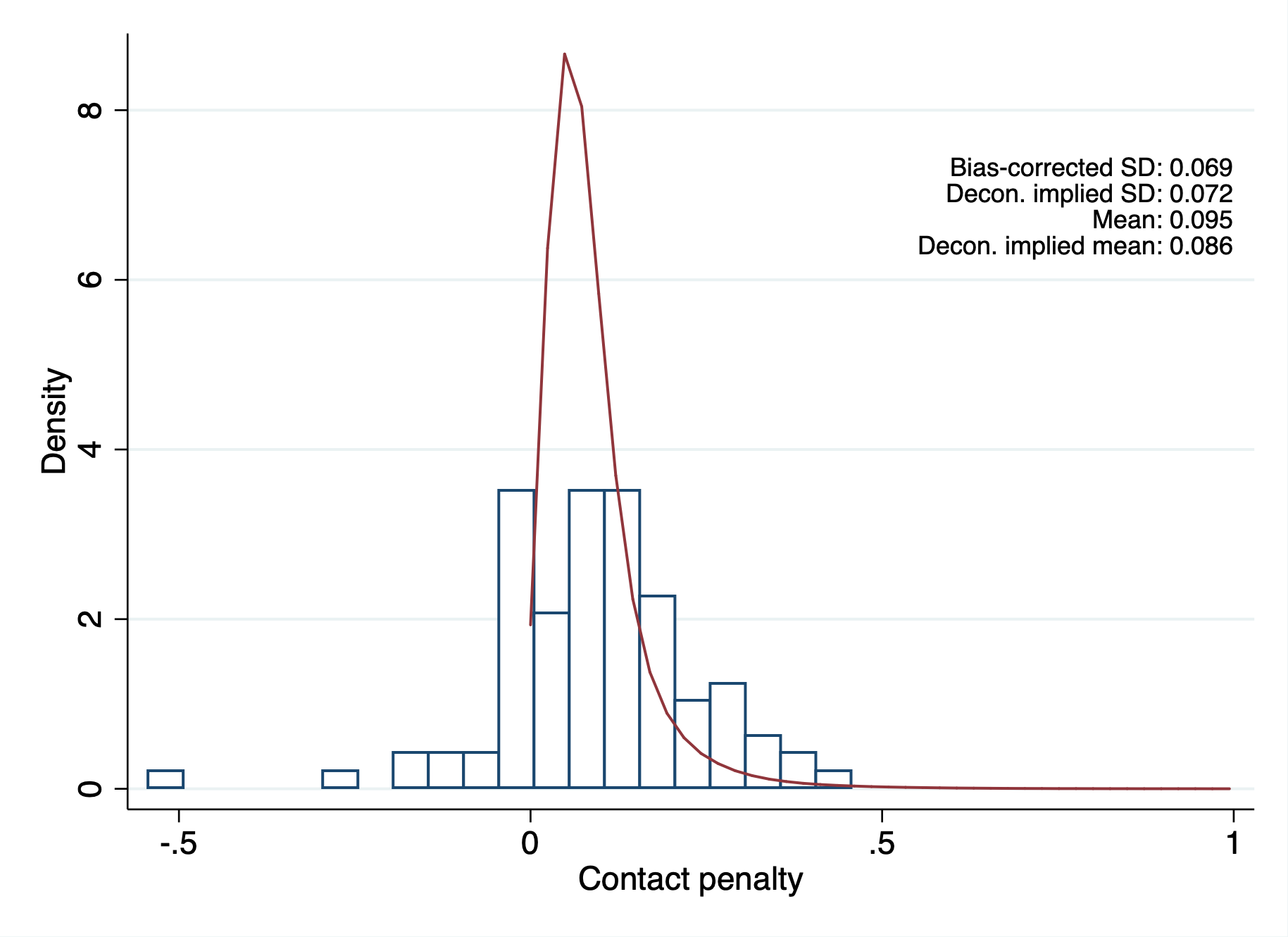}
    \end{tabular}
    \label{fig: industry_g}
\parbox{\textwidth}{\small
\vspace{1eX}\emph{Notes}: This figure presents estimates of the hierarchical random effects model described in Section \ref{sec: industry_rfe}. Panel (a) presents the estimated within- and between- industry distributions of standardized contact penalties. The blue histogram shows the distribution of $\hat{\eta}_{k}$, computed as the industry mean of $\hat{v}_{i}/\hat{\mu}_{v}$, where $\hat{v}_{i}=\hat{\theta}_{i}/s_{i}^{\hat{\beta}}$ and $\hat{\mu}_{v}$ and $\hat{\beta}$ are GMM estimates from Table \ref{tab:gmm_table}. The red histogram displays the distribution of $\hat{\xi}_{i}=\hat{v}_{i}/\hat{\eta}_{k(i)}$. Blue and red densities show corresponding log-spline deconvolution estimates of the distributions of $\eta_{k}$ and $\xi_{i}$. The deconvolution procedure parameterizes the log-density of each distribution with a fifth-order spline, and the parameters are estimated by penalized maximum likelihood, with penalization parameter chosen to match the GMM mean and variance estimates from Table \ref{tab:gmm_table} as closely as possible. Panel (b) shows the marginal distribution of contact penalties $\theta_{i}$ implied by the estimates from panel (a) along with a histogram of estimated contact penalties $\hat{\theta}_{i}$.  See \ref{sec: industry_fx} for complete details.}
\end{figure}

\begin{figure}[ht!]
\caption{Posterior contrasts}
\centering
\begin{tabular}{c}
    a) Baseline \\
    \includegraphics[width=.75\linewidth]{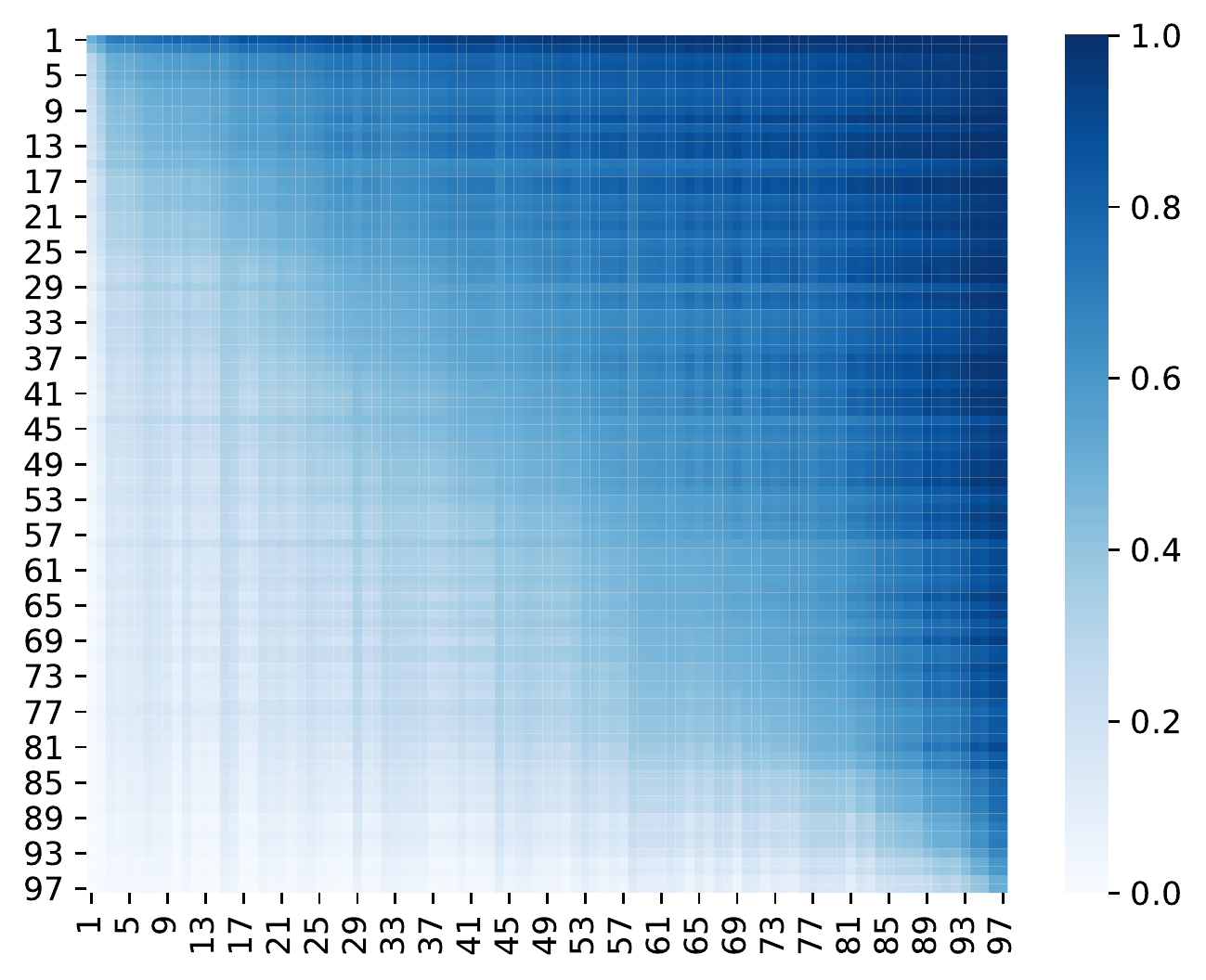}  \\
    b) Industry effects \\
  \includegraphics[width=.75\linewidth]{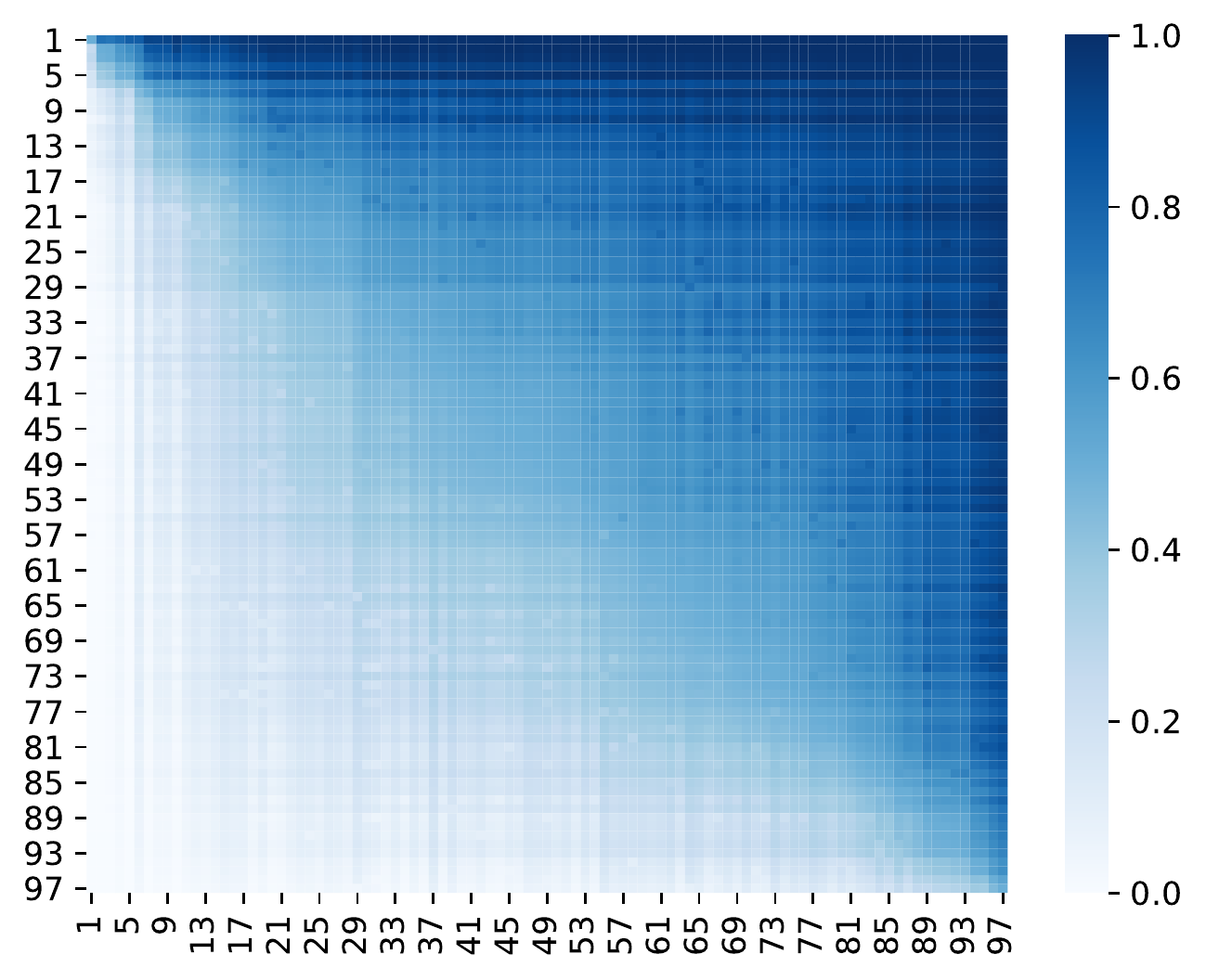}   
\end{tabular}
\label{fig: Pis}
\parbox{\textwidth}{\small
\vspace{1eX}\emph{Notes}: This figure plots pairwise posterior ordering probabilities for firm-specific contact penalties. Posteriors are computed using the log-spline estimates from Figure \ref{fig: poisson_g} as the prior distributions and Firms are ordered by their ranks under $\lambda=1$. The rank implying the largest $\theta_i$ is denoted by 1. Shading indicates the posterior probability that the contact penalty for the firm on the vertical axis exceeds the contact penalty for the firm on the horizontal axis.}
\end{figure}

\begin{figure}[ht!]
 \caption{Optimal pairwise rankings and global orderings}
 \centering
 \begin{tabular}{c}
 a) Baseline \\
  \includegraphics[width=0.75\textwidth]{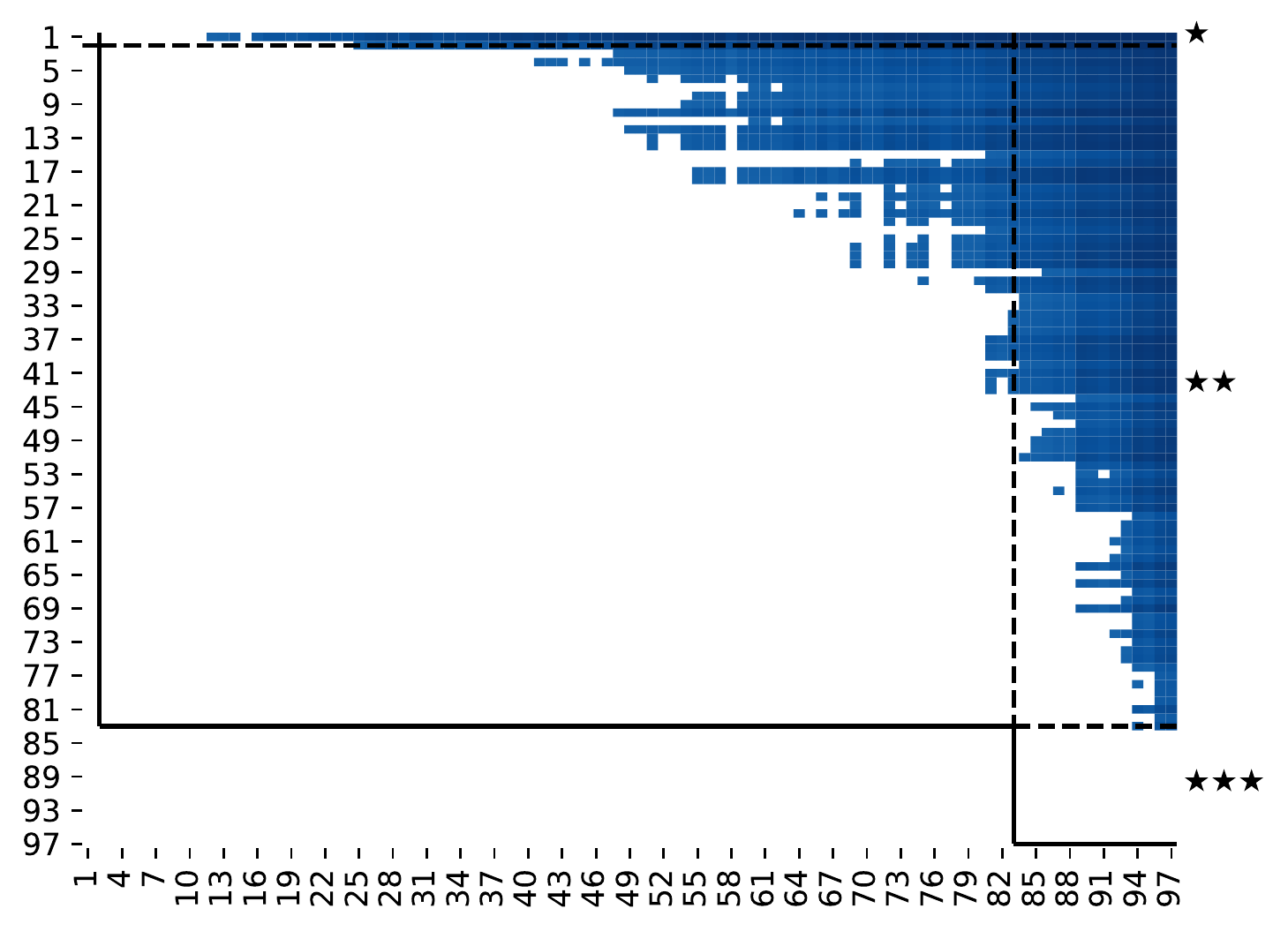} \\
  b) Industry effects \\
  \includegraphics[width=0.75\textwidth]{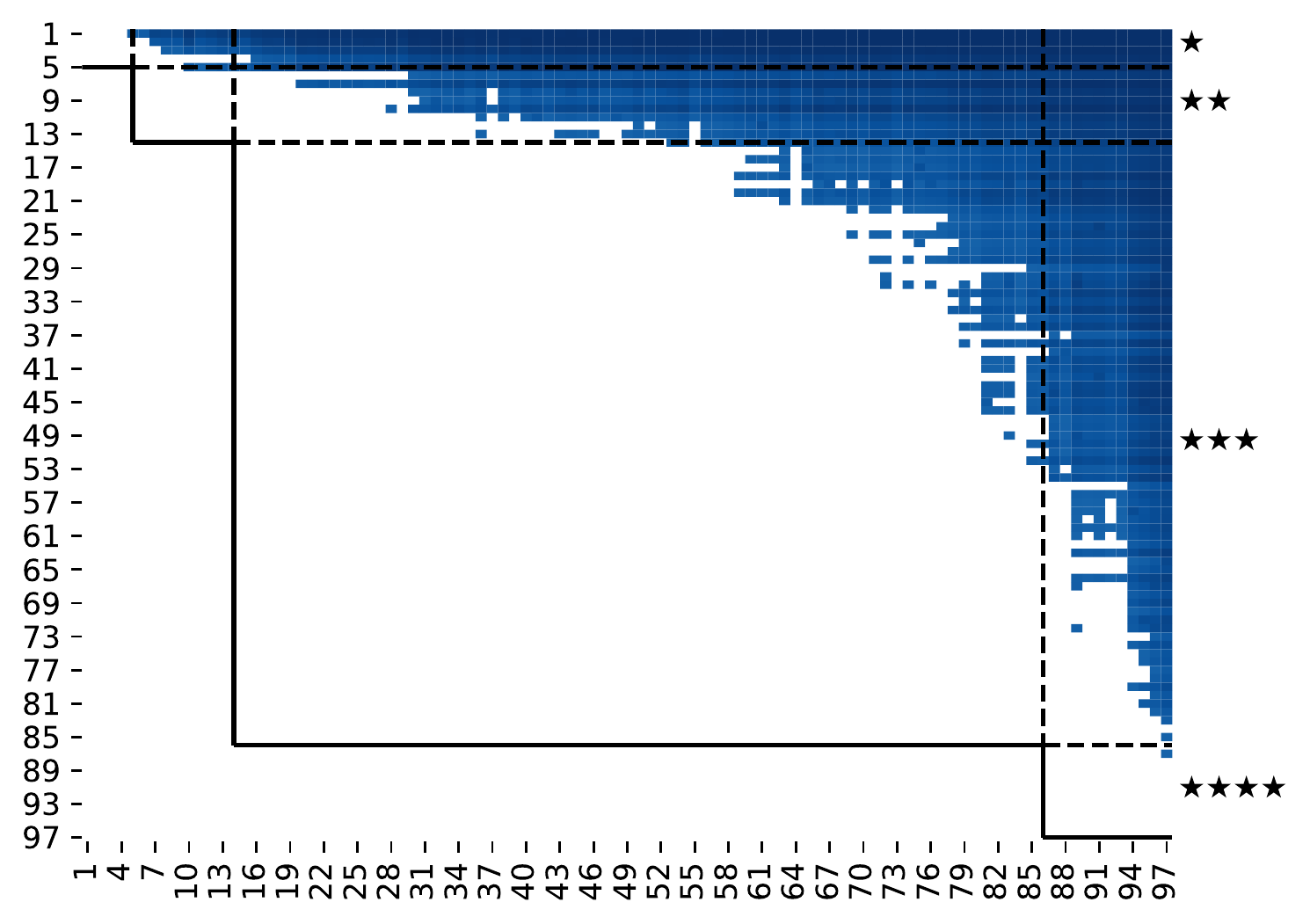} \\
  \label{fig: Pis_graded}
  \end{tabular}
\parbox{\textwidth}{\small
\vspace{1eX}\emph{Notes}: This figure plots the posterior ordering probabilities from Figure \ref{fig: Pis} for firm pairs where $\pi_{ij} > 1/(1+\lambda)$, indicating the pairwise optimal decision would rank the firm on the horizontal axis below the firm on the vertical axis. Both panels use $\lambda = 0.25$, implying an 80\% threshold for posterior ranking probabilities. The black lines denote the boundaries of optimal grades for this $\lambda$ for the firms in the rows. Panel (b) repeats the same exercise, but uses the industry random effect model to compute posteriors.}
 \end{figure}

\begin{figure}[ht!]
    \centering
    \caption{Grades and discordance as a function of $\lambda$}
    \includegraphics[width=\textwidth]{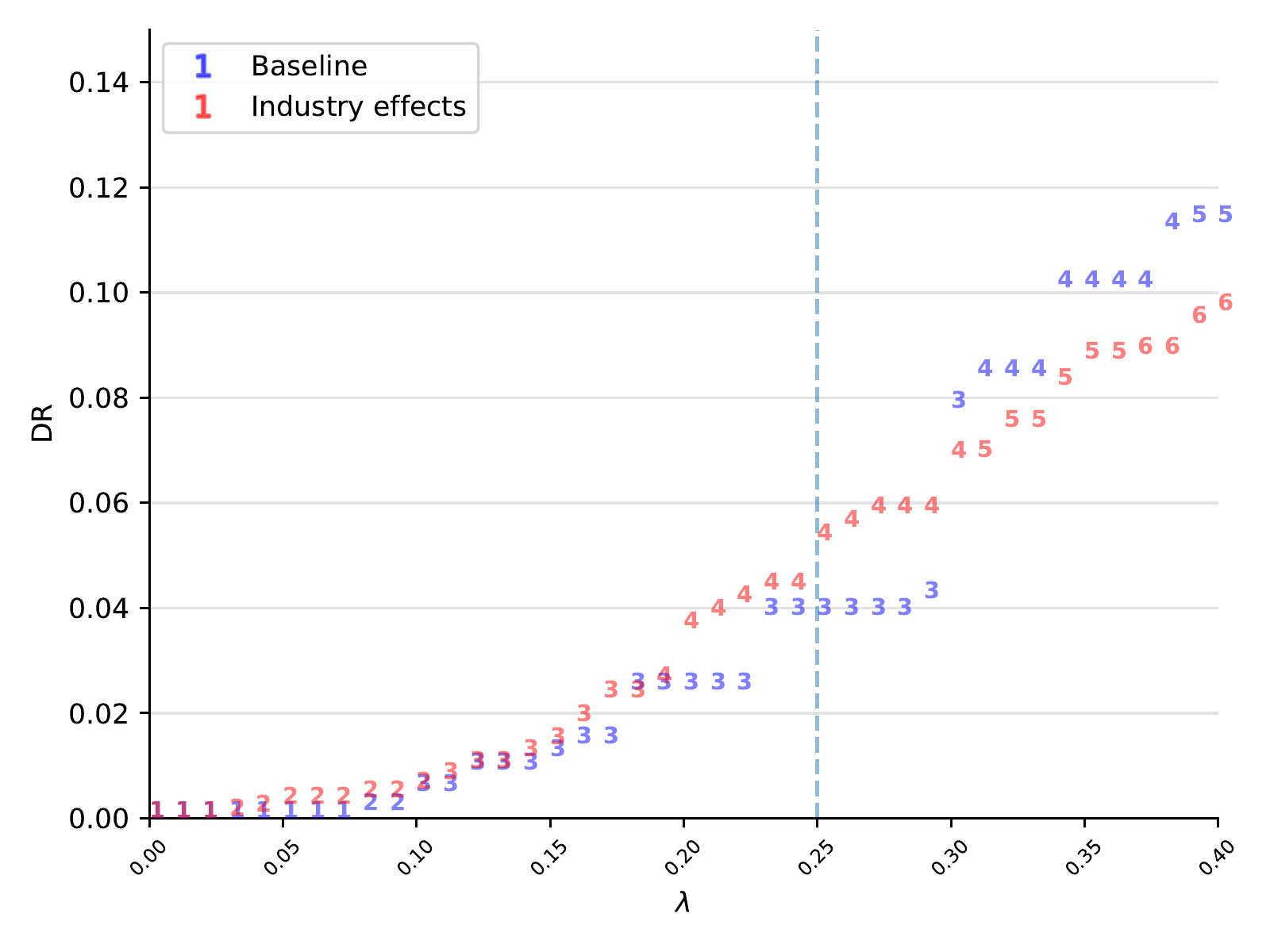}
    \label{fig:poisson_binary_lambda}
\parbox{\textwidth}{\small
\vspace{1eX}\emph{Notes}: This figure shows estimated Discordance Rates (DR) as a function of $\lambda$. The number on each point indicates the number of unique grades in the underlying grading scheme. The vertical dashed line shows results for the benchmark case of $\lambda=0.25$.}
\end{figure}

\begin{figure}[ht!]
    \centering
    \caption{Reporting possibilities}
\includegraphics[width=\textwidth]{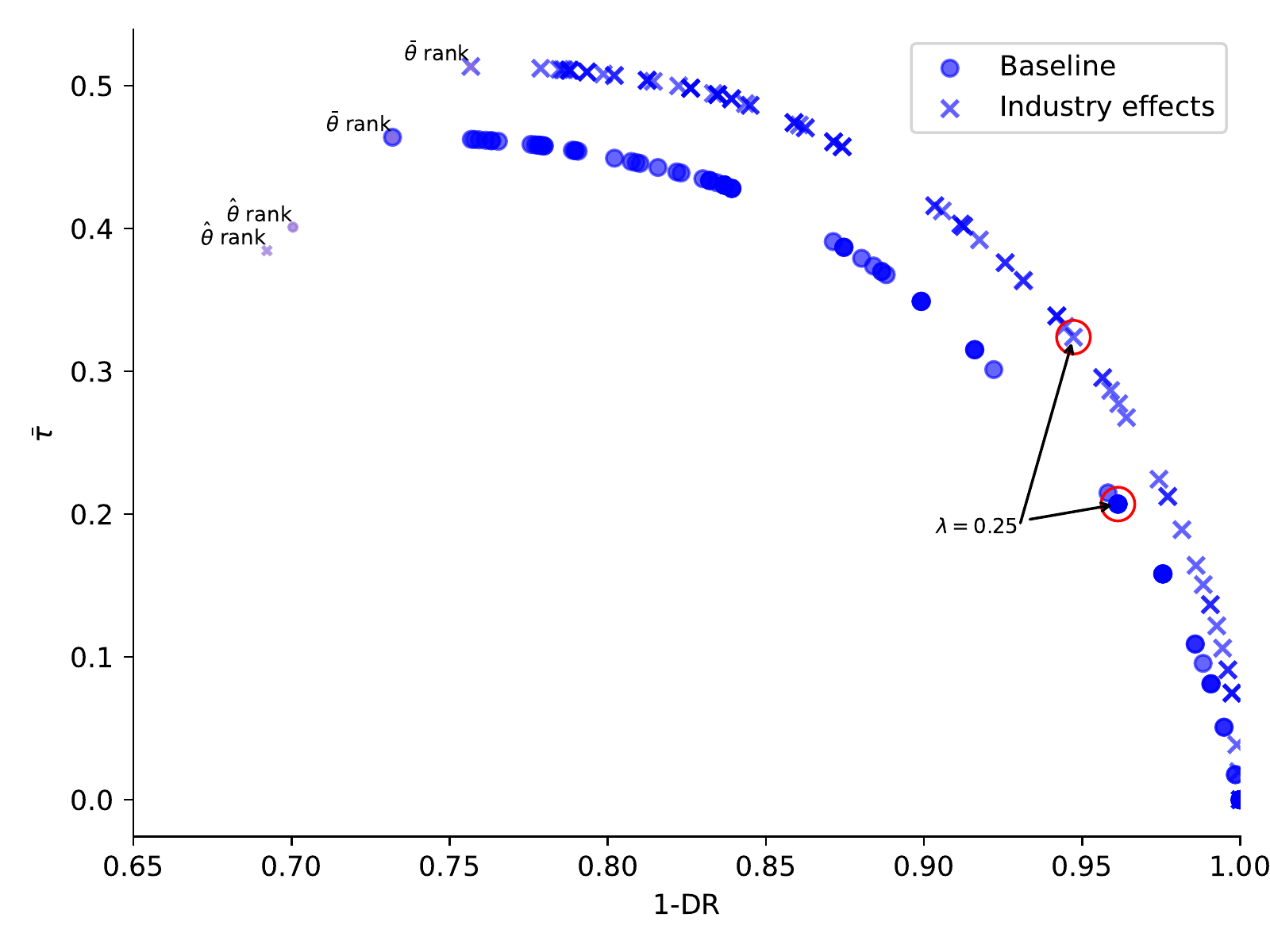}
    \label{fig:PPF}
\parbox{\textwidth}{\small
\vspace{1eX}\emph{Notes}: This figure shows the expectation of Kendall's $\tau$ rank correlation between $\theta$ and assigned grades (labeled $\bar{\tau}$) against Discordance Rates (DR) for a range of grades indexed by $\lambda$.  Red circles highlight the DR and $\bar{\tau}$ corresponding to $\lambda = 0.25$. ``$\hat \theta$ rank'' refers to ranks based upon point estimates. ``$\bar \theta$ rank'' refers to ranks based upon Empirical Bayes posterior means.}
\end{figure}

\begin{figure}[ht!]
    \centering
    \caption{Posterior means and grades of firms}
    \includegraphics[width=\textwidth]{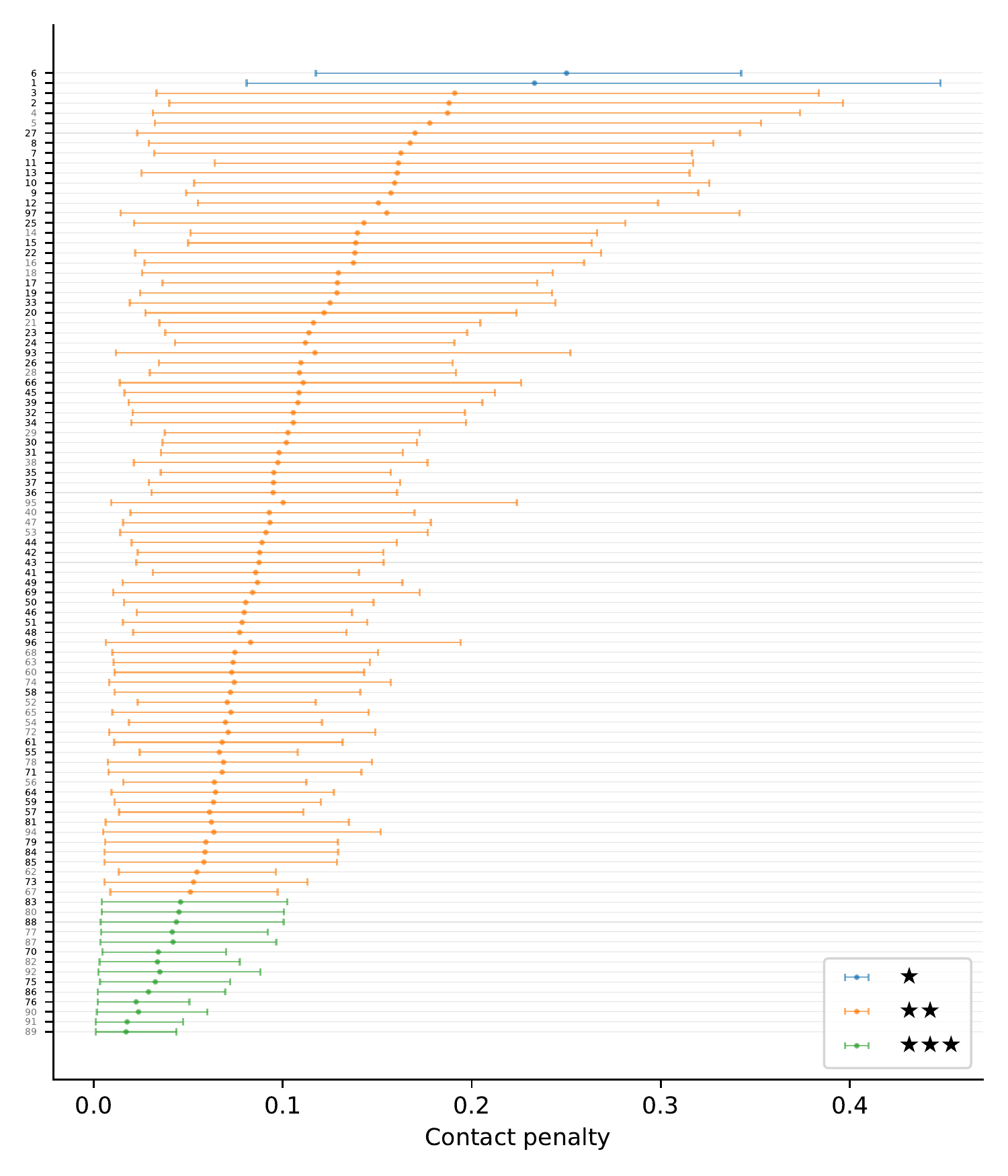}
    \label{fig:poisson_binary}
\parbox{\textwidth}{\small
\vspace{1eX}\emph{Notes}: This figure shows posterior mean proportional contact penalties, 95\% credible intervals, and assigned grades. Results are shown for $\lambda = 0.25$, implying an 80\% threshold for posterior ranking probabilities. Firms are ordered by their rank under $\lambda = 1$, when each firm is assigned its own grade, and labeled by their raw contact penalty rank, with \#1 showing the largest bias towards white applicants. Firms labeled with black text are federal contractors, whereas firms in gray are not.}
\end{figure}

\begin{figure}[ht!]
    \centering
    \caption{Posterior means and grades of firms (Industry effects model)}
    \includegraphics[width=\textwidth]{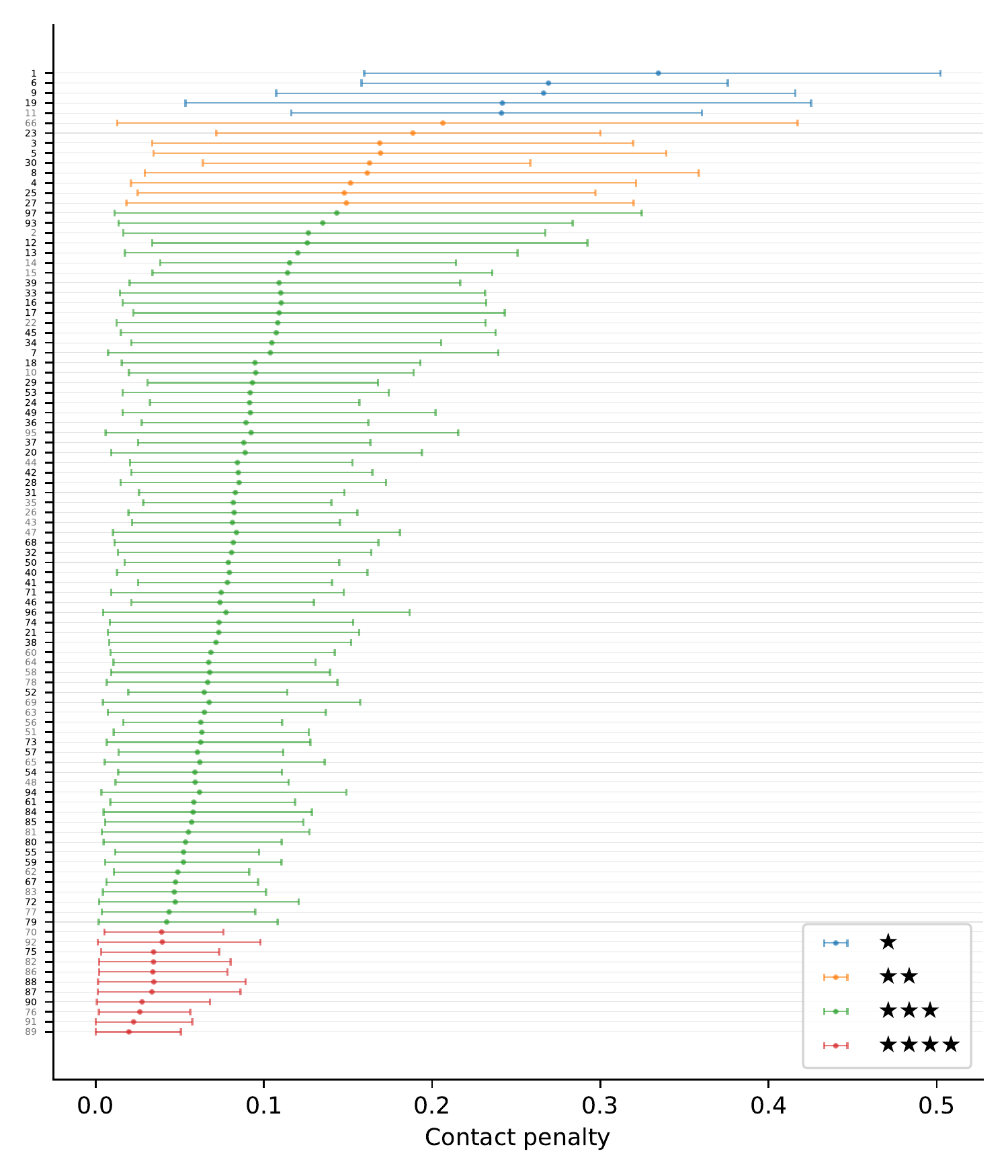}
    \label{fig:poisson_binary_industry}
\parbox{\textwidth}{\small
\vspace{1eX}\emph{Notes}: This figure shows posterior mean proportional contact penalties, 95\% credible intervals, and assigned grades from the industry random effect model. Results are shown for $\lambda = 0.25$, implying an 80\% threshold for posterior ranking probabilities. Firms are ordered by their rank under $\lambda = 1$, when each firm is assigned its own grade, and labeled by their raw contact penalty rank, with \#1 showing the largest bias towards white applicants. Firms labeled with black text are federal contractors, whereas firms in gray are not.}
\end{figure}

\begin{figure}[ht!]
    \centering
    \caption{Posterior means and grades of industries}
    \includegraphics[width=\textwidth]{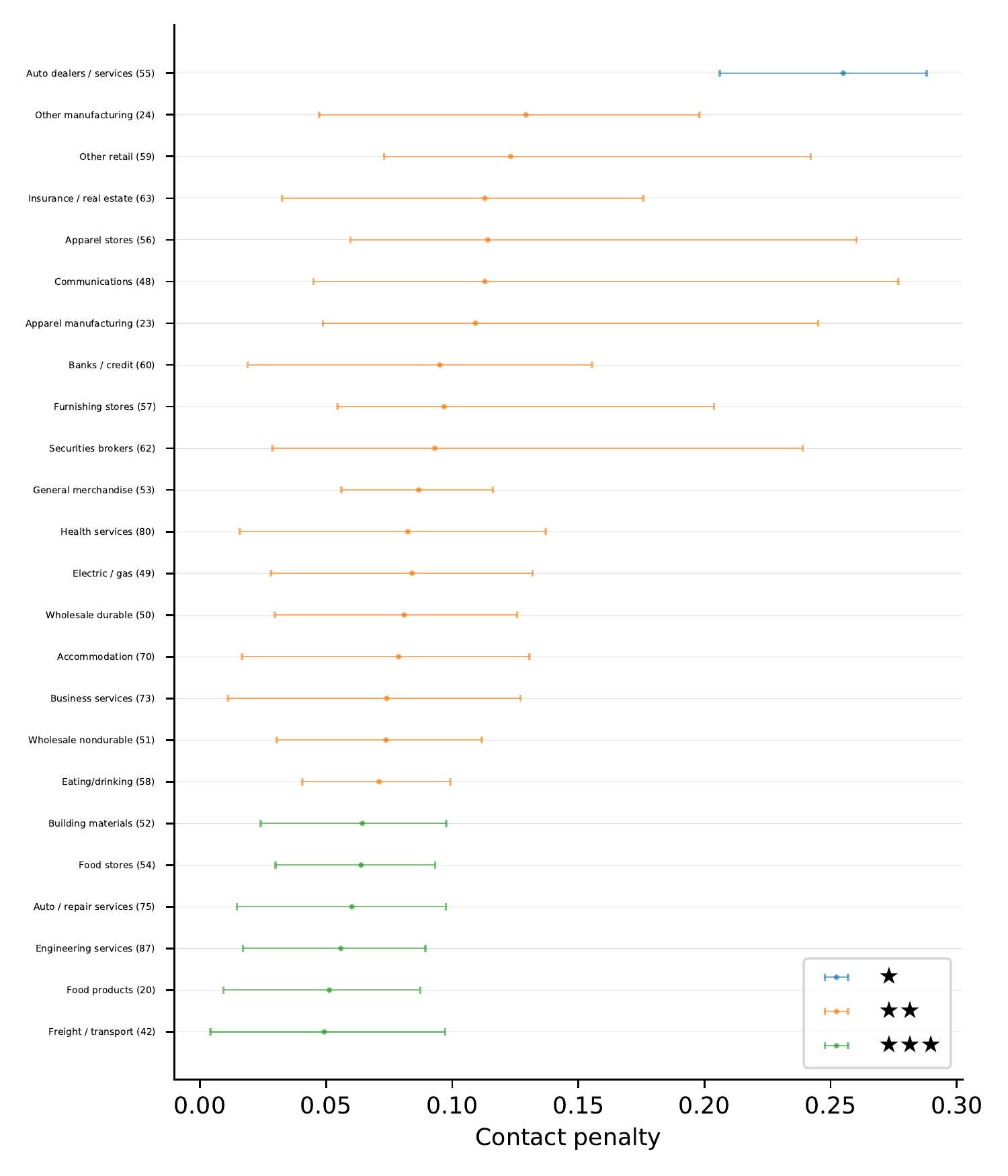}
    \label{fig:btwn_binary}
\parbox{\textwidth}{\small
\vspace{1eX}\emph{Notes}: This figure shows posterior means, 95\% credible intervals, and assigned grades for industry mean proportional contact penalties. Results are shown for $\lambda = 0.25$, implying an 80\% threshold for posterior ranking probabilities. Each industry is labeled by its name and two-digit SIC code.}
\end{figure}

\begin{figure}[!htbp]
    \centering
    \caption{DR in baseline and industry effects model}
    \begin{tabular}{c}
 a) Baseline \\
 \includegraphics[width=0.7\textwidth]{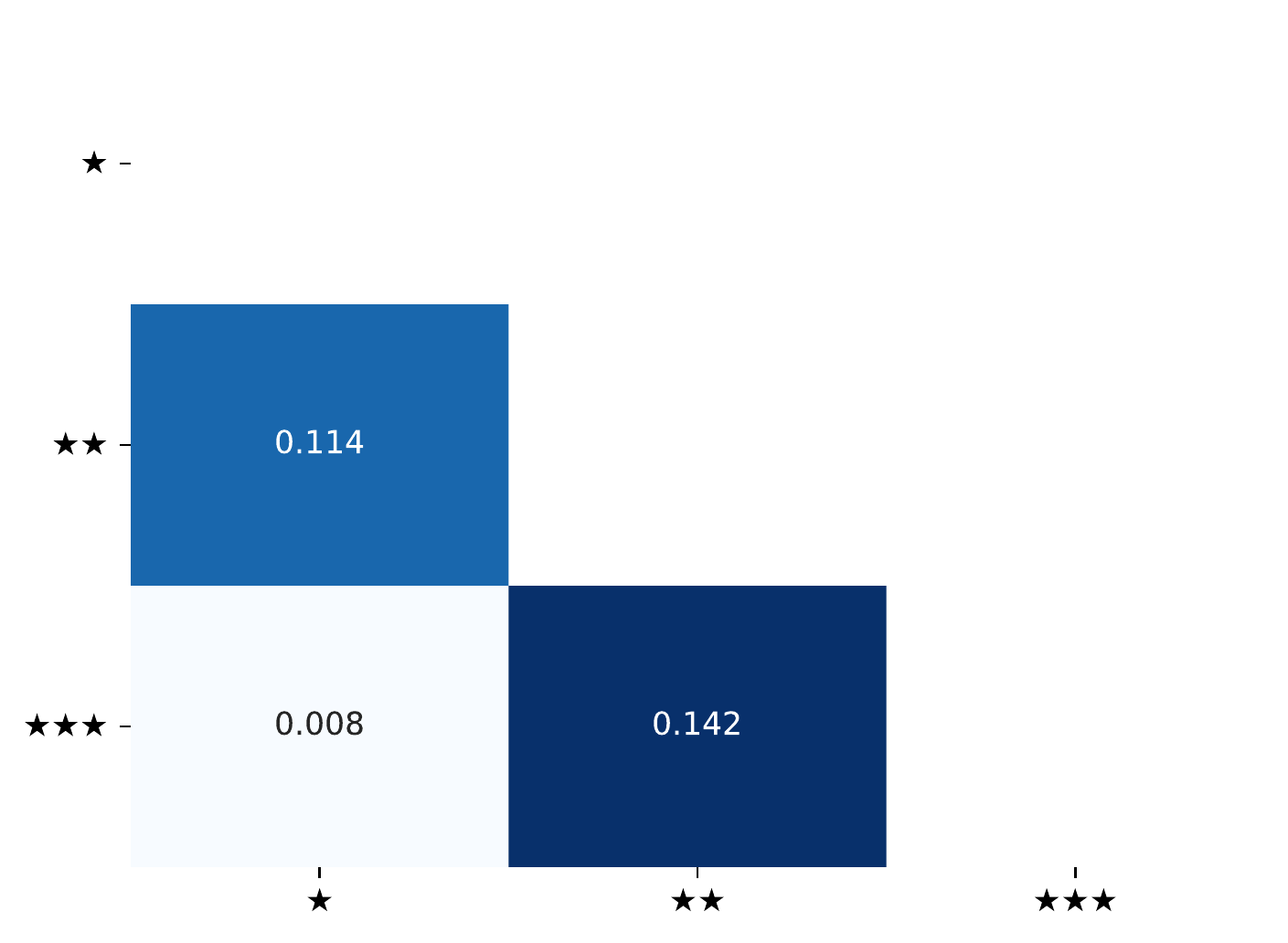} \\[0.5em]
b) Industry effects \\
\includegraphics[width=0.8\textwidth]{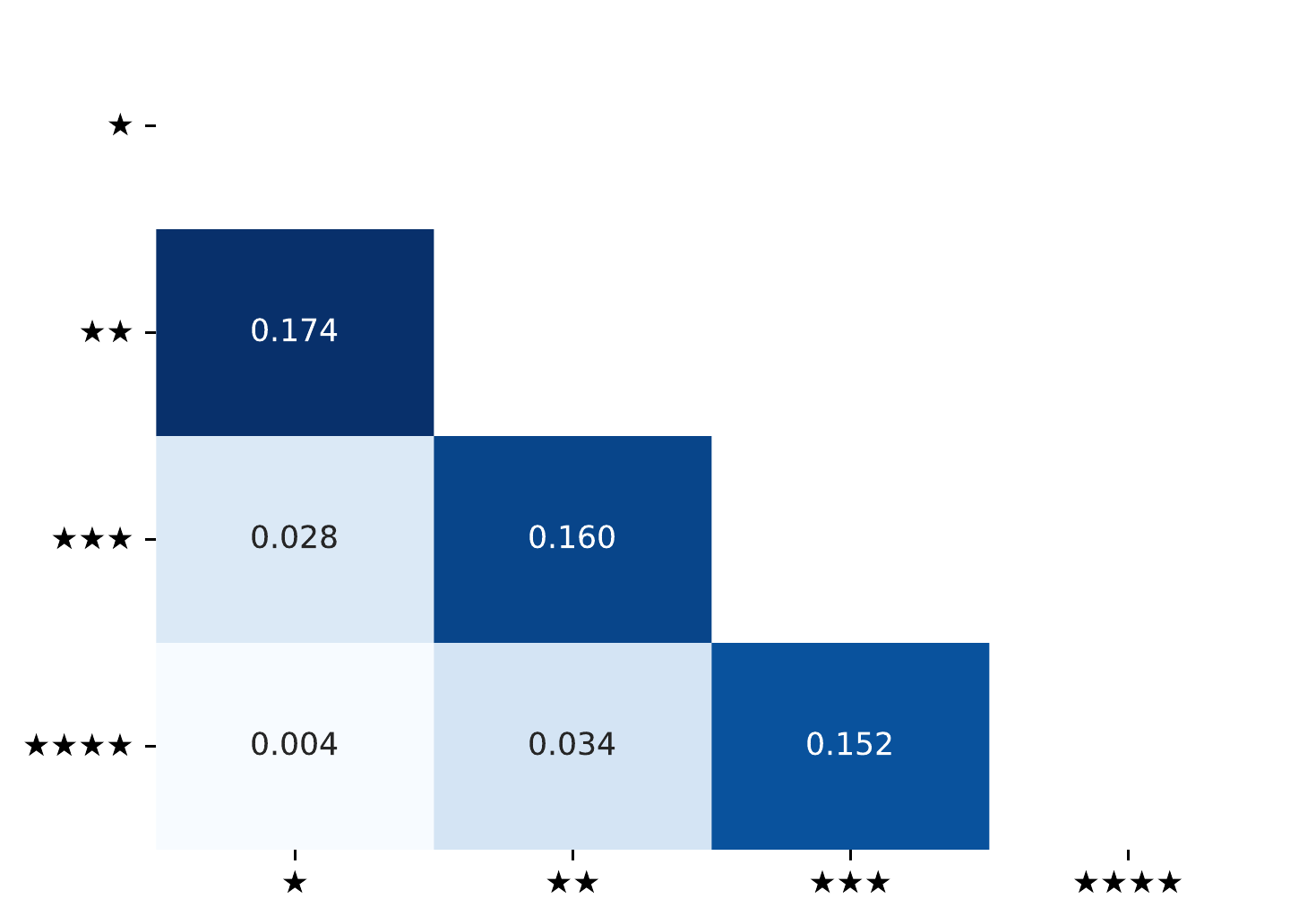}
\end{tabular}
    \label{fig:poisson_DR}
\parbox{\textwidth}{\small
\vspace{1eX}\emph{Notes}: This figure shows mean Discordance Rates (DR) across grade pairs for the baseline model and the model with industry effects. In both panels, DR$_{g,g'}$ is the expected share of pairwise comparisons between firms in grades $g$ and $g'$ where the grade rank differs from the latent firm rank. The upper-left most estimate in panel a, for example, indicates a 11\% chance that a random firm in grade $\bigstar$'s true rank is below a randomly chosen firm's in grade $\bigstar\bigstar$'s. In both panels, DR decays quickly when comparing non-adjacent grades.}
\end{figure}

\clearpage
\FloatBarrier
\section*{Tables}

\begin{table}[ht!]
\centering
\caption{Summary statistics for first names sample}
\label{tab:sumstats_names}
\begin{adjustbox}{center}
\begin{tabular}{l c c c c}		\hline\hline					
	&		&		& & Wald test of 	\\
    &    Contact rate	&	\# apps	&	\# first names & heterogeneity \\\cmidrule(lr){2-2}\cmidrule(lr){3-3}\cmidrule(lr){4-4}\cmidrule(lr){5-5}
	\\[-1em]
Male	&		&		&		\\[0.2em]
\quad Black	&	0.233	&	 20,927 	&	 19 & 12.6 	\\
	&	(0.003)	&		&		& [0.82] \\
\quad White	&	0.246	&	 20,975 	&	 19 & 15.8 	\\
	&	(0.003)	&		&	& [0.61]	\\[0.5em]
Female	&		&		&		\\[0.2em]
\quad Black	&	0.226	&	 20,879 	&	 19 & 21.2	\\
	&	(0.003)	&		&		& [0.24] \\
\quad White	&	0.254	&	 20,862 	&	 19 & 19.9	\\
	&	(0.003)	&		&		& [0.34] \\[0.5em]
Estimated contact rate SD \\[0.2em]
\quad Total & 0.010 \\[0.1em]
\quad Between race/sex & 0.011 \\[0.2em]
 \hline\hline
\end{tabular}													
\end{adjustbox}\\
\parbox{\textwidth}{\small
\vspace{1eX}
\textit{Notes:} This table presents summary statistics for the sample of applications used in the analysis of first names. The table presents the mean 30-day contact rate, total number of applications sent, and number of unique first names used for each race and sex combination. Contact rates are re-weighted to balance the distribution of names across experimental waves. Although Black and white names were sent in pairs during the experiment, the total number of applications across race groups is not identical because some jobs closed before both applications could be sent. The gender of the name assigned to each application was unconditionally randomized. The final column reports Wald tests for equality of contact probabilities across the first names in each demographic group. Under the null hypothesis of equal contact probabilities, each test statistic is distributed $\chi^2(18)$. Corresponding $p$-values are reported in brackets. The estimated contact rate SD is a bias-corrected estimate of the standard deviation of name-specific contact rates, computed by subtracting the average squared standard error from the sample variance of contact rate estimates then taking the square root. The between race/sex standard deviation is a corresponding bias-corrected estimate of the variation in mean contact rates across race and sex groups. See Appendix Table \ref{tab: firstname_list} for a list of first names used in the analysis.}
\end{table}

\begin{landscape}
\begin{table}[ht!]
\centering
\caption{Summary statistics for firm sample}
\label{tab:sumstats}
\begin{adjustbox}{center}
\begin{tabular}{l cccccccc} \hline\hline																					
	&		&		&		&	\multicolumn{4}{c}{Contact rates and gaps}													\\\cmidrule(lr){5-9}
	&	\# Firms	&	\# Jobs	&	\# Apps	&	White 		&	Black		&	Difference		&	Log dif		&	Mean SE	\\\hline \\[-0.7em]
All	&	97	&	 10,453 	&	 78,910 	&	0.256	 (0.004)	&	0.236	 (0.003)	&	0.020	 (0.002)	&	0.095	 (0.013)	&	0.104	\\[0.5em]
																					
2-digit SIC industry (code)																					\\[0.2em]
\quad Food products (20)	&	1	&	 100 	&	 788 	&	0.435	 (0.041)	&	0.440	 (0.040)	&	-0.005	 (0.019)	&	-0.011	 (0.045)	&	0.045	\\
\quad Apparel manufacturing (23)	&	2	&	 200 	&	 1,538 	&	0.205	 (0.026)	&	0.175	 (0.023)	&	0.031	 (0.012)	&	0.177	 (0.062)	&	0.088	\\
\quad Other manufacturing (24)	&	4	&	 375 	&	 2,904 	&	0.119	 (0.012)	&	0.104	 (0.012)	&	0.015	 (0.009)	&	0.179	 (0.116)	&	0.211	\\
\quad Freight / transport (42)	&	4	&	 458 	&	 3,300 	&	0.194	 (0.017)	&	0.197	 (0.017)	&	-0.003	 (0.008)	&	-0.014	 (0.039)	&	0.076	\\
\quad Communications (48)	&	2	&	 175 	&	 1,124 	&	0.273	 (0.036)	&	0.225	 (0.033)	&	0.048	 (0.021)	&	0.163	 (0.086)	&	0.120	\\
\quad Electric / gas (49)	&	3	&	 320 	&	 2,419 	&	0.261	 (0.021)	&	0.247	 (0.020)	&	0.014	 (0.010)	&	0.120	 (0.060)	&	0.094	\\
\quad Wholesale durable (50)	&	2	&	 152 	&	 1,143 	&	0.194	 (0.028)	&	0.177	 (0.027)	&	0.017	 (0.011)	&	0.088	 (0.057)	&	0.081	\\
\quad Wholesale nondurable (51)	&	11	&	 1,117 	&	 8,194 	&	0.299	 (0.011)	&	0.288	 (0.011)	&	0.011	 (0.007)	&	0.092	 (0.032)	&	0.091	\\
\quad Building materials (52)	&	3	&	 377 	&	 2,755 	&	0.297	 (0.019)	&	0.285	 (0.019)	&	0.012	 (0.008)	&	0.024	 (0.039)	&	0.062	\\
\quad General merchandise (53)	&	12	&	 1,380 	&	 10,440 	&	0.320	 (0.010)	&	0.292	 (0.010)	&	0.028	 (0.006)	&	0.108	 (0.030)	&	0.083	\\
\quad Food stores (54)	&	5	&	 530 	&	 4,030 	&	0.451	 (0.018)	&	0.425	 (0.018)	&	0.026	 (0.010)	&	0.063	 (0.029)	&	0.058	\\
\quad Auto dealers / services (55)	&	8	&	 891 	&	 6,930 	&	0.257	 (0.012)	&	0.204	 (0.011)	&	0.053	 (0.008)	&	0.237	 (0.040)	&	0.107	\\
\quad Apparel stores (56)	&	4	&	 400 	&	 3,093 	&	0.237	 (0.017)	&	0.202	 (0.016)	&	0.035	 (0.010)	&	0.173	 (0.067)	&	0.117	\\
\quad Furnishing stores (57)	&	4	&	 482 	&	 3,679 	&	0.286	 (0.018)	&	0.251	 (0.017)	&	0.035	 (0.011)	&	0.131	 (0.045)	&	0.086	\\
\quad Eating/drinking (58)	&	4	&	 500 	&	 4,000 	&	0.368	 (0.018)	&	0.337	 (0.018)	&	0.032	 (0.009)	&	0.086	 (0.027)	&	0.053	\\
\quad Other retail (59)	&	7	&	 816 	&	 6,281 	&	0.206	 (0.011)	&	0.182	 (0.011)	&	0.024	 (0.007)	&	0.133	 (0.060)	&	0.138	\\
\quad Banks / credit (60)	&	2	&	 252 	&	 1,947 	&	0.119	 (0.015)	&	0.121	 (0.016)	&	-0.002	 (0.010)	&	-0.073	 (0.116)	&	0.150	\\
\quad Securities brokers (62)	&	1	&	 125 	&	 965 	&	0.122	 (0.021)	&	0.111	 (0.019)	&	0.011	 (0.012)	&	0.098	 (0.102)	&	0.102	\\
\quad Insurance / real estate (63)	&	5	&	 398 	&	 2,907 	&	0.142	 (0.015)	&	0.142	 (0.016)	&	0.000	 (0.010)	&	0.015	 (0.108)	&	0.203	\\
\quad Accommodation (70)	&	2	&	 243 	&	 1,850 	&	0.200	 (0.022)	&	0.199	 (0.023)	&	0.001	 (0.012)	&	0.043	 (0.068)	&	0.094	\\
\quad Business services (73)	&	3	&	 375 	&	 2,812 	&	0.214	 (0.017)	&	0.212	 (0.017)	&	0.003	 (0.007)	&	0.101	 (0.076)	&	0.113	\\
\quad Auto / repair services (75)	&	3	&	 340 	&	 2,551 	&	0.285	 (0.022)	&	0.275	 (0.022)	&	0.010	 (0.010)	&	0.046	 (0.037)	&	0.062	\\
\quad Health services (80)	&	4	&	 400 	&	 2,886 	&	0.150	 (0.015)	&	0.144	 (0.015)	&	0.006	 (0.008)	&	-0.071	 (0.067)	&	0.127	\\
\quad Engineering services (87)	&	1	&	 47 	&	 374 	&	0.122	 (0.047)	&	0.117	 (0.046)	&	0.005	 (0.005)	&	0.044	 (0.042)	&	0.042	\\\hline\hline
\end{tabular}																					
\end{adjustbox}
\parbox{1.5\textwidth}{\small
\vspace{1eX}
\textit{Notes:} This table presents summary statistics the sample of applications used in the analysis of firm contact penalties. ``White" and ``Black" refer to average firm-level contact rates for white and Black applications. Difference is the average difference. Log dif is the industry average of the primary contact penalty measure used in the analysis: $\ln(\hat{p}_{iw}) - \ln(\hat{p}_{ib})$. Mean SE is the average standard error of firm-level log difference estimate. Standard errors in parentheses.}
\end{table}
\end{landscape}

\begin{table}[htbp]
  \centering
  \caption{GMM estimates of contact penalty parameters}
    \label{tab:gmm_table}%
    \begin{tabular}{p{6.665em}rr}
    \midrule \midrule
    \multicolumn{1}{r}{} & \multicolumn{1}{c}{No industry} & \multicolumn{1}{c}{With industry} \\
    \multicolumn{1}{r}{} & \multicolumn{1}{c}{effects} & \multicolumn{1}{c}{effects} \\
    \multicolumn{1}{r}{} & \multicolumn{1}{c}{(1)} & \multicolumn{1}{c}{(2)} \\
    \midrule
    \multicolumn{1}{l}{$\beta$} & \multicolumn{1}{c}{0.510} & \multicolumn{1}{c}{0.517} \\
    \multicolumn{1}{r}{} & \multicolumn{1}{c}{(0.190)} & \multicolumn{1}{c}{(0.121)}  \\
    \multicolumn{1}{r}{} &       &   \\
    \multicolumn{1}{l}{$\mu_{v}$} & \multicolumn{1}{c}{0.313} & \multicolumn{1}{c}{0.292} \\
    \multicolumn{1}{r}{} & \multicolumn{1}{c}{(0.074)} & \multicolumn{1}{c}{(0.074)}  \\
    \multicolumn{1}{r}{} &       &     \\
    \multicolumn{1}{l}{$\sigma_{v}$} & \multicolumn{1}{c}{0.207} & \multicolumn{1}{c}{ }  \\
    \multicolumn{1}{r}{} & \multicolumn{1}{c}{(0.106)} & \multicolumn{1}{c}{ } \\
    \multicolumn{1}{r}{} &       &   \\
    \multicolumn{1}{l}{$\sigma_{\eta}$} & \multicolumn{1}{c}{ } & \multicolumn{1}{c}{0.452}  \\
    \multicolumn{1}{r}{} & \multicolumn{1}{c}{ } & \multicolumn{1}{c}{(0.171)} \\
    \multicolumn{1}{r}{} &       &   \\
    \multicolumn{1}{l}{$\sigma_{\xi}$} & \multicolumn{1}{c}{ } & \multicolumn{1}{c}{0.144}  \\
    \multicolumn{1}{r}{} & \multicolumn{1}{c}{ } & \multicolumn{1}{c}{(0.066)} \\
    \multicolumn{1}{r}{} &       &     \\
    \multicolumn{1}{l}{Within share} & \multicolumn{1}{c}{ } & \multicolumn{1}{c}{0.556} \\
    \multicolumn{1}{r}{} &       &     \\
    \multicolumn{1}{l}{$J$-statistic (d.f.)} & \multicolumn{1}{c}{0.101 (1)} & \multicolumn{1}{c}{0.111 (2)} \\
    \hline\hline \\
    \end{tabular}
  \parbox{\textwidth}{\small
\vspace{1eX}
\textit{Notes:} This table reports generalized method of moments (GMM) estimates of the parameters of the model $\theta_{i}=s_{i}^{\beta}v_{i}$, with $\mathbb{E}[v_{i}]=\mu_{v}$ and $\mathbb{V}[v_{i}]=\sigma^2_{v}$. Column (2) allows an industry component of the form $v_{i}=\eta_{k(i)}\xi_{i}$, where $k(i)$ is the industry of firm $i$ and $\mathbb{E}[\eta_{k}]=1$. Estimates come from two-step optimally-weighted GMM with an identity weighting matrix in the first step. The variance matrix in column (2) is clustered by industry. The within share is $\ensuremath{\frac{\mathbb{E}\left[\mathbb{V}\left[v_{i}|k(i)\right]\right]}{\mathbb{V}[v_{i}]}=\frac{(\sigma_{\eta}^{2}+1)\sigma_{\xi}^{2}}{\sigma_{\eta}^{2}\sigma_{\xi}^{2}+\sigma_{\eta}^{2}\mu_{v}^{2}+\sigma_{\xi}^{2}}}$.}
\end{table}

\begin{table}[htbp]
  \centering
  \caption{Moments of random effect distributions}
    \label{tab:double_g_table}%
    \begin{tabular}{p{6.665em}rrr}
    \midrule \midrule
    \multicolumn{1}{r}{} & \multicolumn{1}{c}{Industry} & \multicolumn{1}{c}{Firm} & \multicolumn{1}{c}{Contact} \\
    \multicolumn{1}{r}{} & \multicolumn{1}{c}{effect ($\eta_{k}$)} & \multicolumn{1}{c}{effect ($\xi_{i}$)} & \multicolumn{1}{c}{penalty ($\theta_{i}$)} \\
    \multicolumn{1}{r}{} & \multicolumn{1}{c}{(1)} & \multicolumn{1}{c}{(2)} & \multicolumn{1}{c}{(3)} \\
    \midrule
    \multicolumn{1}{l}{Mean} & \multicolumn{1}{c}{1.000} & \multicolumn{1}{c}{0.292} & \multicolumn{1}{c}{0.086} \\
    \multicolumn{1}{r}{} & \multicolumn{1}{c}{-} & \multicolumn{1}{c}{(0.040)} & \multicolumn{1}{c}{(0.012)} \\
    \multicolumn{1}{r}{} &       &       &  \\
    \multicolumn{1}{l}{Std. dev.} & \multicolumn{1}{c}{0.495} & \multicolumn{1}{c}{0.138} & \multicolumn{1}{c}{0.072} \\
    \multicolumn{1}{r}{} & \multicolumn{1}{c}{(0.192)} & \multicolumn{1}{c}{(0.049)} & \multicolumn{1}{c}{(0.019)} \\
    \multicolumn{1}{r}{} &       &       &  \\
    \multicolumn{1}{l}{Skewness} & \multicolumn{1}{c}{1.856} & \multicolumn{1}{c}{2.293} & \multicolumn{1}{c}{3.545} \\
    \multicolumn{1}{r}{} & \multicolumn{1}{c}{(1.824)} & \multicolumn{1}{c}{(4.580)} & \multicolumn{1}{c}{(2.778)} \\
    \multicolumn{1}{r}{} &       &       &  \\
    \multicolumn{1}{l}{Excess kurtosis} & \multicolumn{1}{c}{7.735} & \multicolumn{1}{c}{30.067} & \multicolumn{1}{c}{43.653} \\
    \multicolumn{1}{r}{} & \multicolumn{1}{c}{(6.740)} & \multicolumn{1}{c}{(52.585)} & \multicolumn{1}{c}{(70.467)} \\
    \hline\hline \\
    \end{tabular}
  \parbox{\textwidth}{\small
\vspace{1eX}
\textit{Notes:} This table reports estimated moments of the distributions of industry and firm effects along with moments of the marginal distribution of contact penalties. Results are derived from hierarchical log-spline deconvolution estimates, with spline parameters estimated by penalized maximum likelihood. Standard errors are computed by the Delta method, with the variance matrix for the spline parameters computed from the negative inverse Hessian of the log likelihood function.}
\end{table}%

\clearpage
\FloatBarrier
\setcounter{figure}{0}    
\setcounter{table}{0}    
\setcounter{section}{0}    
\appendix
\renewcommand{\thesection}{Appendix \Alph{section}}
\renewcommand{\thesubsection}{\Alph{section}.\arabic{subsection}}
\renewcommand{\theHsection}{\Alph{section}.\Alph{section}}
\renewcommand\thefigure{\Alph{section}\arabic{figure}}    
\renewcommand{\theHfigure}{\Alph{section}.\arabic{figure}}
\renewcommand\thetable{\Alph{section}\arabic{table}}    
\renewcommand{\theHtable}{\Alph{section}.\arabic{table}}
\setcounter{page}{1}

\begin{center}
\textbf{\huge{Appendix}} \\
\huge{(for online publication only)}
\end{center}

\clearpage
\setcounter{section}{0}    
\section{Proofs of propositions} \label{sec: proofs}

This section provides proofs of the propositions discussed in Section \ref{sec: social_choice}, which are restated here for completeness.

\setcounter{proposition}{0}

\begin{proposition}[$\lambda$-Condorcet Criterion]
Suppose that firm $i$ satisfies $\pi_{ij} > (1+\lambda)^{-1} \ \forall \ j \neq i$. Then $d_i > d_j \ \forall \ j \neq i$. Moreover, suppose that firm $k$ satisfies $\pi_{ik} > (1+\lambda)^{-1}$ and  $\pi_{kj} > (1+\lambda)^{-1} \ \forall \ j \neq i, j \neq k$, then $d_i > d_k > d_j \ \forall \ j \neq i, j \neq k$.
\end{proposition}

\begin{proof} 
First, we establish that no firm can be tied with firm $i$.  Suppose $\exists \ j \ s.t. \ d_j = d_i = d$. Let $\tilde{d} = \inf\{\{d' \in d^*(\lambda) \ s.t. \ d' > d\} \cup \{\infty\}\}$. Then changing firm $i$'s grade to a value in $(d,\tilde{d})$ yields strictly lower loss, because $\sum_{j \neq i \ s.t. \ d_j = d} \pi_{ji} - \lambda \pi_{ij} < 0$, and comparisons between $i$ and all other firms $j \ s.t. \ d_j \neq d$ are unaffected.

Now suppose $\exists \ d \in d^*(\lambda) \ s.t. \ d > d_i$. Let $d' = \inf \{d \in d^*(\lambda) \ s.t. \ d > d_i\}$. Then $\forall j \ s.t. \ d_j = d'$, the risk of re-assigning $d_i = d' + \epsilon < \inf\{\{d \in d^*(\lambda) \ s.t. \ d > d'\} \cup \{\infty\}\}$ is strictly lower because $\sum_{j \neq i \ s.t. \ d_j = d'} \pi_{ji} - \lambda \pi_{ij} < 0 < \sum_{j \neq i \ s.t. \ d_j = d'} \pi_{ij} - \lambda \pi_{ji}$, and comparisons between $i$ and all other firms $j \ s.t. \ d_j \neq d'$ are unaffected. Since the same argument applies to firm $k$ removing firm $i$ from set of firms under consideration, the proof of the second part of the claim is identical. 
\end{proof}

\begin{proposition}[$\lambda$-Smith criterion]\label{prop2}
Let $\mathcal{S}$ denote a collection of firms exhibiting the following dominance property: $\pi_{ij}>(1+\lambda)^{-1}$ $\forall i\in \mathcal{S},j \notin \mathcal{S}$. Then the top graded firms must be a member of $\mathcal{S}$.
\end{proposition}

\begin{proof} 
First, note that if $\mathcal{S}$ is a singleton, then Proposition 1 applies directly. Otherwise, let $\tilde{d} = \sup\{d_i \ s.t. \ i \in \mathcal{S}\}$ and let $\bar{\mathcal{S}}$ denote the set $\{i \in \mathcal{S} \ s.t. \ d_i = \tilde{d}\}$. Suppose $\exists \ j \notin \mathcal{S} \ s.t. \ d_j > \tilde{d}$. Let $d' = \inf\{d \in d^*(\lambda) \ s.t. \ d > \tilde{d}\}$ and ${\underline{\mathcal{S}}}$ denote the set $\{j \notin \mathcal{S} \ s.t. \ d_j = d'\}$. Then swapping grades such that all firms in $\bar{\mathcal{S}}$ receive grade $d'$ and all firms in ${\underline{\mathcal{S}}}$ receive grade $\tilde{d}$ must decrease risk, because $\sum_{i \in \bar{\mathcal{S}}}\sum_{j \in {\underline{\mathcal{S}}}} \pi_{ji} - \lambda \pi_{ij} < 0 < \sum_{i \in \bar{\mathcal{S}}} \sum_{j \in {\underline{\mathcal{S}}}} \pi_{ij} - \lambda \pi_{ji}$, comparisons between all firms within $\bar{\mathcal{S}}$ and $\underline{\mathcal{S}}$ are unaffected, and comparisons between all firms $k \notin \{\bar{\mathcal{S}} \cup \underline{\mathcal{S}}\}$ are unaffected. Thus no firm $j \notin \mathcal{S}$ may be ranked above the top graded member of $\mathcal{S}$.
\end{proof}

\begin{proposition}[Unordered $\lambda$-Smith candidates are tied]
Let $\mathcal{S}$ denote a collection of firms exhibiting the following dominance property: $\pi_{ij}>(1+\lambda)^{-1}$ $\forall i\in \mathcal{S},j \notin \mathcal{S}$. Moreover, suppose $\pi_{ij}<(1+\lambda)^{-1}$ $\forall (i,j) \in \mathcal{S}$. Then all firms in $\mathcal{S}$ receive the highest grade.
\end{proposition}

\begin{proof}
First, we show that all firms $j \notin \mathcal{S}$ must be ranked below every member of $\mathcal{S}$. Suppose not. Let $d' = \inf\{d_j \ s.t. \ j \notin \mathcal{S}, \exists \ i \in \mathcal{S} \ s.t. \ d_j >  d_i\}$, $\underline{\mathcal{S}} = \{j \notin \mathcal{S} \ s.t. \ d_j = d'\}$, $\tilde{d} = \sup\{d_i \ s.t. \ i \in \mathcal{S}, d_i < d'\}$, $\bar{\mathcal{S}} = \{i \in \mathcal{S} \ s.t. \ d_i = \tilde{d}\}$. Then setting grades so that all firms in $\underline{\mathcal{S}}$ receive a grade $m \in (d',\tilde{d})$ and all firms in ${\bar{\mathcal{S}}}$ receive grade $d'$ must decrease risk because $\sum_{i \in \bar{\mathcal{S}}}\sum_{j \in {\underline{\mathcal{S}}}} \pi_{ji} - \lambda \pi_{ij} < 0 < \sum_{i \in \bar{\mathcal{S}}} \sum_{j \in {\underline{\mathcal{S}}}} \pi_{ij} - \lambda \pi_{ji}$, implying it is optimal to rank all firms in $\bar{\mathcal{S}}$ above those in $\underline{\mathcal{S}}$. Moreover, $\sum_{i \in \bar{\mathcal{S}}} \sum_{j \in \mathcal{S} \ s.t. \ d_j = d'} \pi_{ji} - \lambda \pi_{ij} > 0$ implies that it is optimal to tie firms in $\bar{\mathcal{S}}$ with firms in $\mathcal{S}$ that already have grade $d'$, while  $\sum_{j \in \underline{\mathcal{S}}} \sum_{i \in \mathcal{S} \ s.t. \ d_i = d'} \pi_{ji} - \lambda \pi_{ij} < 0$ implies it is optimal to rank any firms in $\mathcal{S}$ that already have grade $d'$ above those firms $\notin \mathcal{S}$ reassigned to grade $m$, and $\sum_{i \in \bar{\mathcal{S}}} \sum_{j \notin \mathcal{S} \ s.t. \ d_j = \tilde{d}} \pi_{ji} - \lambda \pi_{ij} < 0$ implies it is optimal to rank firms in $\bar{\mathcal{S}}$ above any firms $\notin \mathcal{S}$ that currently have grade $\tilde{d}$. Comparisons to all firms with grades higher than $d'$ are unaffected, as well as to any firms with grades below $\tilde{d}$. The top grades thus consist exclusively of firms in $\mathcal{S}$. To see that they also must be tied, note that because $\pi_{ji} -\lambda \pi_{ij}>0$ $\forall (i,j) \in \mathcal{S}$, collapsing any two adjacent grades for firms in $\mathcal{S}$ must decrease risk.
\end{proof}

\clearpage
\FloatBarrier
\section{Computing posteriors} \label{sec: posteriors}

This appendix details computation of posterior distributions for the firm contact gap analysis of Section 5. Computation of posteriors for the name contact rate analysis of Section 4 is a special case of this framework setting the dependence parameter $\beta$ to zero and the standard error $s_{i}$ for name $i$ to $(4N_{i})^{-1}$. Under the model in \eqref{eq:dependence_model}, the posterior density for $v_{i}=\theta_{i}/s_{i}^{\beta}$ is given by
\begin{flalign*}
f_{v}(x|Y_{i};G_{v},\beta) =& \frac{\mathcal{L}(\hat{\theta}_i | v_i = x, s_i;\beta)dG_{v}(x)}{\int \mathcal{L}(\hat{\theta}_i | v_i = u, s_i;\beta)dG_{v}(u)}, \\
\mathcal{L}(\hat{\theta}_i | v_i = x, s_i;\beta) =& \frac{1}{s_i^{1-\beta}}\phi \left( \frac{(\hat{\theta}_i/s_{i}^{\beta})- x}{s_i^{1-\beta}} \right).
\end{flalign*}

\noindent Taking $\hat G_{v}$ as a deconvolution estimate of $G_{v}$ and $\hat{\beta}$ as a GMM estimate of $\beta$, posterior means for $\theta_{i}$ are computed as $s_{i}^{\hat{\beta}}\times \int x f_{v}(x | Y_i;\hat{G}_{v},\hat{\beta})dx$, while the lower and upper limits of 95\% credible intervals are given by the 2.5th and 97.5th percentiles of the posterior cumulative  distribution $F_{\theta}(t | Y_i;\hat{G}_{v},\hat{\beta}) = \int_{-\infty}^{t/s_{i}^{\hat{\beta}}} f_{v}(x | Y_i;\hat{G}_{v},\hat{\beta})dx$.

We also use  $\hat{G}_{v}$ and $\hat{\beta}$ to compute the matrix of pairwise posterior ranking probabilities $\pi_{ij}$. We have:
\begin{flalign*}
\pi_{ij} &= \Pr(\theta_i > \theta_j | Y_i, Y_j;G_{v},\beta) \\
&= \Pr((s_{i}/s_{j})^{\beta}v_{i} > v_{j} | Y_i, Y_j; G_{v},\beta) \\
&= \int_{-\infty}^{\infty} \int_{-\infty}^{(s_{i}/s_{j})^{\beta}x} f_{v}(x|Y_{i};G_{v},\beta)f_{v}(u|Y_{j};G_{v},\beta)du dx.
\end{flalign*}
\noindent Posterior moments of differences in firm biases $\mu^p_{ij}=\mathbb{E}\left[\max\{(\theta_i-\theta_j),0\}^p\mid Y_i,Y_j\right]$ are computed analogously. Specifically:
\begin{flalign*}
\mathbb{E}\left[\max\{(\theta_i-\theta_j),0\}^p\mid Y_i,Y_j\right] &= \int_{-\infty}^{\infty} \int_{-\infty}^{(s_{i}/s_{j})^{\beta}x} (s_{i}^{\beta}x - s_{j}^{\beta}u)^{p}f_{v}(x|Y_{i};G_{v},\beta)f_{v}(u|Y_{j};G_{v},\beta)du dx.
\end{flalign*}

\noindent We plug $\hat{G}_{v}$ and $\hat{\beta}$ into these formulas to construct empirical Bayes posteriors by numerical integration, then solve the linear programming problem to minimize either binary or weighted risk using Gurobi.

\subsection{Industry effects}
Posteriors for the hierarchical industry effects model of Section \ref{sec: industry_rfe} condition on the  data for all firms in an industry. Let $\mathbf{Y}_{k}$ denote the $2n_{k}\times 1$ vector of estimates $\hat{\theta}_{i}$ and standard errors $s_{i}$ for all firms in industry $k$, and let $\pmb{\xi}_{k}$ denote the $n_{k}\times 1$ vector of within-industry deviations $\xi_{i}$ for all firms in this industry. The joint posterior density for $\pmb{\xi}_{k}$ and the industry effect $\eta_{k}$ at the point where $\eta_{k}=x$ and $\mathbf{\xi}_{k}=\mathbf{z}=(z_{1},....,z_{n_{k}})^{\prime}$ is given by:
\begin{flalign*}
f_{\eta,\pmb{\xi}}(x,\mathbf{z}|\mathbf{Y}_{k};G_{\eta},G_{\xi},\beta) =& \frac{\left[\prod_{i:k(i)=k}\mathcal{L}(\hat{\theta}_i | v_i = x\times z_{i}, s_i;\beta)dG_{\xi}(z_{i})\right]dG_{\eta}(x)}{\int_{u} \int_{\mathbf{t}} \left[\prod_{i:k(i)=k}\mathcal{L}(\hat{\theta}_i | v_i = u\times t_{i}, s_i;\beta)dG_{\xi}(t_{i})\right]dG_{\eta}(u)}.
\end{flalign*}
We form Empirical Bayes joint posteriors given by $f_{\eta,\pmb{\xi}}(x,\mathbf{z}|\mathbf{Y}_{k};\hat{G}_{\eta},\hat{G}_{\xi},\hat{\beta})$, where $\hat{\beta}$ is the GMM estimate of $\beta$ from column (2) of Table 3, and $\hat{G}_{\xi}$ and $\hat{G}_{\eta}$ are hierarchical deconvolution estimates from panel (a) of Figure 5. We then integrate over these joint posteriors by simulation to compute posterior means and quantiles for each random effect along with pairwise posterior probabilities $\pi_{ij}$ for the model with industry effects.

\subsection{Between grade variance}
Letting $M$ denote the total number of grades, the (firm-weighted) between grade variance of $\theta_i$ can be written
\[
\sum_{g=1}^{M}w_{g}\bar{\theta}_{g}^{2}-\left(\sum_{g=1}^{M}w_{g}\bar{\theta}_{g}\right)^{2}=\sum_{g=1}^{M}w_{g}\left(1-w_{g}\right)\bar{\theta}_{g}^{2}-\sum_{g=1}^{M}\sum_{g'\neq g}w_{g}w_{g'}\bar{\theta}_{g}\bar{\theta}_{g'},
\]
where $\bar{\theta}_{g}=\frac{\sum_{i=1}^{n}D_{ig}\theta_{i}}{\sum_{i=1}^{n}D_{ig}}$, $D_{ig}=1\{d_i^*=g\}$ is an indicator for being assigned grade $g\in[M]$, and $w_{g}=n^{-1}\sum_{i=1}^{n}D_{ig}$ gives the share of firms assigned grade $g$. 

We compute a Bayes unbiased estimate of each $\bar{\theta}_{g}$ by simply averaging the firm specific posterior firm means $\mathbb{E}\left[\theta_{i}|Y\right]$ within grade. The posterior mean estimate of each $\bar{\theta}_{g}^{2}$ is slightly harder to compute because 
\begin{align*}
\bar{\theta}_{g}^{2}&=\frac{\sum_{i=1}^{n}D_{ig}\theta_{i}^{2}}{\left(\sum_{i=1}^{n}D_{ig}\right)^{2}}+\frac{\sum_{i=1}^{n}\sum_{i'\neq i}D_{ig}D_{i'g}\theta_{i}\theta_{i'}}{\left(\sum_{i=1}^{n}D_{ig}\right)^{2}}\\&=\left(nw_{g}\right)^{-2}\left\{ \sum_{i=1}^{n}D_{ig}\theta_{i}^{2}+\sum_{i=1}^{n}\sum_{i'\neq i}D_{ig}D_{i'g}\theta_{i}\theta_{i'}\right\} .
\end{align*}
Our posterior mean estimate of this quantity is computed analogously as
\begin{align*}
\mathbb{E}\left[\bar{\theta}_{g}^{2}|Y\right]	&=	\left(nw_{g}\right)^{-2}\left\{ \sum_{i=1}^{n}D_{ig}\mathbb{E}\left[\theta_{i}^{2}|Y\right]+\sum_{i=1}^{n}\sum_{i'\neq i}D_{ig}D_{i'g}\mathbb{E}\left[\theta_{i}|Y\right]\mathbb{E}\left[\theta_{i'}|Y\right]\right\} \\
	&=	\left(nw_{g}\right)^{-2}\left\{ \sum_{i=1}^{n}D_{ig}\mathbb{E}\left[\theta_{i}^{2}|Y\right]+\left(\sum_{i=1}^{n}D_{ig}\mathbb{E}\left[\theta_{i}|Y\right]\right)^{2}-\sum_{i=1}^{n}D_{ig}\mathbb{E}\left[\theta_{i}|Y\right]^{2}\right\} ,
\end{align*}
where each $\mathbb{E}\left[\theta_{i}|Y\right]$ and $\mathbb{E}\left[\theta_{i}^2|Y\right]$ is evaluated numerically using the relevant estimated $\hat G$.

\clearpage
\section{Hierarchical log-spline estimator} \label{sec: industry_fx}

Efron's (2016) empirical Bayes log-spline deconvolution approach uses a flexible exponential family mixing distribution model with density parameterized by a flexible $B$-th order spline. The spline parameters are then estimated by penalized maximum likelihood. We adapt this approach to estimate the within- and between-industry distributions for the hierarchical model of Section \ref{sec: industry_rfe}. The between-industry distribution $G_{\eta}$ is approximated with a discrete probability mass function defined on a set of $M_{\eta}$ support points $\{\bar{\eta}_{1},....,\bar{\eta}_{M_{\eta}}\}$. The mass at the $m$-th support point $\bar{\eta}_{m}$ is given by
\[
g_{\eta,m}(\alpha_{\eta})=\exp\left(q_{\eta,m}^{\prime}\alpha_{\eta}-\log\left(\sum_{\ell=1}^{M_{\eta}}\exp(q_{\eta,\ell}^{\prime}\alpha_{\eta})\right)\right),
\]
 where $q_{\eta,m}$ is a $B \times 1$ vector of values of natural spline basis functions for point $m$ (as detailed in Efron 2016) and $\alpha_{\eta}$ is a $B\times 1$ vector of coefficients. Similarly, we approximate the within-industry distribution $G_{\xi}$ with a discrete distribution defined on support $\{\bar{\xi}_{1},....,\bar{\xi}_{M_{\xi}}\}$, with mass function
 \[
g_{\xi,m}(\alpha_{\xi})=\exp\left(q_{\xi,m}^{\prime}\alpha_{\xi}-\log\left(\sum_{\ell=1}^{M_{\xi}}\exp(q_{\xi,\ell}^{\prime}\alpha_{\xi})\right)\right)
\]
for $B\times 1$ spline basis and coefficient vectors $q_{\xi,m}$ and $\alpha_{\xi}$, respectively. 

With this specification of the mixing distributions, the joint likelihood contribution for firms in industry $k$ is given by:
\[
\mathcal{L} \left(\hat{\bm{\theta}}_{k}|\bm{s}_{k};\alpha_{\eta},\alpha_{\xi}\right) = \displaystyle{\sum_{\ell = 1}^{M_{\eta}}g_{\eta,\ell}(\alpha_{\eta})\left\{\prod_{i:k(i)=k}\left[\sum_{m=1}^{M_{\xi}}g_{\xi,m}(\alpha_{\xi})\frac{1}{s_{i}^{1-\beta}}\phi\left(\frac{(\hat{\theta}_{i}/s_{i}^{\beta})-\bar{\eta}_{\ell}\bar{\xi}_{m}}{s_{i}^{1-\beta}}\right)\right]\right\}},
\]
where $\hat{\bm{\theta}}_{k}$ and $\bm{s}_{k}$ are vectors collecting the $\hat{\theta}_{i}$ and $s_{i}$ for firms with $k(i)=k$. Following Efron (2016), we estimate the parameters $\alpha_{\eta}$ and $\alpha_{\xi}$ by penalized maximum likelihood. Our approach extends the Efron (2016) estimator to add a separate penalty for the within- and between-industry spline coefficients. Specifically, the parameter estimates are computed as:
\[
(\hat{\alpha}_{\eta},\hat{\alpha}_{\xi}) = \displaystyle{\arg \max_{(\alpha_{\eta},\alpha_{\xi})}\sum_{k=1}^{K}\log\mathcal{L} \left(\hat{\bm{\theta}}_{k}|\bm{s}_{k};\alpha_{\eta},\alpha_{\xi}\right) -c_{\eta}\sqrt{\alpha_{\eta}^{\prime}\alpha_{\eta}}-c_{\xi}\sqrt{\alpha_{\xi}^{\prime}\alpha_{\xi}}}.
\]
We set the spline order equal to $B=5$. In models with industry effects the number of support points is set equal to $M_{\eta}=M_{\xi}=200$, with points equally spaced on the supports of $\eta_{k}$ and $\xi_{i}$. Models without industry effects use $M_{\xi}=1,000$ and $M_{\eta}=1$ with $\bar{\eta}_{1}=1$, so that $\eta_{k}$ has a degenerate distribution at unity. The lower limits of the supports are set to zero. The upper limits of the support for each component is set equal to the maximum of the empirical distribution of corresponding estimates or five standard deviations above the GMM-estimated mean, whichever is larger. Since the scales of the two mixing distributions are not separately identified we impose the constraint $\sum_{m}g_{\eta,m}(\alpha_{m})\bar{\eta}_{m}=1$. The penalty terms $c_{\eta}$ and $c_{\xi}$ are calibrated so that mean contact ratio and variances of the within- and between-industry components come as close as possible to matching GMM estimates of these same quantities, as measured by the quadratic distance between log-spline and GMM estimates (scaled by the inverse variance matrix of the GMM estimates). The model without industry effects chooses $c_{\xi}$ to minimize the quadratic difference between model-implied and GMM estimates of the mean and total variance of contact gaps. In practice all parameters match well, as can be seen by comparing Tables \ref{tab:gmm_table} and \ref{tab:double_g_table}.

\clearpage
\FloatBarrier
\section{Results under square weighted loss} \label{appdx:sqwt}

In this Appendix we discuss the grades that result from minimizing the square-weighted risk function $\mathcal{R}^2(d;\lambda)$ introduced in Section \ref{sec:sqwt}. Mirroring the formulation in \eqref{eq:obj}, square-weighted risk can be written as a linear combination of a square-weighted Discordance Rate and a square-weighted rank correlation:
\begin{align}
\mathcal{R}^2(d;\lambda)=(1-\lambda)DR^2(d) - \lambda \bar{\tau}^2(d). \label{eq:sqdecomp}
\end{align}
The square-weighted rank correlation coefficient
\begin{align}
\bar{\tau}^2(d)=\frac{1}{\sum_{i=2}^{n}\sum_{j=1}^{i-1} m_{ij}^{2}}\sum_{i=2}^{n}\sum_{j=1}^{i-1}1\left\{ d_{i}<d_{j}\right\} (\mu_{ij}^2 - \mu_{ji}^2)+1\left\{ d_{i}>d_{j}\right\} (\mu_{ji}^2 - \mu_{ij}^2) \nonumber
\end{align}
provides a measure of the information conveyed by a set of grades. The square-weighted Discordance Rate
\begin{align}
DR^2(d)=\frac{1}{\sum_{i=2}^{n}\sum_{j=1}^{i-1} m_{ij}^{2}}\sum_{i=2}^{n}\sum_{j=1}^{i-1}1\left\{ d_{i}<d_{j}\right\} \mu_{ij}^2+1\left\{ d_{i}>d_{j}\right\} \mu_{ji}^2 \nonumber
\end{align}
is an average of the $DR^2_{g,g'}$ introduced in Section \ref{sec:DR} that summarizes the reliability of a set of grades. In words, $DR^2(d)$ gives the chance that a randomly selected pair of firms, sampled with weights proportional to the square of the difference in their latent contact gaps, is misranked.

We use the decomposition given by \eqref{eq:sqdecomp} in the applications below to construct a reporting possibilities frontier contrasting $DR^2(d^*(\lambda))$ with $\bar {\tau}^2(d^*(\lambda))$ for different choices of $\lambda$. In both applications, square-weighting shifts out the reporting possibilities frontier relative to the frontier implied by binary (i.e. $p=0$) loss, which is intuitive as one should expect less loss from ranking exercises when small mistakes are inconsequential. As a result, switching from binary to square-weighted loss tends to yield more grades when $\lambda$ is held constant. 

\subsection{Square-weighted results for first names}

Panel (a) of Figure \ref{fig: names_res_sq} depicts the number of grades that emerge from minimizing square-weighted risk as a function of the preference parameter $\lambda$. At our preferred value of $\lambda=0.25$, five grades emerge with a square weighted discordance rate of \namedrsq\%. Panel (b) of Figure \ref{fig: names_res_sq} shows that square-weighting shifts out the reporting possibilities frontier. As in our binary baseline, naively ranking either the point estimates $\hat \theta_i$ or the posterior means $\bar \theta_i$ according to their empirical distribution yields grades with information content and reliability very close to those delivered by the square-weighted grades with $\lambda=1$.

The five grades that result from choosing $\lambda=0.25$ under square-weighting are shown in Figure \ref{fig: names_ranks_sq_weighted}. Interestingly, under this weighting scheme, the posterior means of the names are more nearly monotone in the $n$ assigned grades that result when setting $\lambda=1$ than was found under binary loss. Two distinctively white female names ``Misty'' and ``Heather'' earn the top grade, while three out of the four names assigned the bottom grade ``Lawanda'', ``Tameka'', and ``Latisha'' are distinctively Black and female. As shown in Figure \ref{fig:pseudoR2_sq}, however, a logistic regression of sex on these five grade categories yields a pseudo-$R^2$ of roughly 0.05, while a corresponding race regression yields a pseudo-$R^2$ of 0.24 revealing that these grades are much stronger predictors of race than sex.

\subsection{Square-weighted results for firms}
Figure \ref{fig:poisson_binary_lambda_sq} shows the number of grades that emerge from different choices of $\lambda$. Regardless of whether industry effects are included, square weighting substantially increases the number of grades assigned relative to the solutions under binary loss. At $\lambda=0.25$ square weighting yields six grades without industry effects and eight grades when industry effects are included. While the number of grades increases, their reliability remains high, with the six grade ranking exhibiting a square-weighted discordance rate of \sqdr\% \ and the eight grade ranking a corresponding discordance rate of \industrysqdr\%. Evidently, it is possible to assign many grades without exposing oneself to many large mistakes.

Figure \ref{fig:PPF_sq} displays the reporting possibilities under binary and square weighting. Square-weighting the losses pushes out the reporting possibilities frontier relative to binary ranking. Naively ranking based upon the posterior mean $\bar \theta_i$ again yields results very similar to the grades generated by setting $\lambda=1$. However ranking based upon the naive (i.e., unshrunk) point estimates $\hat \theta_i$ yields a combination of very high square-weighted Discordance Rates and a modest square-weighted rank correlation that is dominated by our square-weighted grades when $\lambda=0.25$.

Figure \ref{fig:poisson_squared} displays the firm rankings derived from minimizing square-weighted risk without conditioning on industry effects. As before, firms are ordered by their ranking under $\lambda=1$, which aligns closely with the Condorcet ranks reported in Figure \ref{fig:poisson_binary}. Relative to the binary loss solution, square-weighting leads to three additional grades, and a more even distribution of firms across grades. Using industry effects to evaluate the square-weighted risk generates the eight grade ranking portrayed in Figure \ref{fig:poisson_squared_cov}. Notably, a single firm with the highest posterior mean contact penalty earns the worst grade.

As in the binary loss case, these grades explain a meaningful share of total variance in $\theta_i$. The estimated between-grade standard deviation in contact penalties is \sqbtwnvar \ for the six grades reported in Figure \ref{fig:poisson_squared}, implying an $R^2$ of \sqbtwnr\%. The grades computed from the model with industry effects and reported in Figure \ref{fig:poisson_squared_cov} explain \indsqbtwnr\% of the total variance. Figures \ref{fig: sqwt_all_grades} and \ref{fig: sqwt_irfe_all_grades} depict the grades assigned under different choices of $\lambda$ for the square-weighting schemes excluding and including industry effects respectively.

Figure \ref{fig: sqwt_irfe_all_grades} reports square-weighted DR estimates for the grades that minimize $\mathcal{R}^2(d;\lambda)$. Under square weighted loss, we find weighted Discordance Rates between adjacent grades ranging from 10 to 15\% implying Bayes factors ranging from 8 to 6. Weighted Discordance Rates between non-adjacent grades are uniformly smaller than 3\%. Our finding that the square-weighted DRs in Figure \ref{fig: sqwt_irfe_all_grades} fall below those in Figure \ref{fig: sqwt_irfe_all_grades}, reflects that most of the ranking mistakes likely to arise are small in magnitude, meaning that if one firm is erroneously listed as more biased than another, the two firms likely exhibit similar levels of bias.

\clearpage
\FloatBarrier
\section{Additional Figures and Tables}

\begin{table}[h!]
\caption{First names assigned by race and gender}
\label{tab: firstname_list}
\begin{adjustbox}{center}
\begin{tabular}{l l c l c l c l c} \hline\hline
& \multicolumn{2}{c}{Black male} & \multicolumn{2}{c}{White male} & \multicolumn{2}{c}{Black female}  & \multicolumn{2}{c}{White female} \\\cmidrule(lr){2-3}\cmidrule(lr){4-5}\cmidrule(lr){6-7}\cmidrule(lr){8-9}
& Name & Source & Name & Source  & Name & Source  & Name & Source \\\hline
1	&	Antwan	&	NC	&	Adam	&	NC	&	Aisha	&	Both	&	Allison	&	BM	\\
2	&	Darnell	&	BM	&	Brad	&	Both	&	Ebony	&	Both	&	Amanda	&	NC	\\
3	&	Donnell	&	NC	&	Bradley	&	NC	&	Keisha	&	BM	&	Amy	&	NC	\\
4	&	Hakim	&	BM	&	Brendan	&	Both	&	Kenya	&	BM	&	Anne	&	BM	\\
5	&	Jamal	&	Both	&	Brett	&	BM	&	Lakeisha	&	NC	&	Carrie	&	BM	\\
6	&	Jermaine	&	Both	&	Chad	&	NC	&	Lakesha	&	NC	&	Emily	&	Both	\\
7	&	Kareem	&	Both	&	Geoffrey	&	BM	&	Lakisha	&	Both	&	Erin	&	NC	\\
8	&	Lamar	&	NC	&	Greg	&	BM	&	Lashonda	&	NC	&	Heather	&	NC	\\
9	&	Lamont	&	NC	&	Jacob	&	NC	&	Latasha	&	NC	&	Jennifer	&	NC	\\
10	&	Leroy	&	BM	&	Jason	&	NC	&	Latisha	&	NC	&	Jill	&	Both	\\
11	&	Marquis	&	NC	&	Jay	&	BM	&	Latonya	&	Both	&	Julie	&	NC	\\
12	&	Maurice	&	NC	&	Jeremy	&	NC	&	Latoya	&	Both	&	Kristen	&	Both	\\
13	&	Rasheed	&	BM	&	Joshua	&	NC	&	Lawanda	&	NC	&	Laurie	&	BM	\\
14	&	Reginald	&	NC	&	Justin	&	NC	&	Patrice	&	NC	&	Lori	&	NC	\\
15	&	Roderick	&	NC	&	Matthew	&	Both	&	Tameka	&	NC	&	Meredith	&	BM	\\
16	&	Terrance	&	NC	&	Nathan	&	NC	&	Tamika	&	Both	&	Misty	&	NC	\\
17	&	Terrell	&	NC	&	Neil	&	BM	&	Tanisha	&	BM	&	Rebecca	&	NC	\\
18	&	Tremayne	&	BM	&	Scott	&	NC	&	Tawanda	&	NC	&	Sarah	&	Both	\\
19	&	Tyrone	&	Both	&	Todd	&	BM	&	Tomeka	&	NC	&	Susan	&	NC	\\
\hline\hline
\end{tabular}
\end{adjustbox}
\parbox{\textwidth}{\small
\vspace{1eX}
      \textit{Notes:} This table lists the first names assigned by race and gender and their sources. ``BM" indicates that the name appeared in original set of nine names used for each group in \citet{bertrand_mullainathan_2004}. ``NC" indicates the name was drawn from data on North Carolina speeding infractions and arrests. ``Both" indicates the name appeared in both sources. Names from N.C. speeding tickets were selected from the most common names where at least 90\% of individuals are reported to belong to the relevant race and gender group.}
\end{table}

\clearpage
\begin{center}
\renewcommand*{\arraystretch}{.88}
\begin{longtable}{c c cccc cccc}
\caption[Detailed results by firm]{Detailed results by firm} \label{tab: results_detail}
\\
\hline\hline						
	&		&	\multicolumn{4}{c}{Baseline model}							&	\multicolumn{4}{c}{Industry effects model}							\\\cmidrule(lr){3-6}\cmidrule(lr){7-10}
         Firm   &               & Post. & Post. & & Cond.& Post. & Post. & & Cond. \\
(SIC)	&	$\theta_i$	&	mean	& CI	&	Grd	& 	rank	&	mean	&	 CI	&	Grd	& 	rank	\\\hline \\[-0.7em]
\endfirsthead

\multicolumn{5}{l}{{\textit{Continued on next page}}} \\ \hline
\endfoot

\hline\hline
\endlastfoot
	
6 (55)	&	0.33 (0.07)	& 	0.25	& 	[0.12, 0.34]	& 	1	& 	1	&	0.27	& 	[0.16, 0.38]	& 	1	& 	2	\\
1 (55)	&	0.43 (0.13)	& 	0.23	& 	[0.08, 0.45]	& 	1	& 	2	&	0.33	& 	[0.16, 0.5]	& 	1	& 	1	\\
3 (53)	&	0.38 (0.28)	& 	0.19	& 	[0.03, 0.38]	& 	2	& 	3	&	0.17	& 	[0.03, 0.32]	& 	2	& 	8	\\
2 (63)	&	0.4 (0.22)	& 	0.19	& 	[0.04, 0.4]	& 	2	& 	4	&	0.13	& 	[0.02, 0.27]	& 	3	& 	17	\\
4 (24)	&	0.35 (0.29)	& 	0.19	& 	[0.03, 0.37]	& 	2	& 	5	&	0.15	& 	[0.02, 0.32]	& 	2	& 	12	\\
5 (59)	&	0.34 (0.24)	& 	0.18	& 	[0.03, 0.35]	& 	2	& 	6	&	0.17	& 	[0.03, 0.34]	& 	2	& 	9	\\
27 (24)	&	0.17 (0.31)	& 	0.17	& 	[0.02, 0.34]	& 	2	& 	7	&	0.15	& 	[0.02, 0.32]	& 	2	& 	14	\\
8 (56)	&	0.3 (0.23)	& 	0.17	& 	[0.03, 0.33]	& 	2	& 	8	&	0.16	& 	[0.03, 0.36]	& 	2	& 	11	\\
7 (73)	&	0.3 (0.19)	& 	0.16	& 	[0.03, 0.32]	& 	2	& 	9	&	0.10	& 	[0.01, 0.24]	& 	3	& 	29	\\
11 (55)	&	0.27 (0.08)	& 	0.16	& 	[0.06, 0.32]	& 	2	& 	10	&	0.24	& 	[0.12, 0.36]	& 	1	& 	5	\\
13 (51)	&	0.24 (0.24)	& 	0.16	& 	[0.03, 0.32]	& 	2	& 	11	&	0.12	& 	[0.02, 0.25]	& 	3	& 	19	\\
10 (51)	&	0.29 (0.11)	& 	0.16	& 	[0.05, 0.33]	& 	2	& 	12	&	0.10	& 	[0.02, 0.19]	& 	3	& 	31	\\
9 (55)	&	0.29 (0.11)	& 	0.16	& 	[0.05, 0.32]	& 	2	& 	13	&	0.27	& 	[0.11, 0.42]	& 	1	& 	3	\\
12 (23)	&	0.26 (0.09)	& 	0.15	& 	[0.06, 0.3]	& 	2	& 	14	&	0.13	& 	[0.03, 0.29]	& 	3	& 	18	\\
97 (63)	&	-0.54 (0.44)	& 	0.16	& 	[0.01, 0.34]	& 	2	& 	15	&	0.14	& 	[0.01, 0.32]	& 	3	& 	15	\\
25 (59)	&	0.17 (0.21)	& 	0.14	& 	[0.02, 0.28]	& 	2	& 	16	&	0.15	& 	[0.02, 0.3]	& 	2	& 	13	\\
14 (59)	&	0.24 (0.08)	& 	0.14	& 	[0.05, 0.27]	& 	2	& 	17	&	0.12	& 	[0.04, 0.21]	& 	3	& 	20	\\
15 (56)	&	0.24 (0.09)	& 	0.14	& 	[0.05, 0.26]	& 	2	& 	18	&	0.11	& 	[0.03, 0.24]	& 	3	& 	21	\\
22 (63)	&	0.18 (0.19)	& 	0.14	& 	[0.02, 0.27]	& 	2	& 	19	&	0.11	& 	[0.01, 0.23]	& 	3	& 	26	\\
16 (49)	&	0.22 (0.15)	& 	0.14	& 	[0.03, 0.26]	& 	2	& 	20	&	0.11	& 	[0.02, 0.23]	& 	3	& 	24	\\
18 (51)	&	0.2 (0.14)	& 	0.13	& 	[0.03, 0.24]	& 	2	& 	21	&	0.09	& 	[0.02, 0.19]	& 	3	& 	30	\\
17 (48)	&	0.22 (0.1)	& 	0.13	& 	[0.04, 0.23]	& 	2	& 	22	&	0.11	& 	[0.02, 0.24]	& 	3	& 	25	\\
19 (55)	&	0.19 (0.14)	& 	0.13	& 	[0.02, 0.24]	& 	2	& 	23	&	0.24	& 	[0.05, 0.43]	& 	1	& 	4	\\
33 (24)	&	0.14 (0.17)	& 	0.12	& 	[0.02, 0.24]	& 	2	& 	24	&	0.11	& 	[0.01, 0.23]	& 	3	& 	23	\\
20 (70)	&	0.19 (0.12)	& 	0.12	& 	[0.03, 0.22]	& 	2	& 	25	&	0.09	& 	[0.01, 0.19]	& 	3	& 	39	\\
21 (80)	&	0.19 (0.08)	& 	0.12	& 	[0.03, 0.2]	& 	2	& 	26	&	0.07	& 	[0.01, 0.16]	& 	3	& 	57	\\
23 (55)	&	0.18 (0.08)	& 	0.11	& 	[0.04, 0.2]	& 	2	& 	27	&	0.19	& 	[0.07, 0.3]	& 	2	& 	7	\\
24 (53)	&	0.17 (0.06)	& 	0.11	& 	[0.04, 0.19]	& 	2	& 	28	&	0.09	& 	[0.03, 0.16]	& 	3	& 	34	\\
93 (59)	&	-0.12 (0.24)	& 	0.12	& 	[0.01, 0.25]	& 	2	& 	29	&	0.13	& 	[0.01, 0.28]	& 	3	& 	16	\\
26 (54)	&	0.17 (0.08)	& 	0.11	& 	[0.03, 0.19]	& 	2	& 	30	&	0.08	& 	[0.02, 0.16]	& 	3	& 	45	\\
28 (49)	&	0.17 (0.08)	& 	0.11	& 	[0.03, 0.19]	& 	2	& 	31	&	0.09	& 	[0.01, 0.17]	& 	3	& 	42	\\
66 (55)	&	0.05 (0.17)	& 	0.11	& 	[0.01, 0.23]	& 	2	& 	32	&	0.21	& 	[0.01, 0.42]	& 	2	& 	6	\\
45 (48)	&	0.11 (0.14)	& 	0.11	& 	[0.02, 0.21]	& 	2	& 	33	&	0.11	& 	[0.01, 0.24]	& 	3	& 	27	\\
39 (57)	&	0.13 (0.12)	& 	0.11	& 	[0.02, 0.21]	& 	2	& 	34	&	0.11	& 	[0.02, 0.22]	& 	3	& 	22	\\
32 (51)	&	0.14 (0.11)	& 	0.11	& 	[0.02, 0.2]	& 	2	& 	35	&	0.08	& 	[0.01, 0.16]	& 	3	& 	49	\\
34 (57)	&	0.14 (0.11)	& 	0.11	& 	[0.02, 0.2]	& 	2	& 	36	&	0.10	& 	[0.02, 0.21]	& 	3	& 	28	\\
29 (59)	&	0.15 (0.06)	& 	0.10	& 	[0.04, 0.17]	& 	2	& 	37	&	0.09	& 	[0.03, 0.17]	& 	3	& 	32	\\
30 (55)	&	0.15 (0.06)	& 	0.10	& 	[0.04, 0.17]	& 	2	& 	38	&	0.16	& 	[0.06, 0.26]	& 	2	& 	10	\\
31 (58)	&	0.14 (0.06)	& 	0.10	& 	[0.04, 0.16]	& 	2	& 	39	&	0.08	& 	[0.03, 0.15]	& 	3	& 	43	\\
38 (75)	&	0.13 (0.09)	& 	0.10	& 	[0.02, 0.18]	& 	2	& 	40	&	0.07	& 	[0.01, 0.15]	& 	3	& 	58	\\
35 (53)	&	0.14 (0.05)	& 	0.10	& 	[0.04, 0.16]	& 	2	& 	41	&	0.08	& 	[0.03, 0.14]	& 	3	& 	44	\\
37 (57)	&	0.14 (0.06)	& 	0.10	& 	[0.03, 0.16]	& 	2	& 	42	&	0.09	& 	[0.03, 0.16]	& 	3	& 	38	\\
36 (59)	&	0.14 (0.06)	& 	0.09	& 	[0.03, 0.16]	& 	2	& 	43	&	0.09	& 	[0.03, 0.16]	& 	3	& 	36	\\
95 (60)	&	-0.19 (0.22)	& 	0.10	& 	[0.01, 0.22]	& 	2	& 	44	&	0.09	& 	[0.01, 0.22]	& 	3	& 	37	\\
40 (50)	&	0.12 (0.08)	& 	0.09	& 	[0.02, 0.17]	& 	2	& 	45	&	0.08	& 	[0.01, 0.16]	& 	3	& 	51	\\
47 (62)	&	0.1 (0.1)	& 	0.09	& 	[0.02, 0.18]	& 	2	& 	46	&	0.08	& 	[0.01, 0.18]	& 	3	& 	47	\\
53 (53)	&	0.09 (0.1)	& 	0.09	& 	[0.01, 0.18]	& 	2	& 	47	&	0.09	& 	[0.02, 0.17]	& 	3	& 	33	\\
44 (53)	&	0.11 (0.07)	& 	0.09	& 	[0.02, 0.16]	& 	2	& 	48	&	0.08	& 	[0.02, 0.15]	& 	3	& 	40	\\
42 (56)	&	0.12 (0.06)	& 	0.09	& 	[0.02, 0.15]	& 	2	& 	49	&	0.08	& 	[0.02, 0.16]	& 	3	& 	41	\\
43 (53)	&	0.12 (0.07)	& 	0.09	& 	[0.02, 0.15]	& 	2	& 	50	&	0.08	& 	[0.02, 0.15]	& 	3	& 	46	\\
41 (57)	&	0.12 (0.05)	& 	0.09	& 	[0.03, 0.14]	& 	2	& 	51	&	0.08	& 	[0.03, 0.14]	& 	3	& 	52	\\
49 (23)	&	0.09 (0.09)	& 	0.09	& 	[0.02, 0.16]	& 	2	& 	52	&	0.09	& 	[0.02, 0.2]	& 	3	& 	35	\\
69 (73)	&	0.03 (0.11)	& 	0.08	& 	[0.01, 0.17]	& 	2	& 	53	&	0.07	& 	[0, 0.16]	& 	3	& 	64	\\
50 (53)	&	0.09 (0.07)	& 	0.08	& 	[0.02, 0.15]	& 	2	& 	54	&	0.08	& 	[0.02, 0.14]	& 	3	& 	50	\\
46 (53)	&	0.1 (0.05)	& 	0.08	& 	[0.02, 0.14]	& 	2	& 	55	&	0.07	& 	[0.02, 0.13]	& 	3	& 	54	\\
51 (51)	&	0.09 (0.07)	& 	0.08	& 	[0.02, 0.14]	& 	2	& 	56	&	0.06	& 	[0.01, 0.13]	& 	3	& 	67	\\
48 (51)	&	0.1 (0.05)	& 	0.08	& 	[0.02, 0.13]	& 	2	& 	57	&	0.06	& 	[0.01, 0.11]	& 	3	& 	72	\\
96 (80)	&	-0.27 (0.2)	& 	0.08	& 	[0.01, 0.19]	& 	2	& 	58	&	0.08	& 	[0, 0.19]	& 	3	& 	55	\\
68 (56)	&	0.04 (0.09)	& 	0.07	& 	[0.01, 0.15]	& 	2	& 	59	&	0.08	& 	[0.01, 0.17]	& 	3	& 	48	\\
63 (63)	&	0.05 (0.08)	& 	0.07	& 	[0.01, 0.15]	& 	2	& 	60	&	0.06	& 	[0.01, 0.14]	& 	3	& 	65	\\
60 (50)	&	0.06 (0.08)	& 	0.07	& 	[0.01, 0.14]	& 	2	& 	61	&	0.07	& 	[0.01, 0.14]	& 	3	& 	59	\\
74 (54)	&	0 (0.11)	& 	0.07	& 	[0.01, 0.16]	& 	2	& 	62	&	0.07	& 	[0.01, 0.15]	& 	3	& 	56	\\
58 (24)	&	0.06 (0.07)	& 	0.07	& 	[0.01, 0.14]	& 	2	& 	63	&	0.07	& 	[0.01, 0.14]	& 	3	& 	61	\\
52 (58)	&	0.09 (0.04)	& 	0.07	& 	[0.02, 0.12]	& 	2	& 	64	&	0.06	& 	[0.02, 0.11]	& 	3	& 	63	\\
65 (60)	&	0.05 (0.08)	& 	0.07	& 	[0.01, 0.15]	& 	2	& 	65	&	0.06	& 	[0.01, 0.14]	& 	3	& 	70	\\
54 (54)	&	0.08 (0.05)	& 	0.07	& 	[0.02, 0.12]	& 	2	& 	66	&	0.06	& 	[0.01, 0.11]	& 	3	& 	71	\\
72 (42)	&	0.02 (0.09)	& 	0.07	& 	[0.01, 0.15]	& 	2	& 	67	&	0.05	& 	[0, 0.12]	& 	3	& 	84	\\
61 (51)	&	0.06 (0.07)	& 	0.07	& 	[0.01, 0.13]	& 	2	& 	68	&	0.06	& 	[0.01, 0.12]	& 	3	& 	74	\\
55 (52)	&	0.08 (0.03)	& 	0.07	& 	[0.02, 0.11]	& 	2	& 	69	&	0.05	& 	[0.01, 0.1]	& 	3	& 	79	\\
78 (52)	&	-0.01 (0.1)	& 	0.07	& 	[0.01, 0.15]	& 	2	& 	70	&	0.07	& 	[0.01, 0.14]	& 	3	& 	62	\\
71 (53)	&	0.02 (0.09)	& 	0.07	& 	[0.01, 0.14]	& 	2	& 	71	&	0.07	& 	[0.01, 0.15]	& 	3	& 	53	\\
56 (53)	&	0.07 (0.04)	& 	0.06	& 	[0.02, 0.11]	& 	2	& 	72	&	0.06	& 	[0.02, 0.11]	& 	3	& 	66	\\
64 (58)	&	0.05 (0.07)	& 	0.06	& 	[0.01, 0.13]	& 	2	& 	73	&	0.07	& 	[0.01, 0.13]	& 	3	& 	60	\\
59 (75)	&	0.06 (0.06)	& 	0.06	& 	[0.01, 0.12]	& 	2	& 	74	&	0.05	& 	[0.01, 0.11]	& 	3	& 	80	\\
57 (58)	&	0.06 (0.05)	& 	0.06	& 	[0.01, 0.11]	& 	2	& 	75	&	0.06	& 	[0.01, 0.11]	& 	3	& 	69	\\
81 (80)	&	-0.01 (0.09)	& 	0.06	& 	[0.01, 0.13]	& 	2	& 	76	&	0.06	& 	[0, 0.13]	& 	3	& 	77	\\
94 (80)	&	-0.18 (0.14)	& 	0.06	& 	[0, 0.15]	& 	2	& 	77	&	0.06	& 	[0, 0.15]	& 	3	& 	73	\\
79 (42)	&	-0.01 (0.08)	& 	0.06	& 	[0.01, 0.13]	& 	2	& 	78	&	0.04	& 	[0, 0.11]	& 	3	& 	86	\\
84 (63)	&	-0.02 (0.09)	& 	0.06	& 	[0.01, 0.13]	& 	2	& 	79	&	0.06	& 	[0, 0.13]	& 	3	& 	75	\\
85 (51)	&	-0.02 (0.09)	& 	0.06	& 	[0.01, 0.13]	& 	2	& 	80	&	0.06	& 	[0.01, 0.12]	& 	3	& 	76	\\
62 (54)	&	0.06 (0.04)	& 	0.05	& 	[0.01, 0.1]	& 	2	& 	81	&	0.05	& 	[0.01, 0.09]	& 	3	& 	81	\\
73 (59)	&	0.01 (0.07)	& 	0.05	& 	[0.01, 0.11]	& 	2	& 	82	&	0.06	& 	[0.01, 0.13]	& 	3	& 	68	\\
67 (87)	&	0.04 (0.04)	& 	0.05	& 	[0.01, 0.1]	& 	2	& 	83	&	0.05	& 	[0.01, 0.1]	& 	3	& 	82	\\
83 (51)	&	-0.01 (0.06)	& 	0.05	& 	[0, 0.1]	& 	3	& 	84	&	0.05	& 	[0, 0.1]	& 	3	& 	83	\\
80 (53)	&	-0.01 (0.06)	& 	0.05	& 	[0, 0.1]	& 	3	& 	85	&	0.05	& 	[0, 0.11]	& 	3	& 	78	\\
88 (42)	&	-0.03 (0.07)	& 	0.04	& 	[0, 0.1]	& 	3	& 	86	&	0.03	& 	[0, 0.09]	& 	4	& 	92	\\
77 (52)	&	0 (0.05)	& 	0.04	& 	[0, 0.09]	& 	3	& 	87	&	0.04	& 	[0, 0.09]	& 	3	& 	85	\\
87 (42)	&	-0.03 (0.06)	& 	0.04	& 	[0, 0.1]	& 	3	& 	88	&	0.03	& 	[0, 0.09]	& 	4	& 	93	\\
70 (53)	&	0.02 (0.03)	& 	0.03	& 	[0, 0.07]	& 	3	& 	89	&	0.04	& 	[0.01, 0.08]	& 	4	& 	87	\\
82 (20)	&	-0.01 (0.04)	& 	0.03	& 	[0, 0.08]	& 	3	& 	90	&	0.03	& 	[0, 0.08]	& 	4	& 	90	\\
92 (70)	&	-0.11 (0.07)	& 	0.03	& 	[0, 0.09]	& 	3	& 	91	&	0.04	& 	[0, 0.1]	& 	4	& 	88	\\
75 (51)	&	0 (0.04)	& 	0.03	& 	[0, 0.07]	& 	3	& 	92	&	0.03	& 	[0, 0.07]	& 	4	& 	89	\\
86 (49)	&	-0.03 (0.04)	& 	0.03	& 	[0, 0.07]	& 	3	& 	93	&	0.03	& 	[0, 0.08]	& 	4	& 	91	\\
76 (54)	&	0 (0.03)	& 	0.02	& 	[0, 0.05]	& 	3	& 	94	&	0.03	& 	[0, 0.06]	& 	4	& 	95	\\
90 (75)	&	-0.05 (0.04)	& 	0.02	& 	[0, 0.06]	& 	3	& 	95	&	0.03	& 	[0, 0.07]	& 	4	& 	94	\\
91 (51)	&	-0.07 (0.04)	& 	0.02	& 	[0, 0.05]	& 	3	& 	96	&	0.02	& 	[0, 0.06]	& 	4	& 	96	\\
89 (73)	&	-0.03 (0.03)	& 	0.02	& 	[0, 0.04]	& 	3	& 	97	&	0.02	& 	[0, 0.05]	& 	4	& 	97	\\
									
\end{longtable}	
\parbox{\textwidth}{\small
\vspace{1eX}
      \textit{Notes:} This table reports estimated contact penalties and the results of Empirical Bayes and grading exercises for each firm in the sample. Firms are labeled by their raw contact penalty rank, with their industry (2-digit SIC code) shown in parentheses. The column $\theta_i$ reports contact penalties with a job-clustered standard error in parentheses. The remaining columns report posterior means (Post. mean), 95\% credible intervals (Post. CI), assigned grades using $\lambda=.25$ (Grd), and Condorcet ranks (Cond. rank), which are grades under $\lambda=1$, in the baseline model and the model with industry effects.}
\end{center}

\FloatBarrier

\begin{figure}[ht!]
    \centering
    \caption{Contact rates, standard errors, and name grades}
    \includegraphics[width=\textwidth]{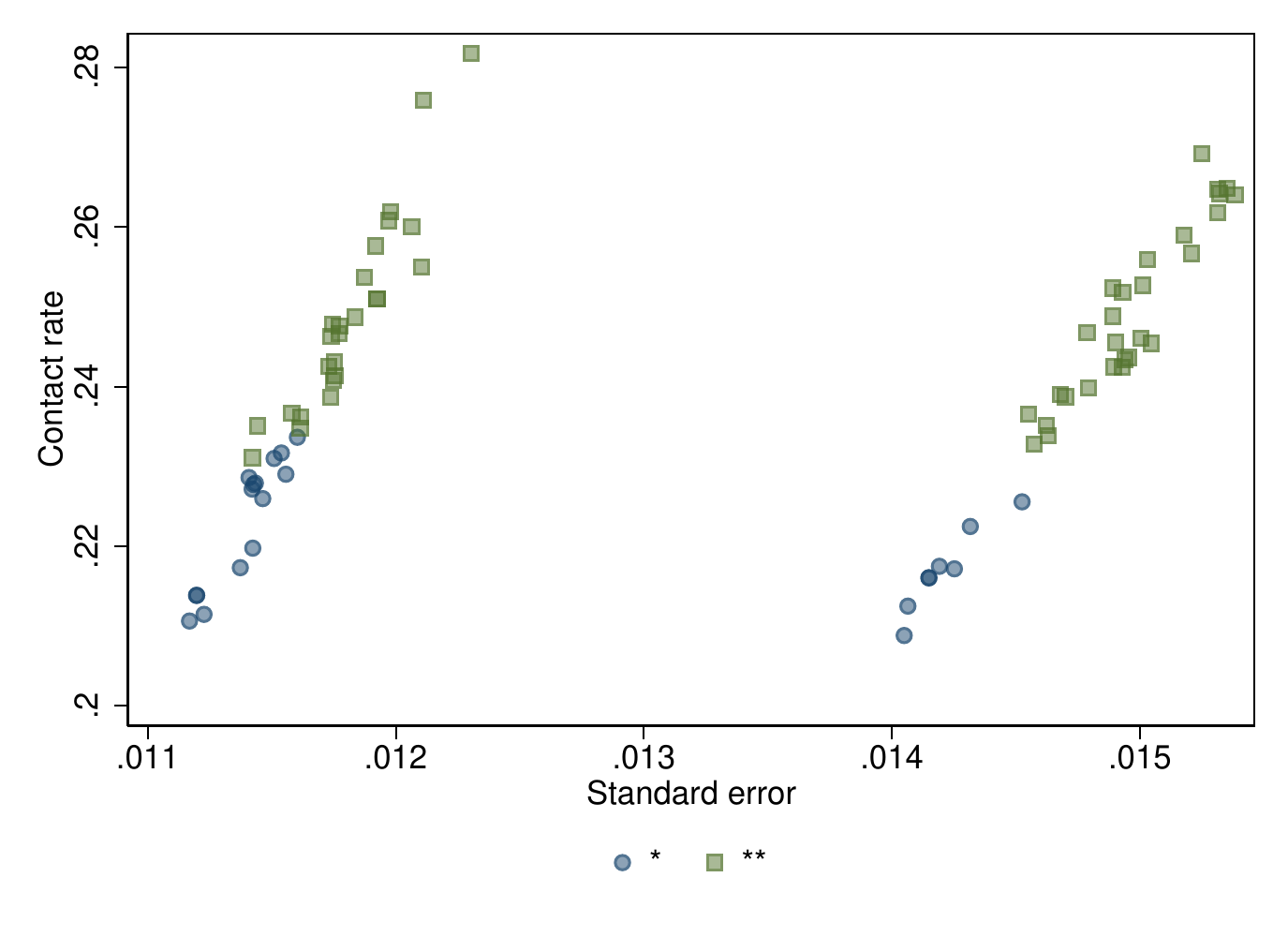}
    \label{fig:grade_frontiers_names}
\parbox{\textwidth}{\small
\vspace{1eX}\emph{Notes}: This figure plots the estimated contact rates for each name against its standard error. The shape and color of each point indicate the grade assigned to the name using the same specification as Figure \ref{fig: names_ranks}.}
\end{figure}

\begin{figure}[ht!]
    \centering
    \caption{Predictive power of grades name for race and sex labels}
    \begin{tabular}{c}
    a) Pseudo $R^2$ \\
    \includegraphics[width=0.8\textwidth]{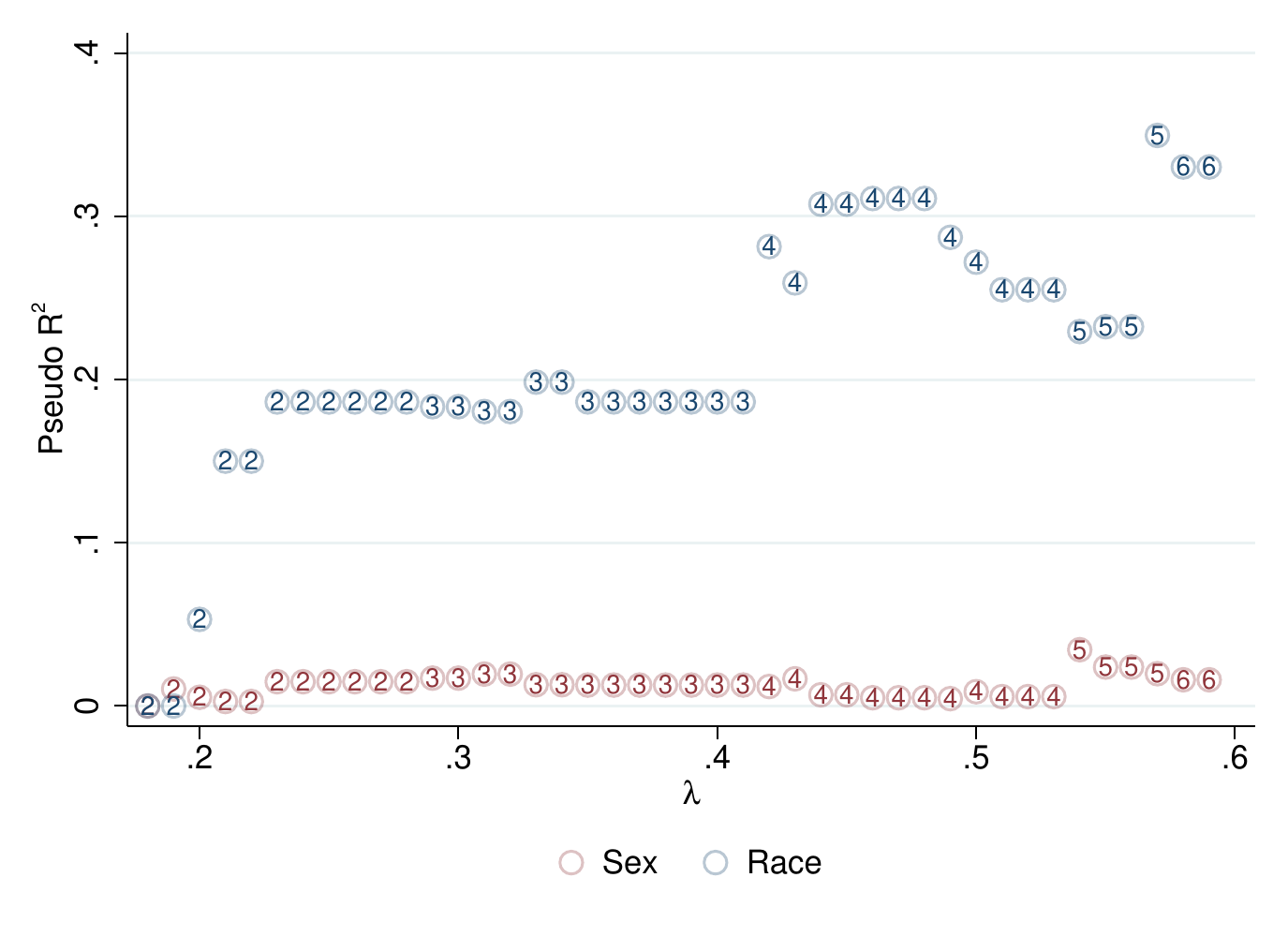} \\
    b) Area under the curve \\
    \includegraphics[width=0.8\textwidth]{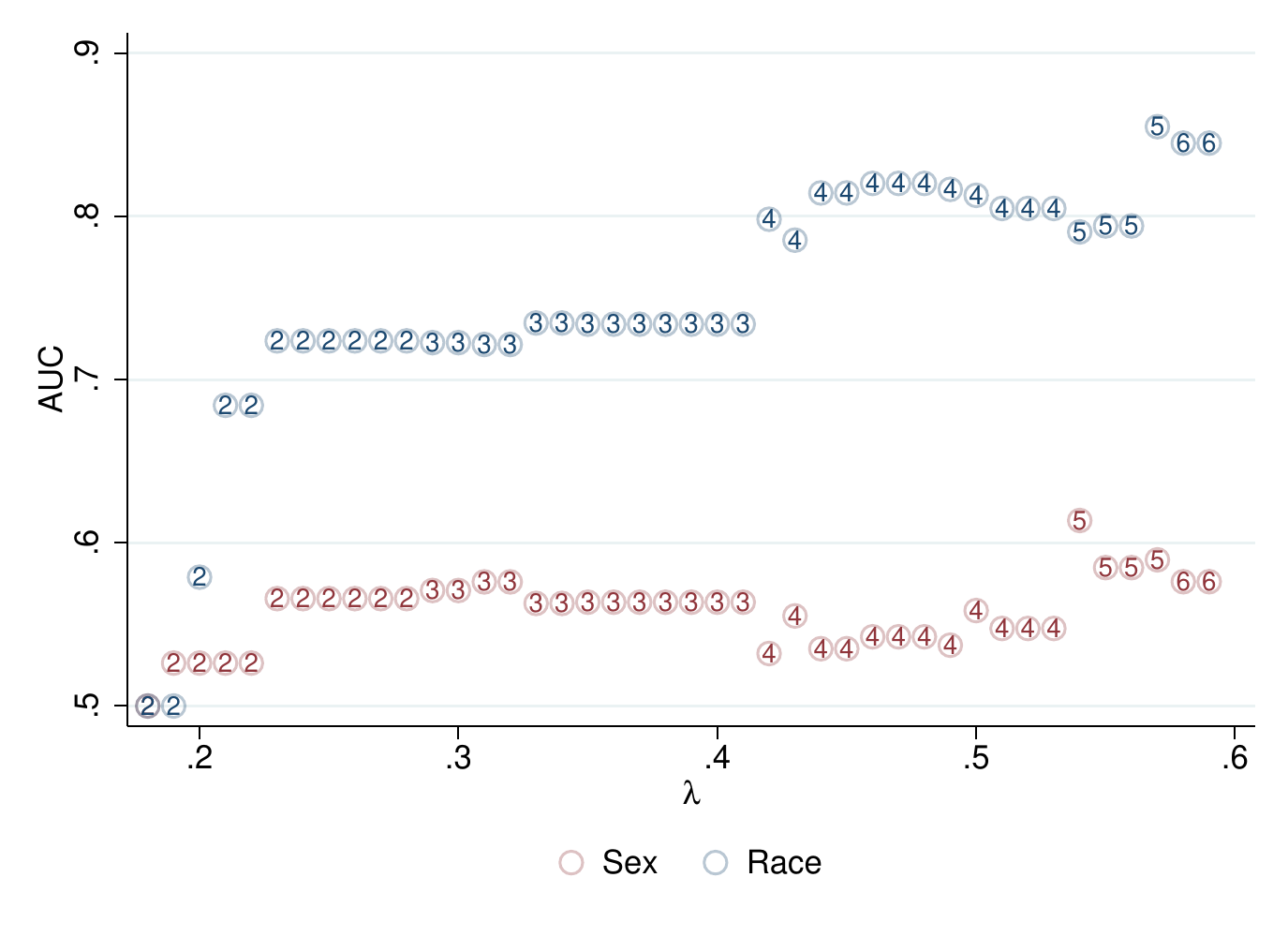}
    \end{tabular}
    \label{fig:pseudoR2}
\parbox{\textwidth}{\small
\vspace{1eX}\emph{Notes}: This figure plots the psuedo-$R^2$ (panel a) and AUC (panel b) for a series of logistic regressions using an indicator for the race or sex of the name as the outcome and dummies for assigned grades as the explanatory variables for an intermediate range of $\lambda$. The number shown indicates the number of grades assigned.}
\end{figure}

\begin{figure}
\centering
\caption{Square weighted name ranking exercises}
    \label{fig: names_res_sq}
\begin{adjustbox}{center}
\begin{tabular}{c c}
\multicolumn{2}{c}{a) Grades and discordance} \\
 \includegraphics[width=0.8\textwidth]{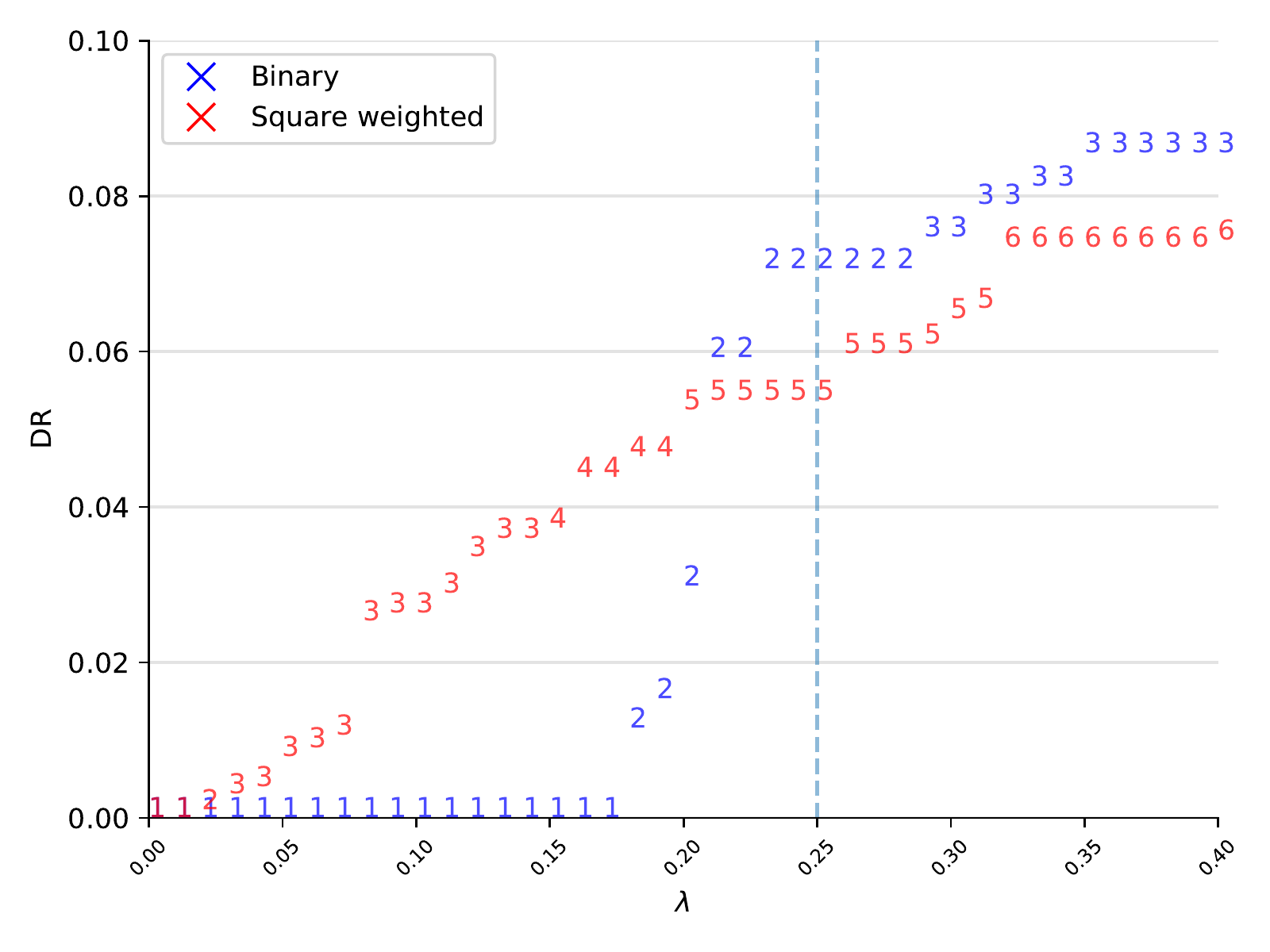} \\
\multicolumn{2}{c}{b) Reporting possibilities } \\
\multicolumn{2}{c}{\includegraphics[width=0.8\textwidth]{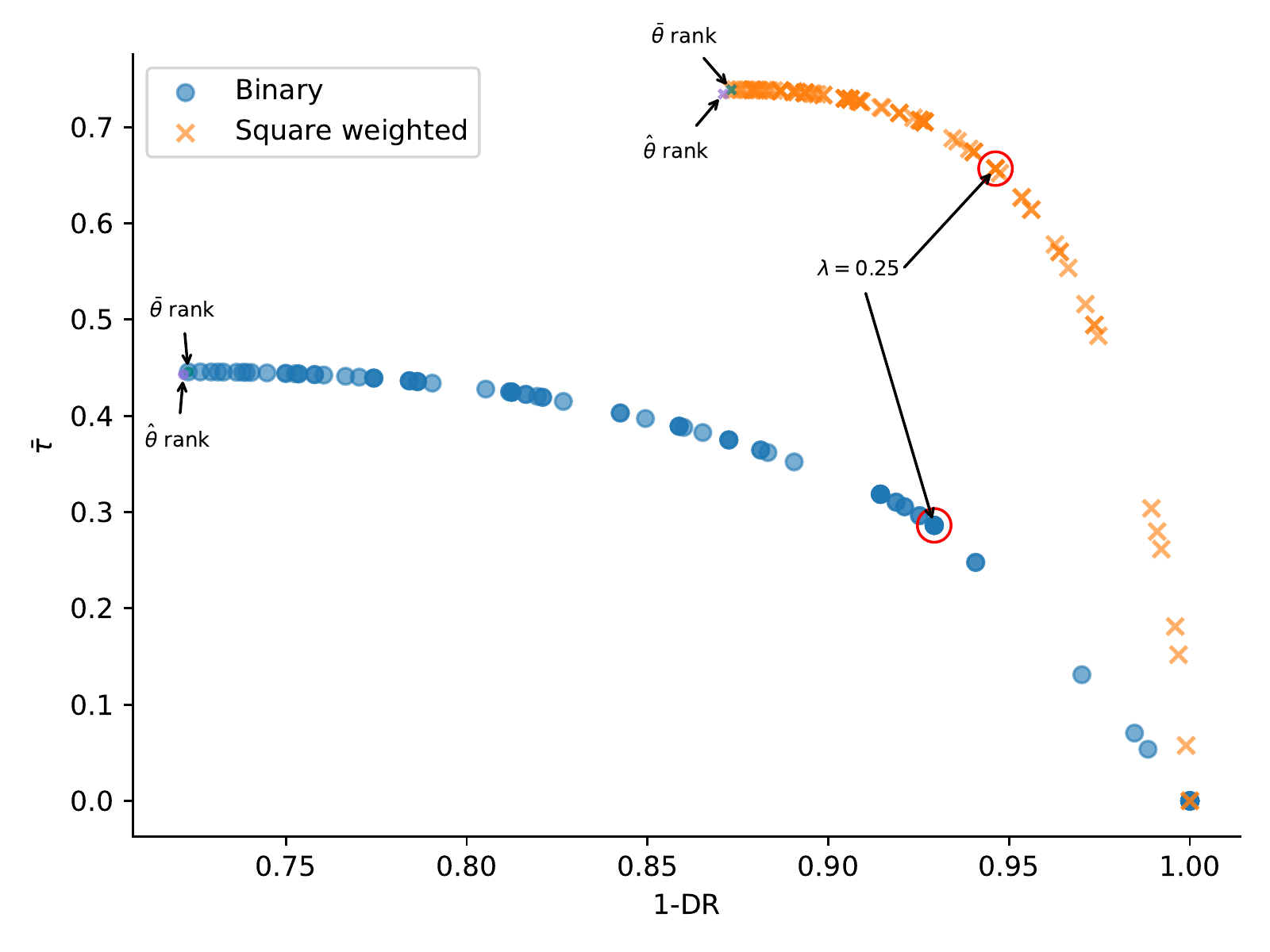}}
\end{tabular}
\end{adjustbox}
\parbox{\textwidth}{\small
\vspace{1eX}\emph{Notes}: This figure presents additional  results from grading contact rates for names under binary and square-weighted loss. In both panels, posteriors are computed using the log-spline estimate plotted in Figure \ref{fig: name_g} as the prior. Panel (a) shows estimated Discordance Rates (DR) for an intermediate range of $\lambda$ under binary and square-weighted loss. Panel (b) plots the expectation of Kendall's $\tau$ rank correlation between true contact rates and grades against Discordance Rates (DR) for a range of grades indexed by $\lambda$. Red circles highlight the DR and expected $\tau$ corresponding to $\lambda = 0.25$. ``$\hat \theta$ rank'' refers to ranks based upon point estimates. ``$\bar \theta$ rank'' refers to ranks based upon Empirical Bayes posterior means.}
\end{figure}

\begin{figure}
\centering
\caption{Posterior means and grades for names under square-weighted loss}
    \label{fig: names_ranks_sq_weighted}
\begin{adjustbox}{center}
\includegraphics[width=\textwidth]{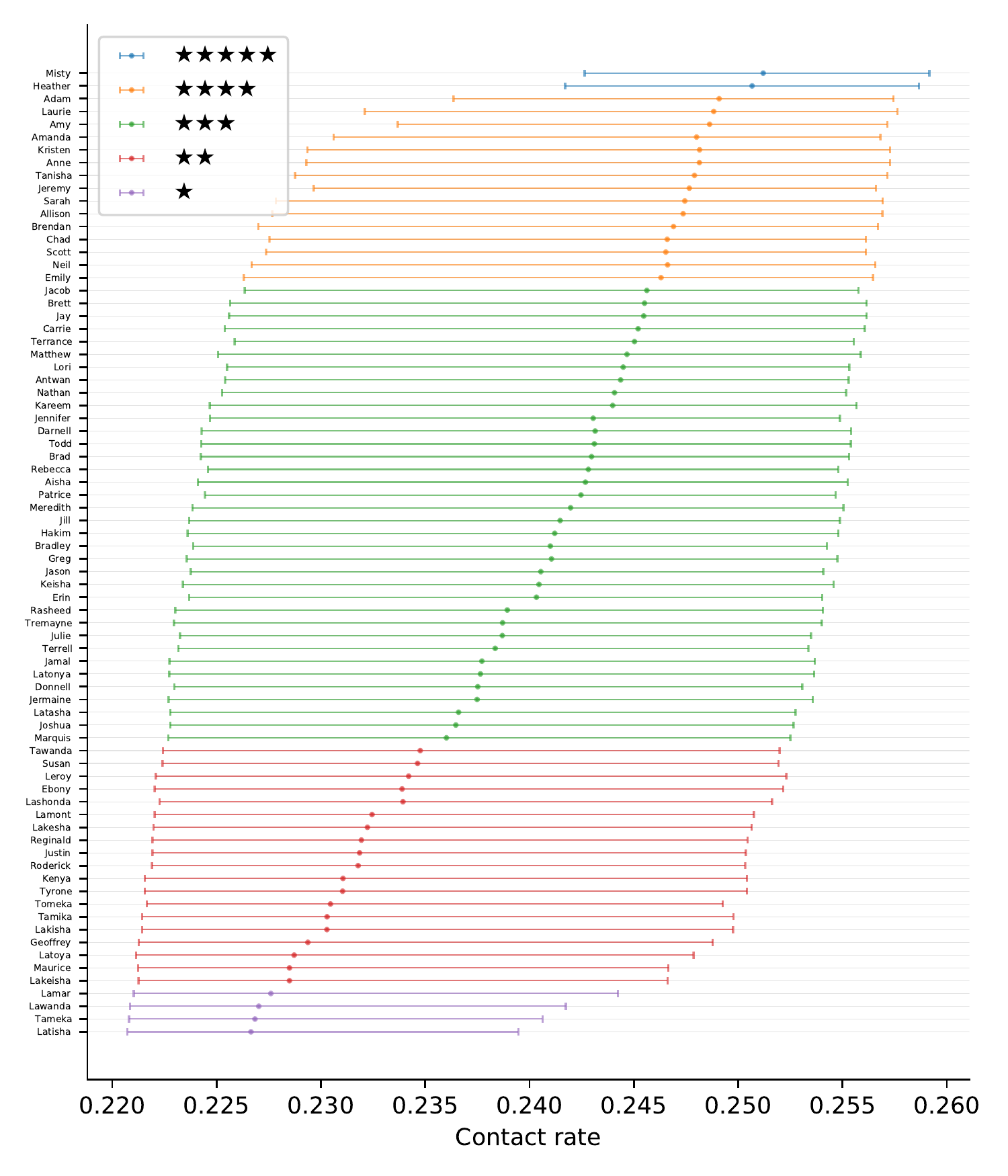}
\end{adjustbox}
\parbox{\textwidth}{\small
\vspace{1eX}\emph{Notes}: This figure shows posterior mean contact rates, 95\% credible intervals, and assigned grades for names. Results are shown for $\lambda = 0.25$, implying an 80\% threshold for posterior ranking probabilities. Names are ordered by their rank under $\lambda = 1$, when each name is assigned its own grade.}
\end{figure}

\begin{figure}[ht!]
    \centering
    \caption{Predictive power of square-weighted name grades for race and sex labels}
    \begin{tabular}{c}
    a) Pseudo $R^2$ \\
    \includegraphics[width=0.8\textwidth]{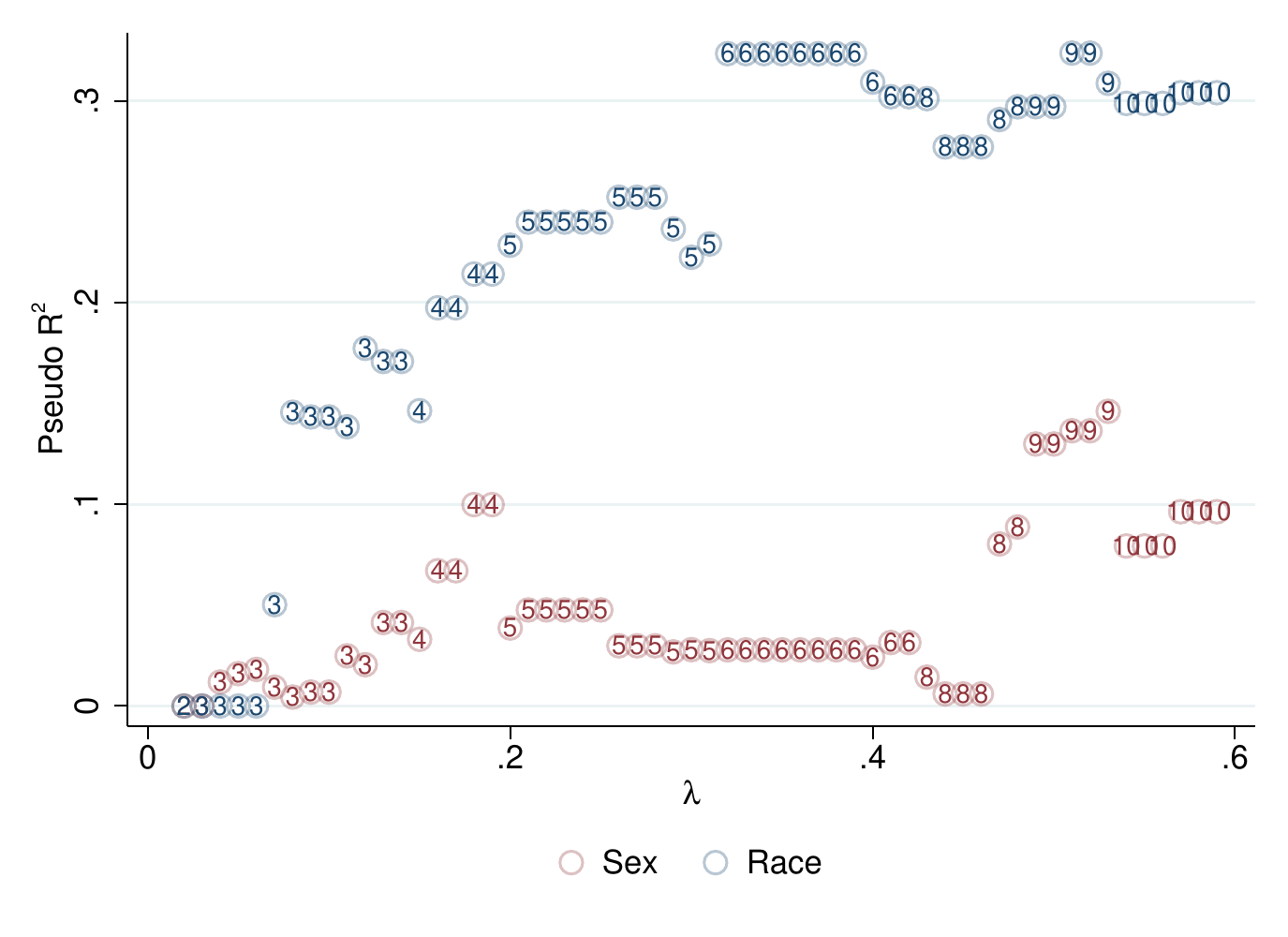} \\
    b) Area under the curve \\
    \includegraphics[width=0.8\textwidth]{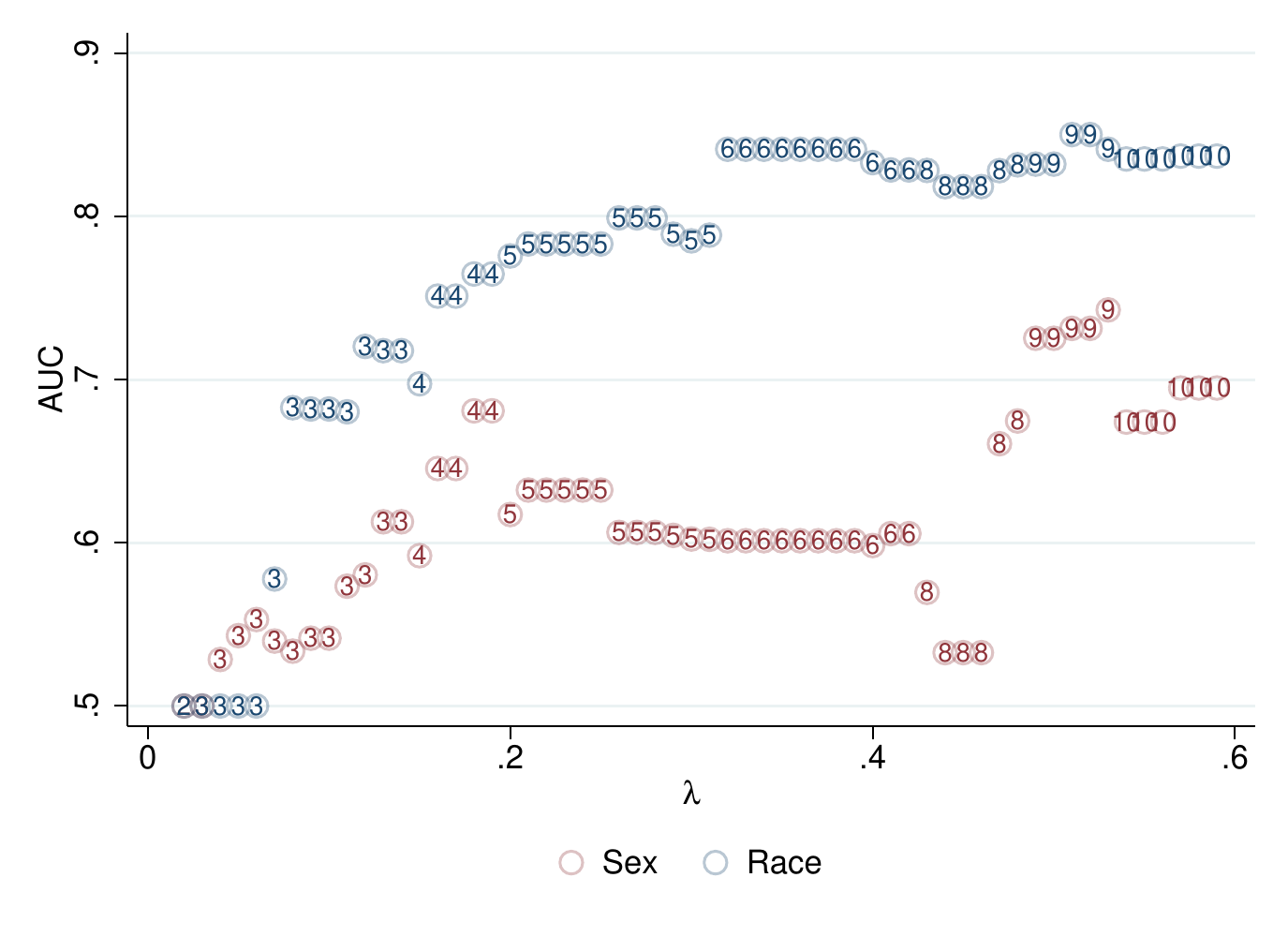}
    \end{tabular}
    \label{fig:pseudoR2_sq}
\parbox{\textwidth}{\small
\vspace{1eX}\emph{Notes}: This figure plots the psuedo-$R^2$ (panel a) and AUC (panel b) for a series of logistic regressions using an indicator for the race or sex of the name as the outcome and dummies for assigned grades under square-weighted loss as the explanatory variables for an intermediate range of $\lambda$. The number shown indicates the number of grades assigned.}
\end{figure}

\begin{figure}[ht!]
    \centering
    \caption{Unadjusted and studentized contact gaps against standard errors}
    \begin{tabular}{c}
    a) $\hat{\theta}_i$ vs. $s_i$ \\
    \includegraphics[width=.8\textwidth]{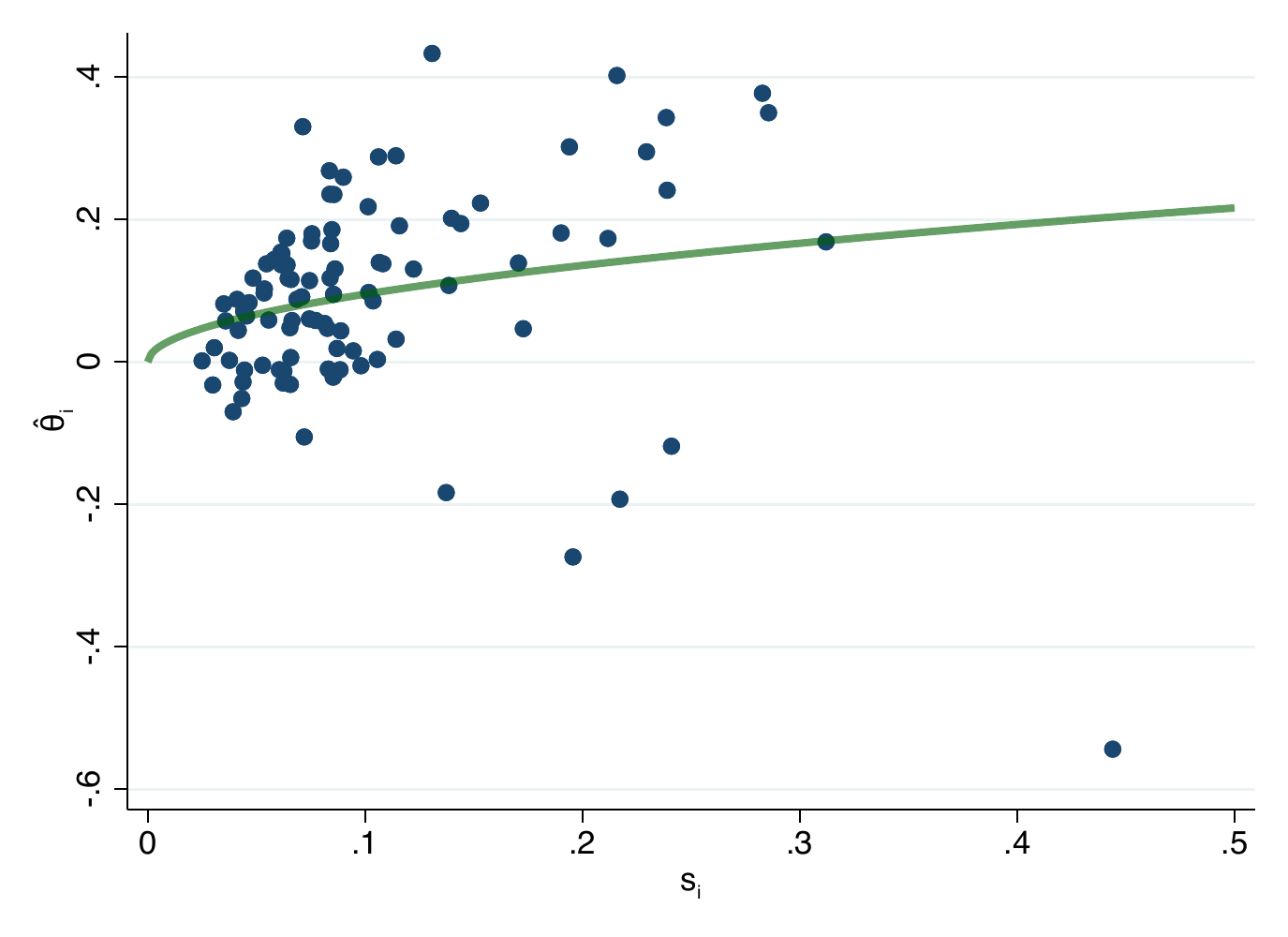} \\
    b) $\hat{T}_i$ vs. $s_i$ \\
    \includegraphics[width=.8\textwidth]{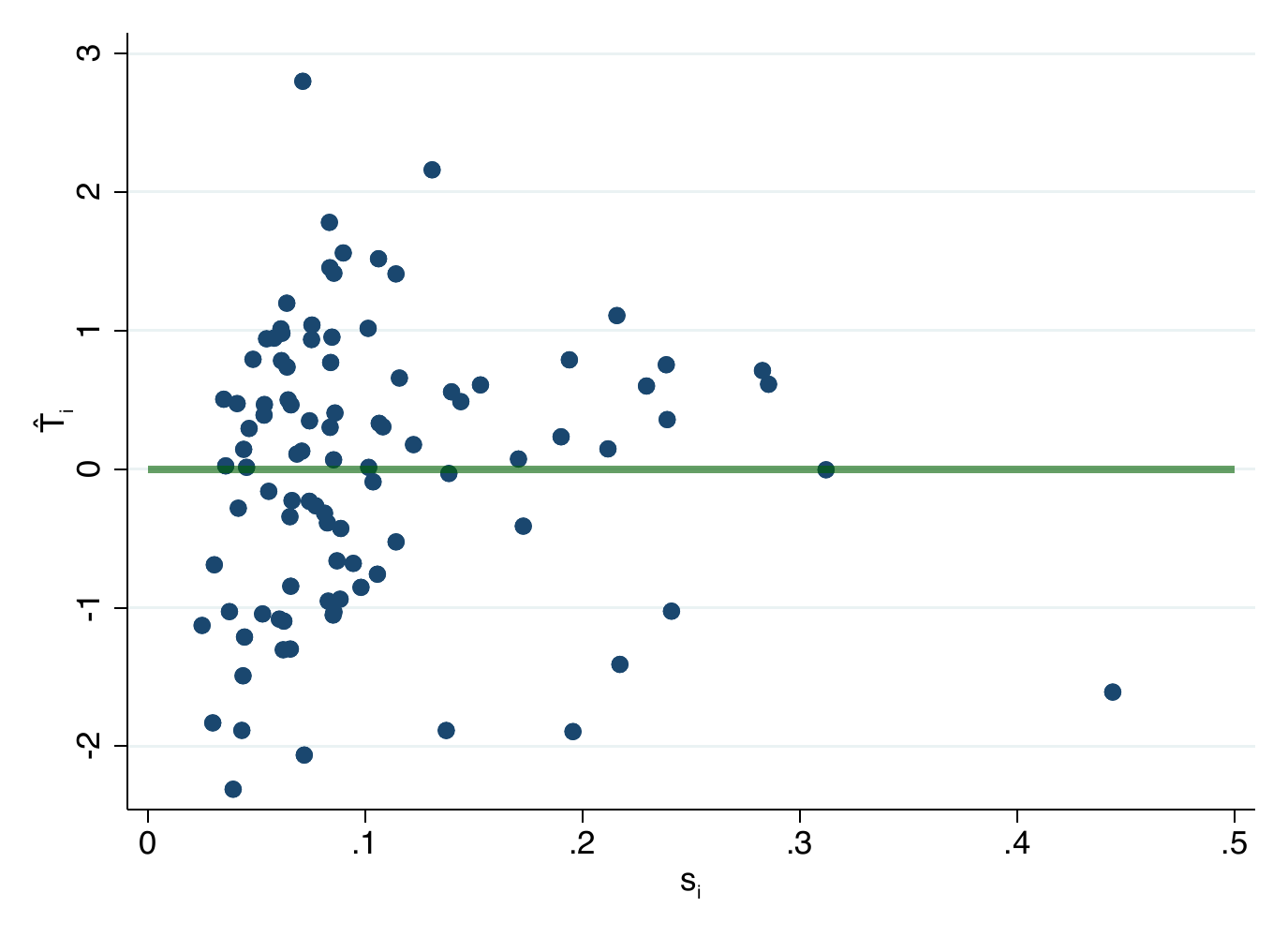} \\
    \end{tabular}
    \label{fig: theta_vs_s}
\parbox{\textwidth}{\small
\vspace{1eX}\emph{Notes}: Panel (a) of this figure plots estimated contact penalties against their standard errors. The green line plots the conditional mean of $\theta_i$ given $s_i$ implied by the GMM estimates. Panel (b) plots studentized contact gaps $\hat{T}_i$ against standard errors. The green line plots the relationship implied by the model.}
\end{figure}

\begin{figure}[ht!]
    \centering
    \caption{Contact penalties, standard errors, and report card grades}
    \includegraphics[width=\textwidth]{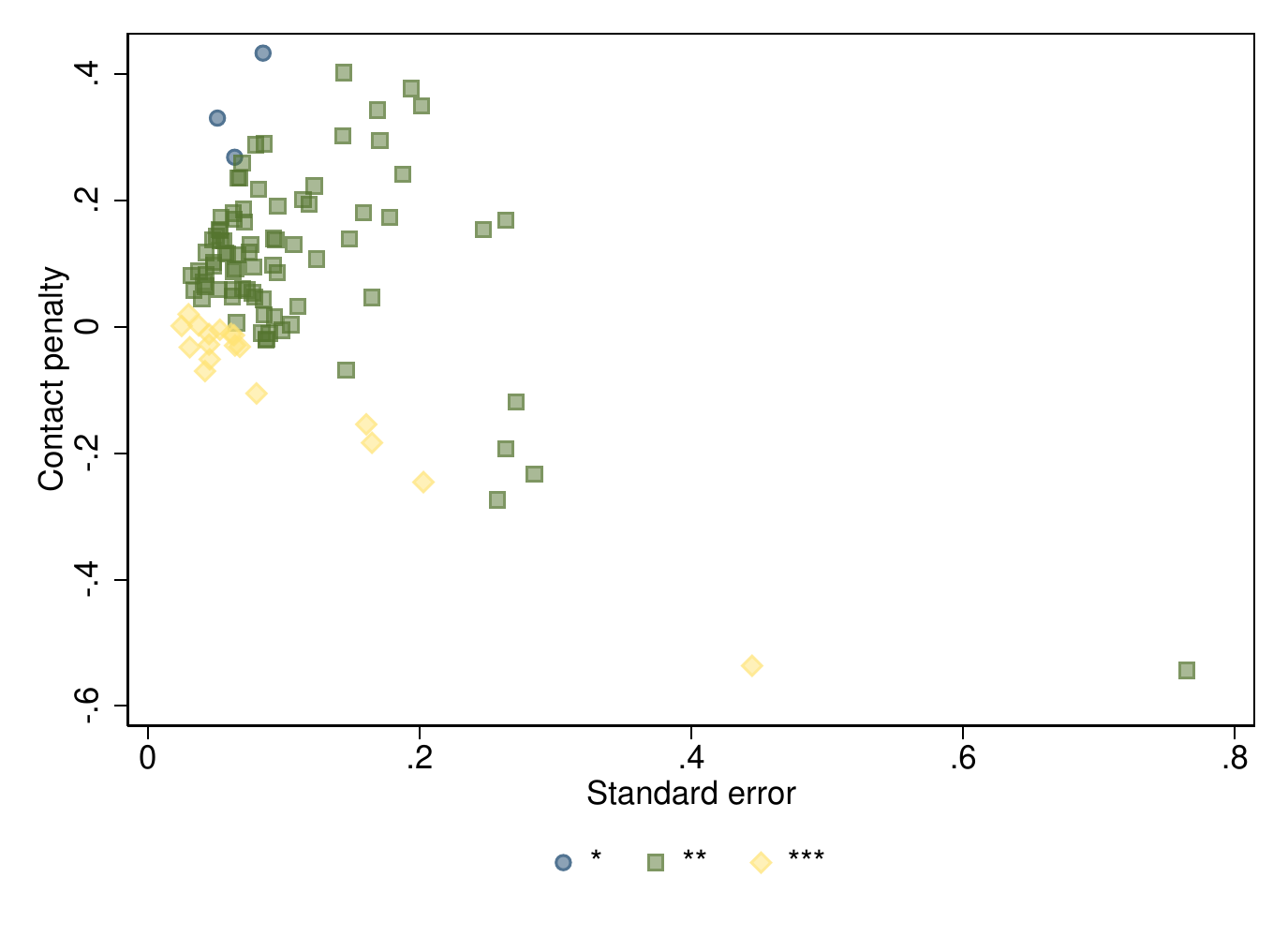}
    \label{fig:grade_frontiers}
\parbox{\textwidth}{\small
\vspace{1eX}\emph{Notes}: This figure plots the estimated contact penalty for a Black name at each firm against the standard error of the contact penalty estimate. The shape and color of each point indicate the grade assigned to the firm using the same specification as Figure \ref{fig:poisson_binary}.}
\end{figure}

\begin{figure}[ht!]
    \centering
    \caption{Grades and discordance as a function of $\lambda$}
    \includegraphics[width=\textwidth]{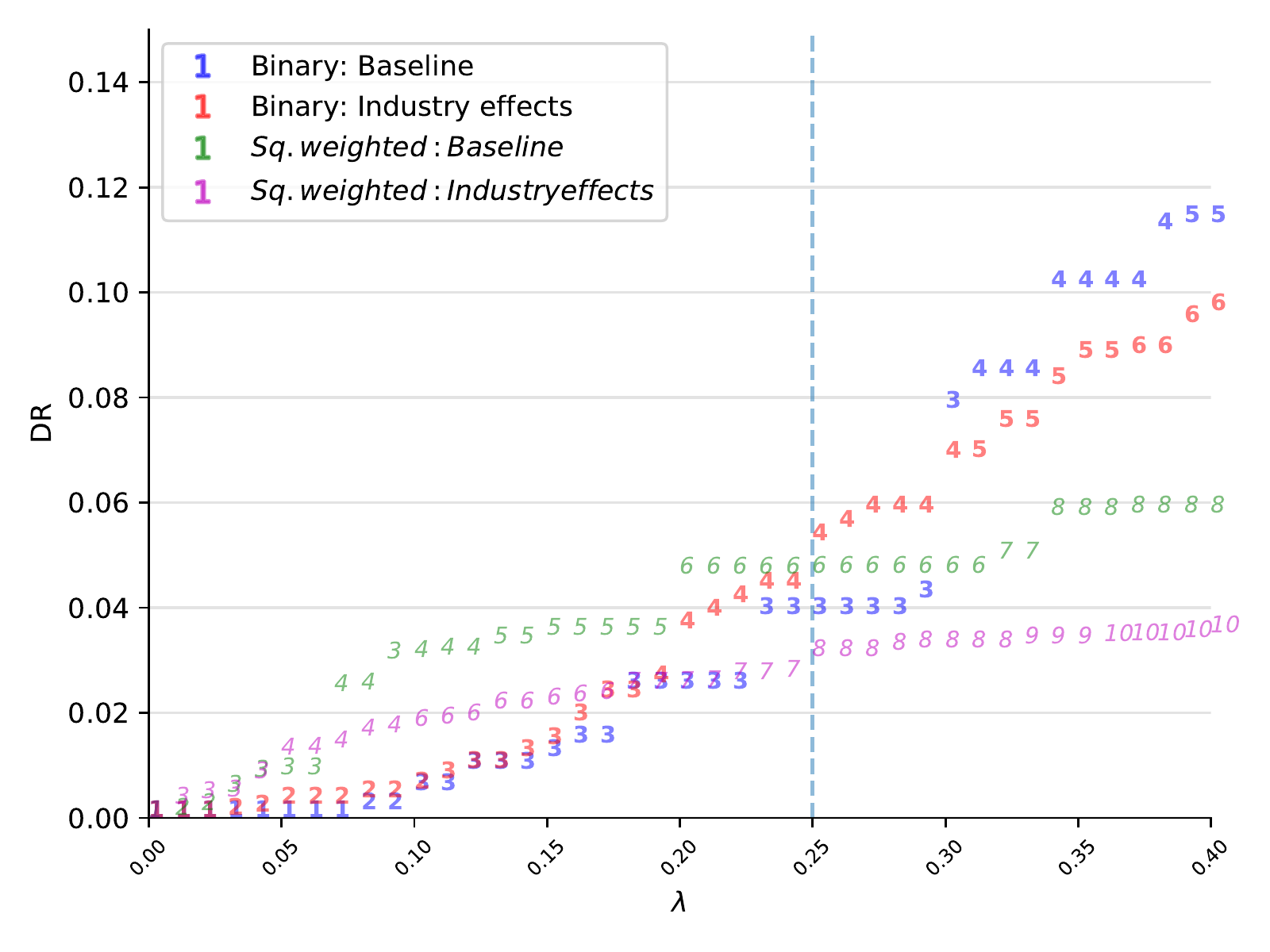}
    \label{fig:poisson_binary_lambda_sq}
\parbox{\textwidth}{\small
\vspace{1eX}\emph{Notes}: This figure shows estimated Discordance Rates (DR) as a function of $\lambda$ for both binary and square-weighted loss. The number on each point indicates the number of unique grades in the underlying grading scheme. The vertical dashed line shows results for the benchmark case of $\lambda=0.25$.}
\end{figure}

\begin{figure}[ht!]
    \centering
    \caption{Reporting possibilities}
\includegraphics[width=\textwidth]{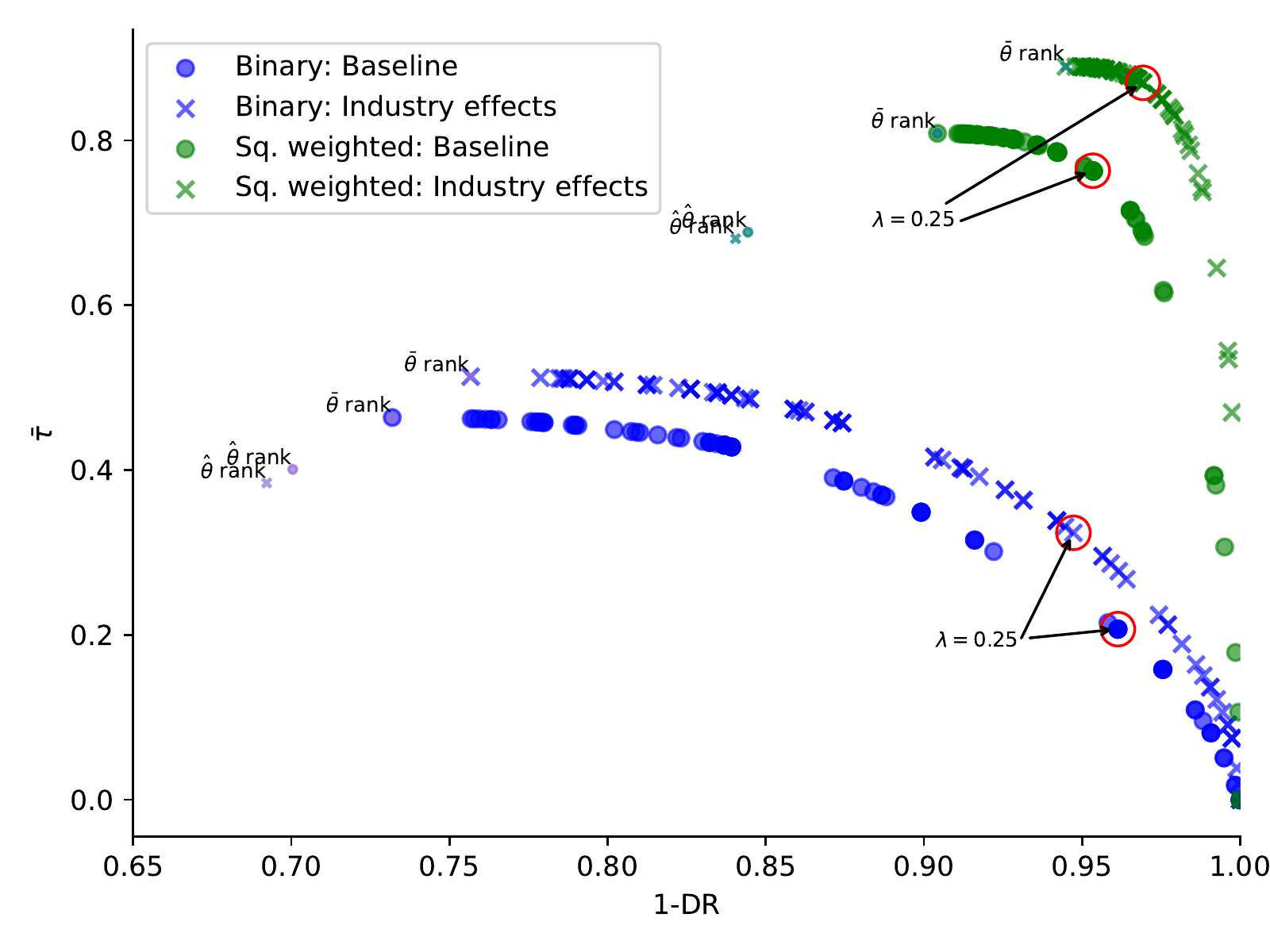}
    \label{fig:PPF_sq}
\parbox{\textwidth}{\small
\vspace{1eX}\emph{Notes}: This figure shows the expectation of Kendall's $\tau$ rank correlation between $\theta$ and assigned grades (labeled $\bar{\tau}$) against Discordance Rates (DR) for a range of grades indexed by $\lambda$ and for both binary and square-weighted loss.  Red circles highlight the DR and $\bar{\tau}$ corresponding to $\lambda = 0.25$. ``$\hat \theta$ rank'' refers to ranks based upon point estimates. ``$\bar \theta$ rank'' refers to ranks based upon Empirical Bayes posterior means.}
\end{figure}

\begin{figure}[ht!]
    \centering
    \caption{Square-weighted loss: Posterior means and grades}
    \includegraphics[width=\textwidth]{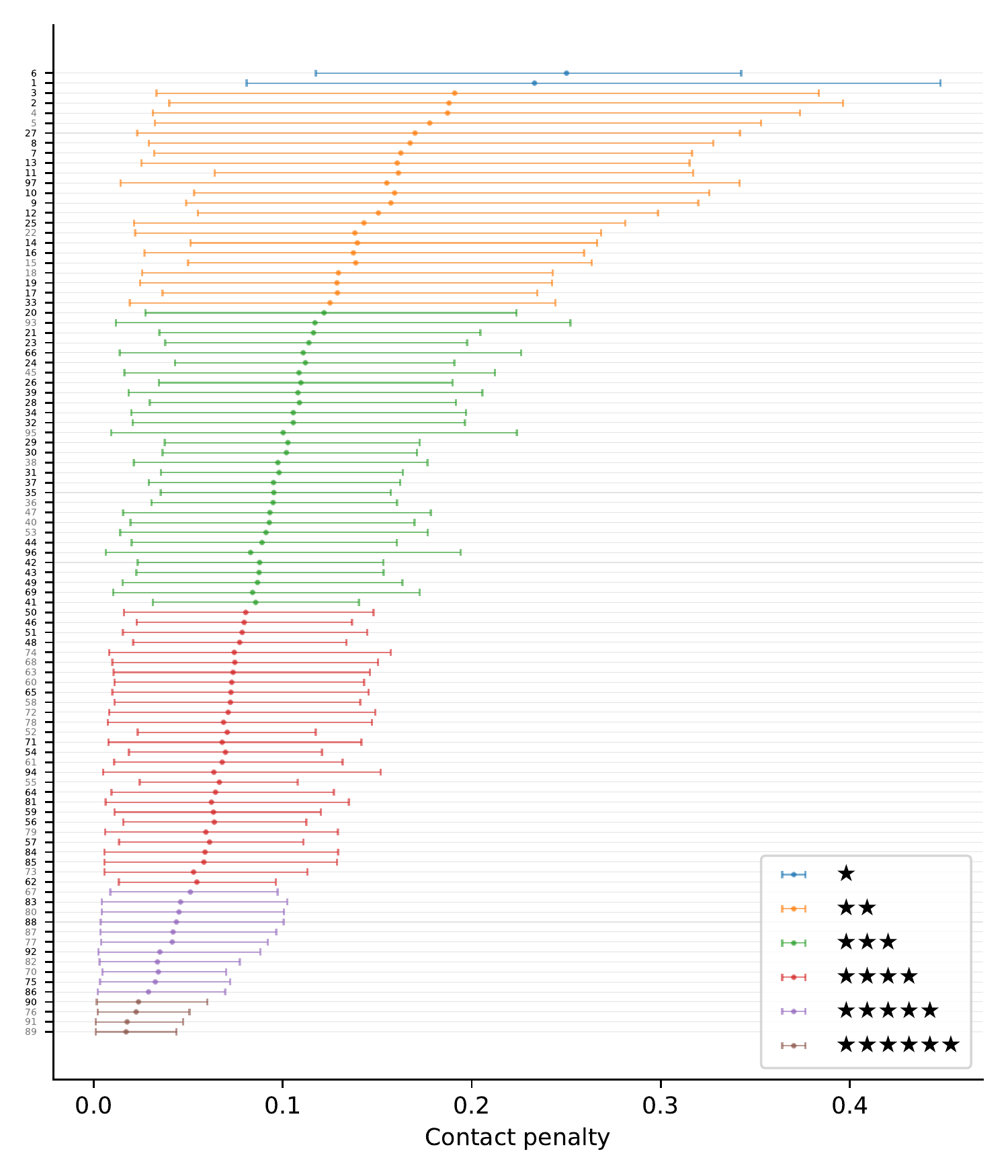}
    \label{fig:poisson_squared}
\parbox{\textwidth}{\small
\vspace{1eX}\emph{Notes}: This figure shows posterior mean proportional contact penalties, 95\% credible intervals, and assigned grades under square-weighted loss. Results are shown for $\lambda = 0.25$. Firms are ordered by their rank under $\lambda = 1$, when each firm is assigned its own grade, and labeled by the raw proportional contact gap rank, with \#1 showing the largest gap in favor of white applicants. Firms labeled with black text are federal contractors, whereas firms in gray are not.}
\end{figure}

\begin{figure}[ht!]
    \centering
    \caption{Square-weighted loss with industry effects: Posterior means and grades}
   \includegraphics[width=\textwidth]{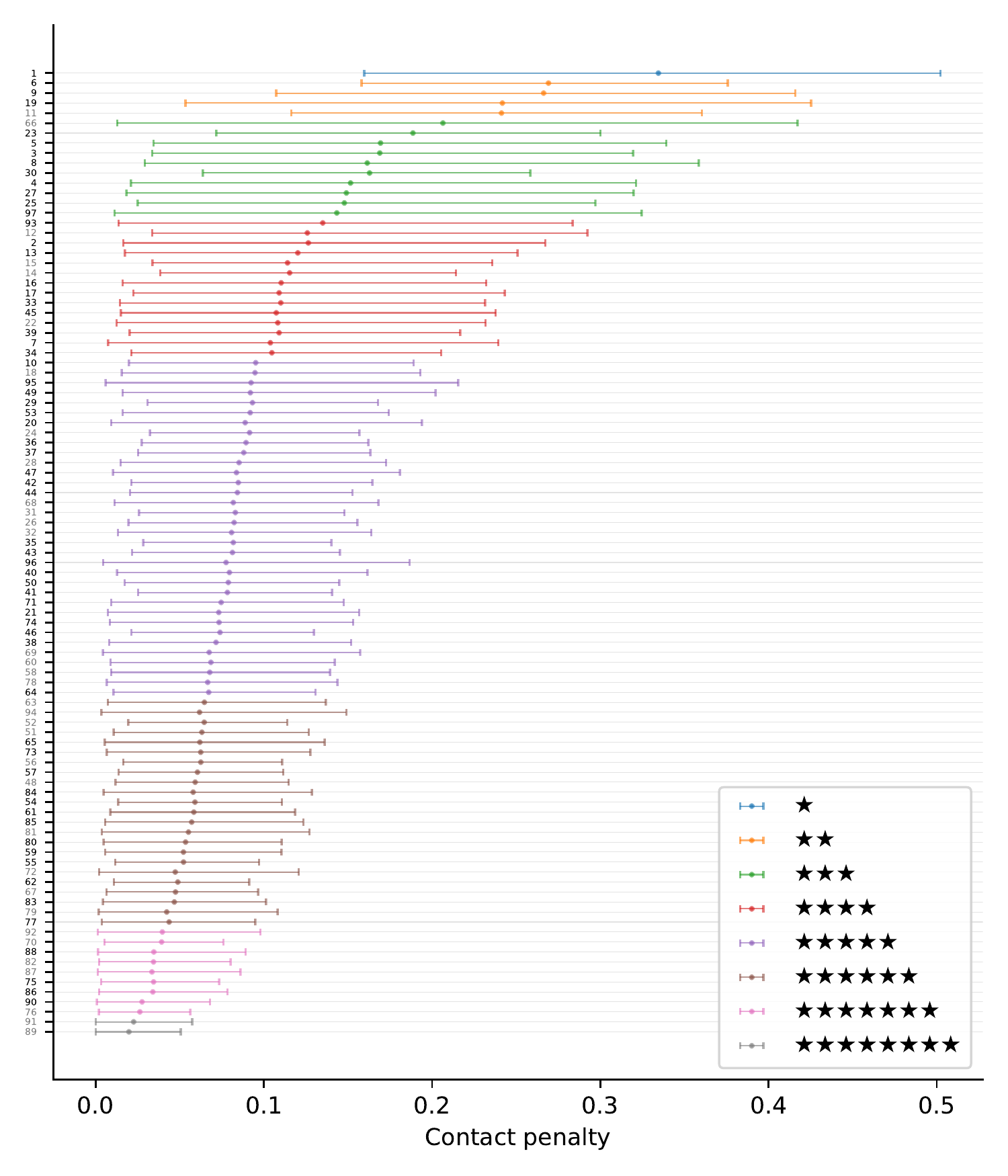}
    \label{fig:poisson_squared_cov}
\parbox{\textwidth}{\small
\vspace{1eX}\emph{Notes}: This figure shows posterior mean proportional contact penalties, 95\% credible intervals, and assigned grades from the model with industry effects under square-weighted loss. Results are shown for $\lambda = 0.25$. Firms are ordered by their rank under $\lambda = 1$, when each firm is assigned its own grade, and labeled by the raw proportional contact gap rank, with \#1 showing the largest gap in favor of white applicants. Firms labeled with black text are federal contractors, whereas firms in gray are not.}
\end{figure}

\begin{figure}[!htbp]
    \centering
    \caption{DR for square-weighted loss}
    \begin{tabular}{c}
 a) Binary \\
\includegraphics[width=0.8\textwidth]{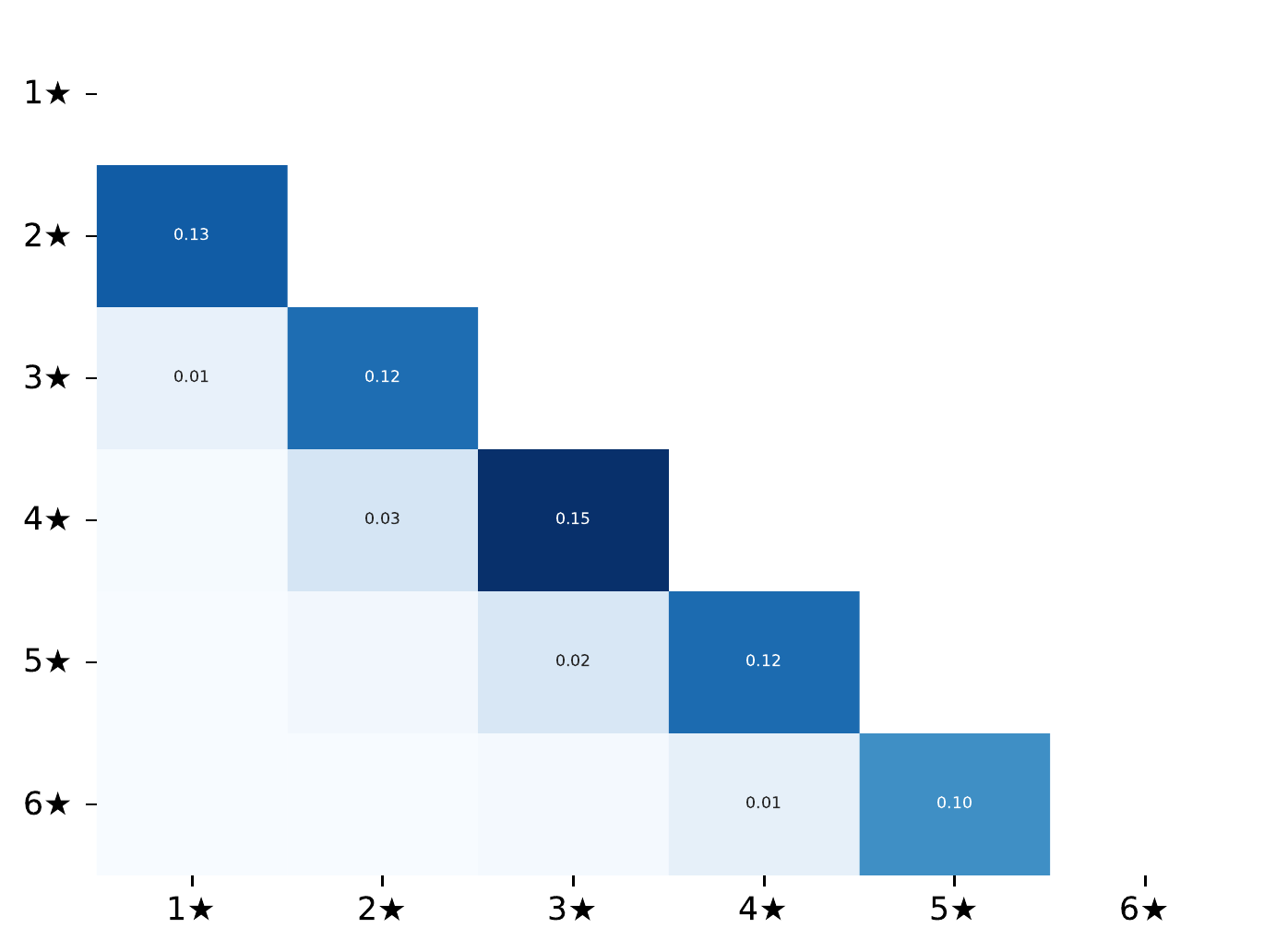} \\
b) Industry effects \\
\includegraphics[width=0.8\textwidth]{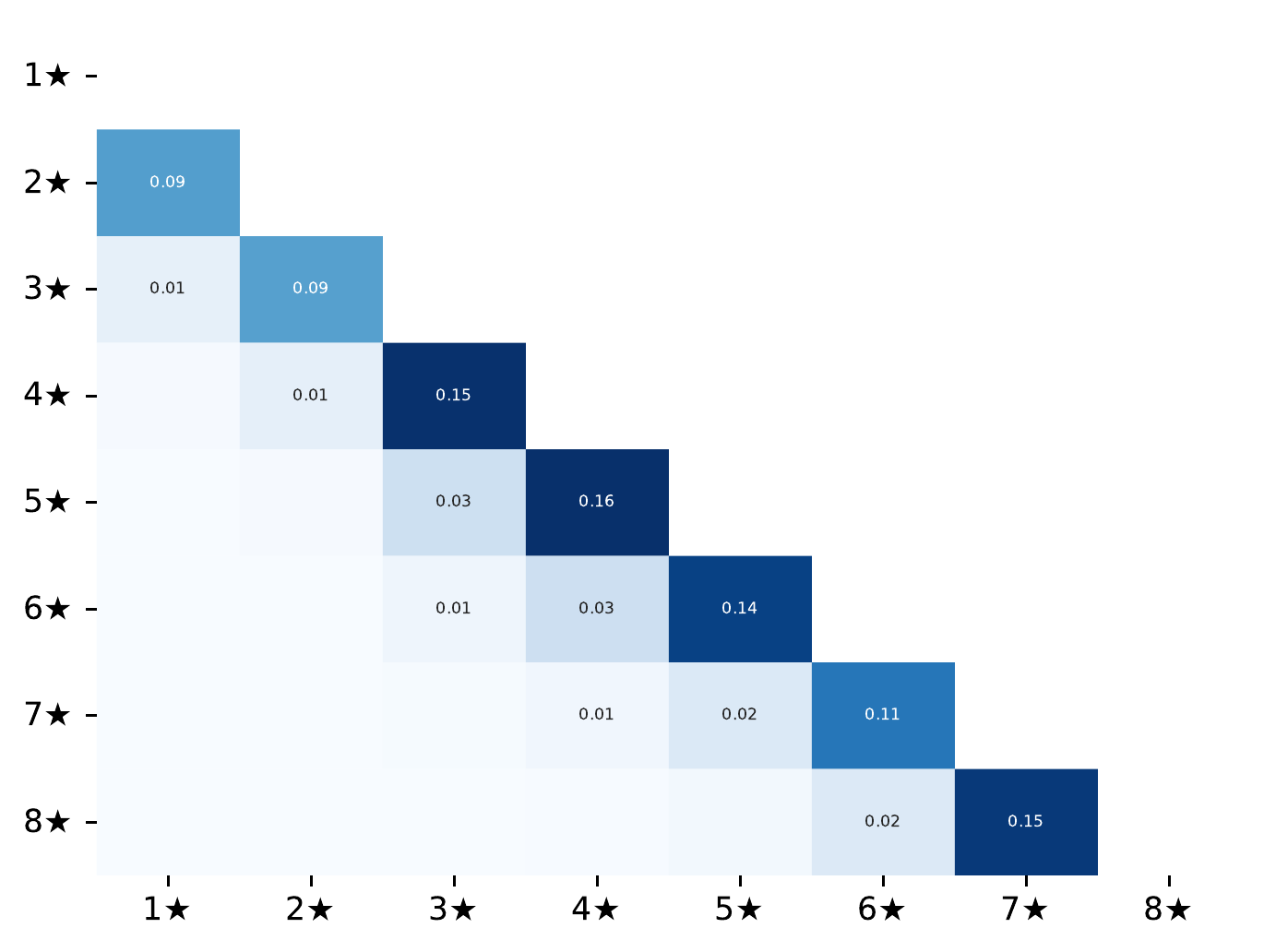}
\end{tabular}
    \label{fig:DR_matri_industry}
\parbox{\textwidth}{\small
\vspace{1eX}\emph{Notes}: This figure shows mean Discordance Rates (DR) across grade pairs under square weighted loss. Panel a uses the baseline model, while panel b uses the model with industry effects. Only cells where DR is above 0.01 are annotated. In both panels, DR decays quickly when comparing non-adjacent grades.}
\end{figure}

\begin{figure}
\centering
\caption{Binary loss: All grades}
    \label{fig: binary_all_grades}
\begin{adjustbox}{center}
\includegraphics[width=\textwidth]{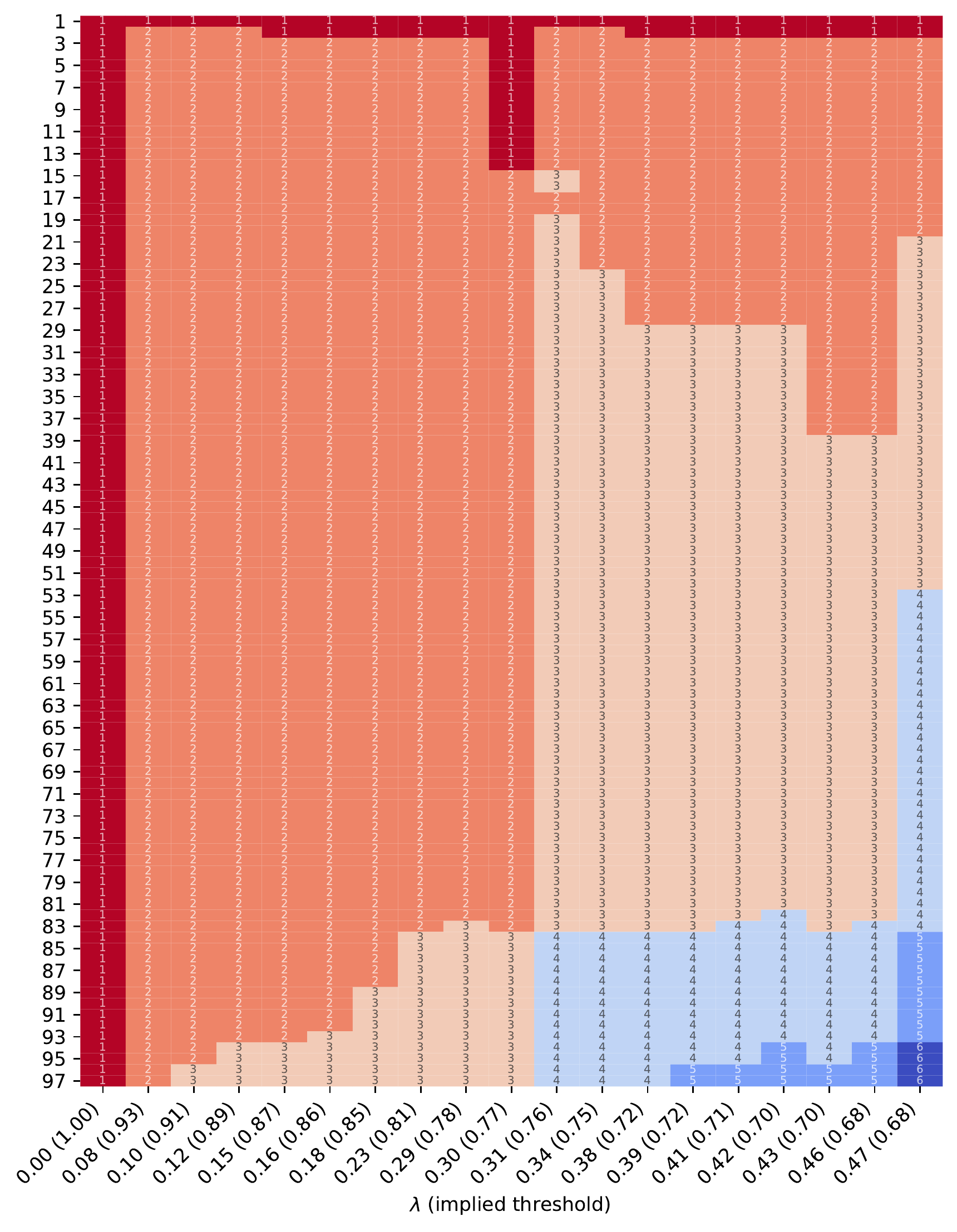}
\end{adjustbox}
\parbox{\textwidth}{\small
\vspace{1eX}\emph{Notes}: This figure shows grade assignments for each value of $\lambda <= 0.5$. To increase readability, only the smallest lambda that yields each unique set of grades is retained. The x-axis reports this $\lambda$ and the value of $1/(1+\lambda)$, which is the implied posterior threshold for pairwise ranking decisions. Firms are ordered by their rank under $\lambda=1$, when each firm is assigned
its own grade.}
\end{figure}

\begin{figure}
\centering
\caption{Binary loss with industry effects: All grades}
    \label{fig: binary_irfe_all_grades}
\begin{adjustbox}{center}
\includegraphics[width=\textwidth]{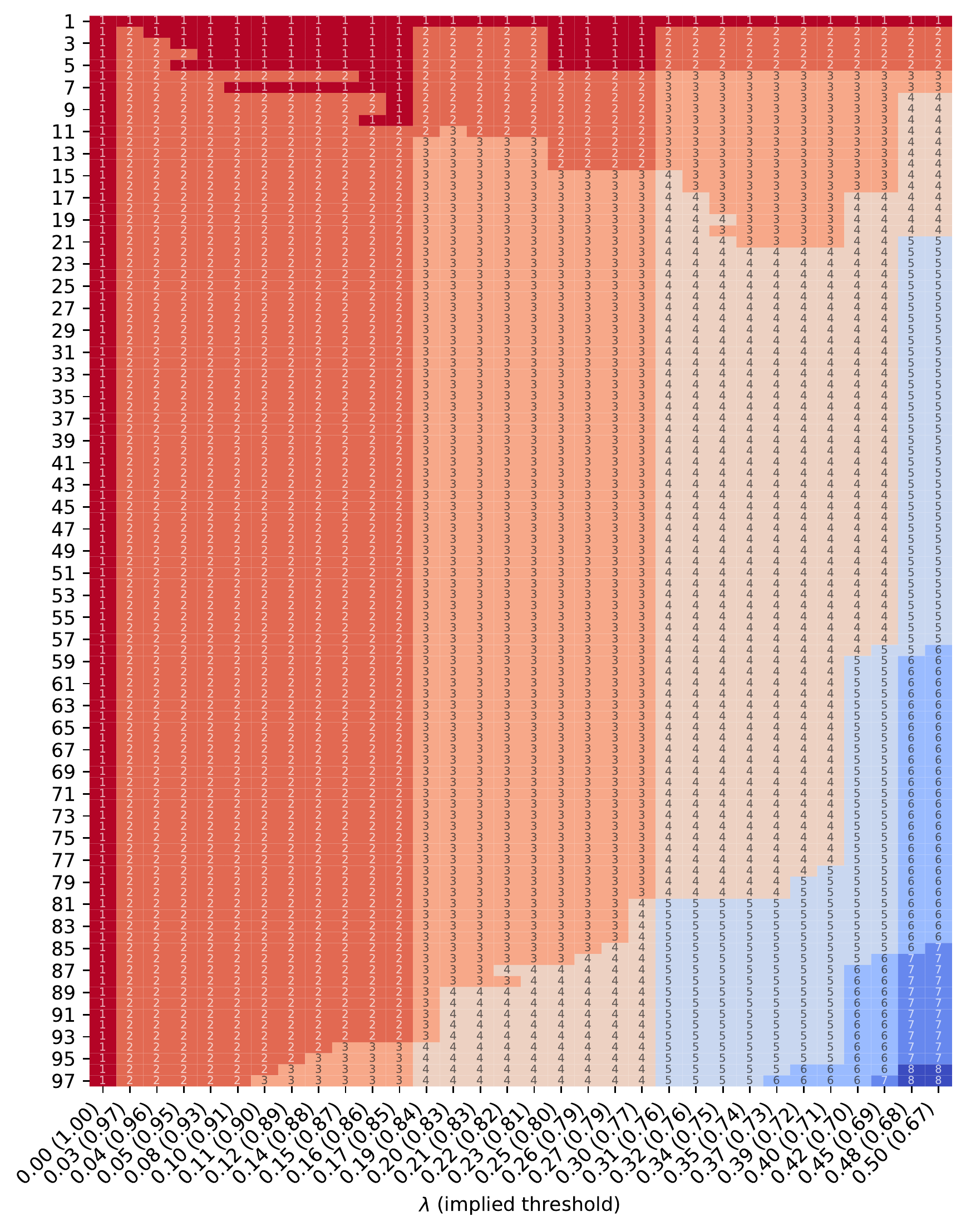}
\end{adjustbox}
\parbox{\textwidth}{\small
\vspace{1eX}\emph{Notes}: This figure shows grade assignments for each value of $\lambda \leq 0.5$. To increase readability, only the smallest lambda that yields each unique set of grades is retained. The x-axis reports this $\lambda$ and the value of $1/(1+\lambda)$, which is the implied posterior threshold for pairwise ranking decisions. Firms are ordered by their rank under $\lambda=1$, when each firm is assigned
its own grade.}
\end{figure}

\begin{figure}
\centering
\caption{Square-weighted loss: All grades}
    \label{fig: sqwt_all_grades}
\begin{adjustbox}{center}
\includegraphics[width=\textwidth]{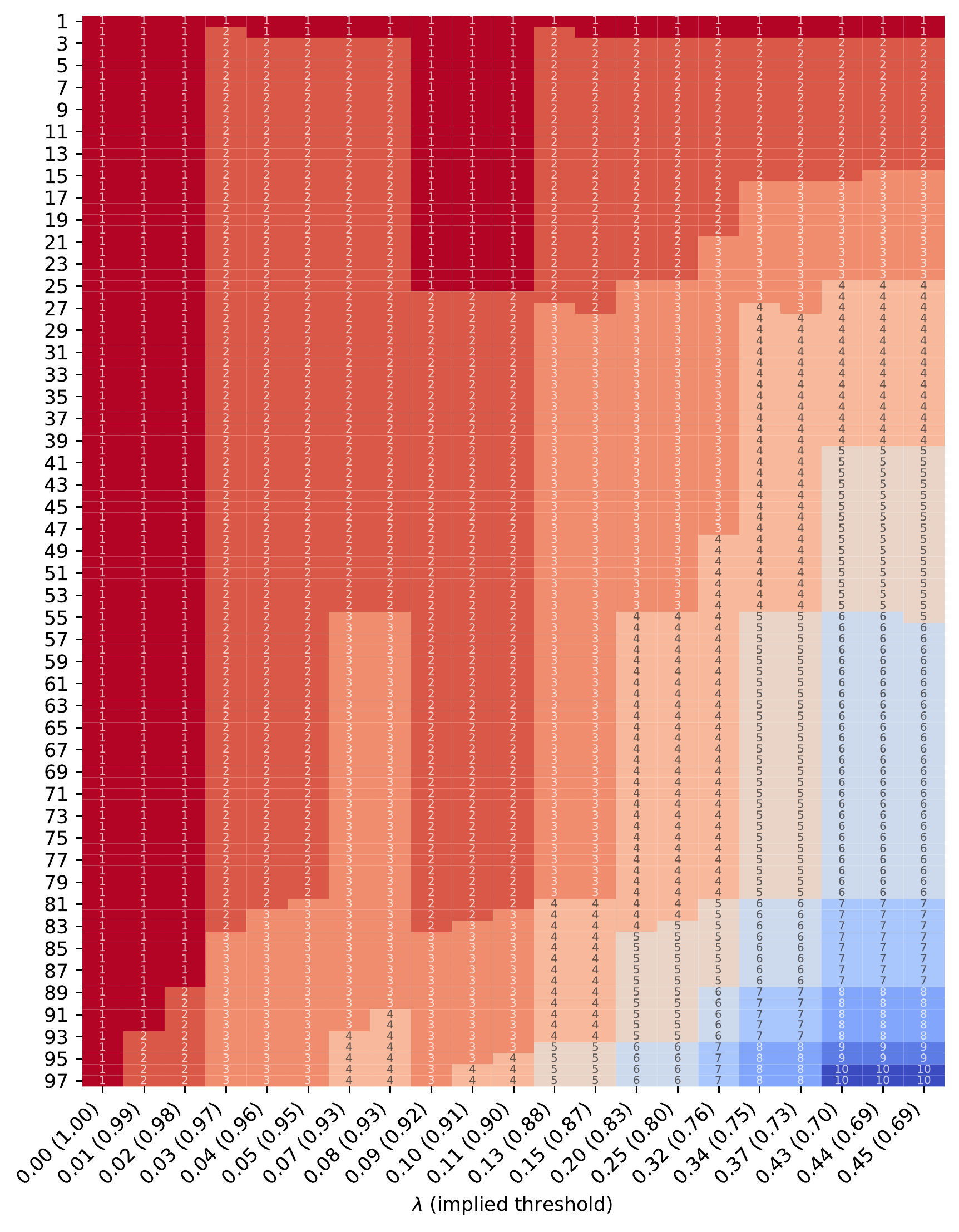}
\end{adjustbox}
\parbox{\textwidth}{\small
\vspace{1eX}\emph{Notes}: This figure shows grade assignments for each value of $\lambda \leq 0.5$. To increase readability, only the smallest lambda that yields each unique set of grades is retained. The x-axis reports this $\lambda$ and the value of $1/(1+\lambda)$, which is the implied posterior threshold for pairwise ranking decisions. Firms are ordered by their rank under $\lambda=1$, when each firm is assigned
its own grade.}
\end{figure}

\begin{figure}
\centering
\caption{Square-weighted loss with industry effects: All grades}
    \label{fig: sqwt_irfe_all_grades}
\begin{adjustbox}{center}
\includegraphics[width=\textwidth]{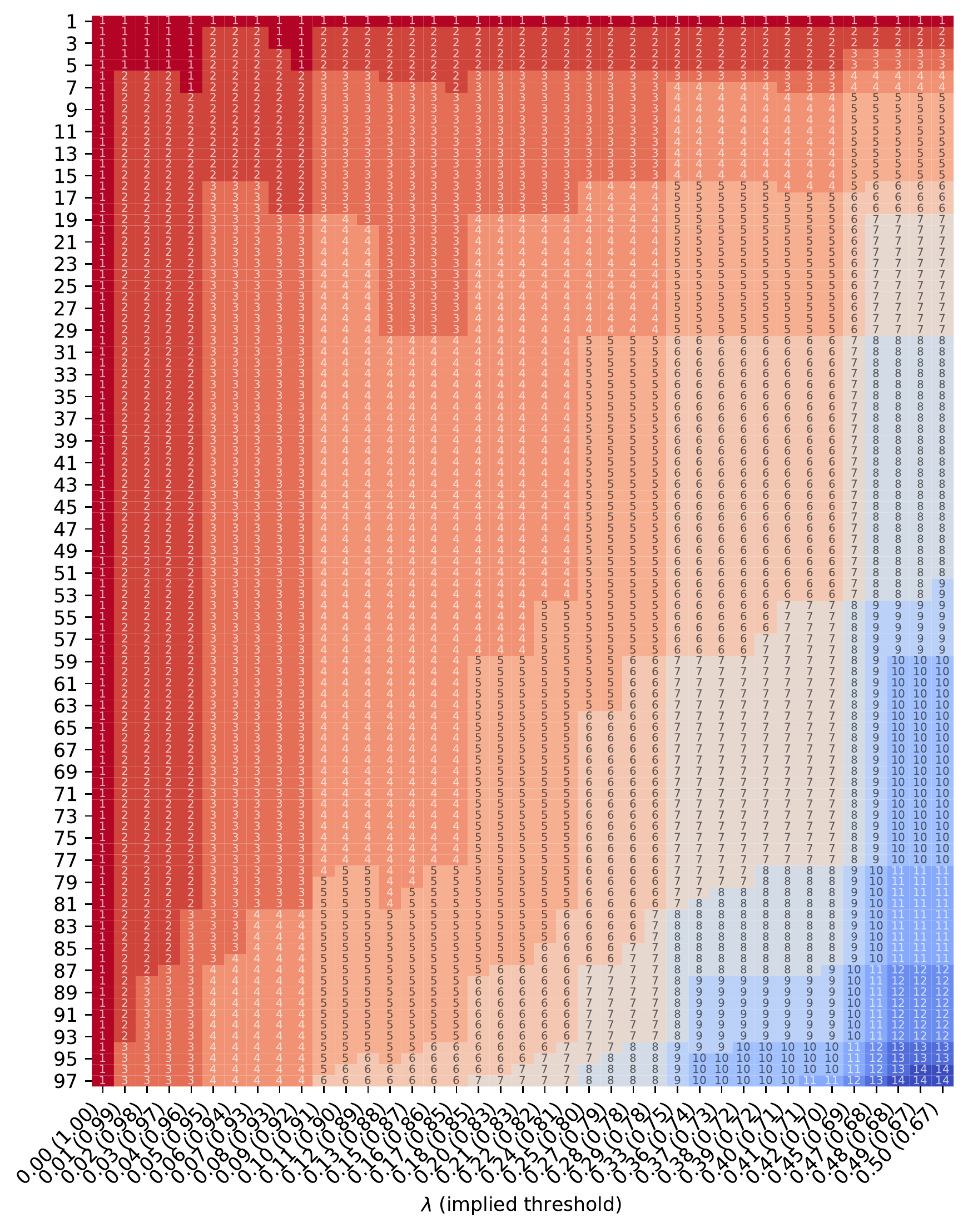}
\end{adjustbox}
\parbox{\textwidth}{\small
\vspace{1eX}\emph{Notes}: This figure shows grade assignments for each value of $\lambda \leq 0.5$. To increase readability, only the smallest lambda that yields each unique set of grades is retained. The x-axis reports this $\lambda$ and the value of $1/(1+\lambda)$, which is the implied posterior threshold for pairwise ranking decisions. Firms are ordered by their rank under $\lambda=1$, when each firm is assigned
its own grade.}
\end{figure}

\end{document}